\documentclass[11pt]{article}
\usepackage{times}
\usepackage{amssymb}
\usepackage{amsmath}
\usepackage{xspace,multirow}
\usepackage[shortlabels]{enumitem}
\usepackage{calc}
\usepackage{amsthm}
\usepackage{color}
\usepackage{comment}
\usepackage{mathtools}
\usepackage{multirow,makecell}

\newcommand\ocirc[1]{\ThisStyle{\ensurestackMath{%
  \stackon[1pt]{\SavedStyle#1}{\SavedStyle\kern.6\LMpt\circ}}}}

\newenvironment{proofof}[1]{\begin{proof}[Proof of {#1}]}{\end{proof}}

\newtheorem{theorem}{Theorem}[section]
\newtheorem{lemma}[theorem]{Lemma}
\newtheorem{claim}[theorem]{Claim}

\newtheorem{corollary}[theorem]{Corollary}

\newtheorem{inftheorem}{Theorem}
\newtheorem{infcor}{Corollary}

{\theoremstyle{definition} \newtheorem{definition}[theorem]{Definition} \newtheorem{defn}[theorem]{Definition}}

\newcommand{\np}{{\em NP}\xspace}
\newcommand{\nphard}{\np-hard\xspace} 
\newcommand{\npcomplete}{\np-complete\xspace}
\newcommand{\apx}{{\em APX}\xspace}

\DeclareMathOperator{\argmin}{argmin}
\DeclareMathOperator{\argmax}{argmax}
\DeclareMathOperator*{\Exp}{E}
\newcommand{\E}[2][{}]{\ensuremath{{\textstyle \Exp_{#1}}\bigl[#2\bigr]}}

\newcommand{\A}{\ensuremath{\mathcal{A}}}
\newcommand{\bA}{\ensuremath{\overline{A}}}

\newcommand{\short}{\ensuremath{\mathrm{sh}}}
\newcommand{\lng}{\ensuremath{\mathrm{lg}}}

\newcommand{\gbar}{\ensuremath{\bar{g}}}
\newcommand{\gdisc}{\ensuremath{g^{\mathrm{int}}}}
\newcommand{\gdiscmon}{\ensuremath{g^{\mathrm{int,mon}}}}
\newcommand{\pcentral}{\ensuremath{\mathring{p}}} 
\newcommand{\pcent}{\pcentral}

\newcommand{\ld}{\ensuremath{\lambda}}
\newcommand{\Ld}{\ensuremath{\Lambda}}

\newcommand{\Rplus}{\ensuremath{\mathbb{R}_+}}
\newcommand{\T}{\ensuremath{\intercal}}

\newcommand{\kinf}{\ensuremath{t}}

\newcommand{\R}{\ensuremath{\mathbb R}}
\newcommand{\Z}{\ensuremath{\mathbb Z}}

\newcommand{\Ebi}{\ensuremath{\overleftrightarrow{E}}}
\newcommand{\deltain}{\ensuremath{\delta^{\text{in}}}}
\newcommand{\deltaout}{\ensuremath{\delta^{\text{out}}}}
\newcommand{\fIAv}{\ensuremath{f^{\one, A, v}}}
\newcommand{\fIIAv}{\ensuremath{f^{\two, A, v}}}

\newcommand{\B}{\ensuremath{\mathcal{B}}}
\newcommand{\C}{\ensuremath{\mathcal{C}}}
\newcommand{\G}{\ensuremath{\mathcal{G}}}

\newcommand{\I}{\ensuremath{\mathcal I}}
\newcommand{\F}{\ensuremath{\mathcal F}}
\newcommand{\D}{\ensuremath{\mathcal D}}

\newcommand{\K}{\ensuremath{\mathcal{K}}}

\newcommand{\Nc}{\ensuremath{\mathcal N}}

\newcommand{\Sc}{\ensuremath{\mathcal S}}
\newcommand{\Pc}{\ensuremath{\mathcal P}}
\newcommand{\Qc}{\ensuremath{\mathcal Q}}
\newcommand{\OPT}{\ensuremath{\mathit{OPT}}}

\newcommand{\sm}{\ensuremath{\setminus}}
\newcommand{\es}{\ensuremath{\emptyset}}
\newcommand{\assign}{\ensuremath{\leftarrow}}

\newcommand{\ceil}[1]{\ensuremath{\left\lceil#1\right\rceil}}

\newcommand{\poly}{\operatorname{\mathsf{poly}}}

\newcommand{\e}{\ensuremath{\epsilon}}

\newcommand{\gm}{\ensuremath{\gamma}}

\newcommand{\sse}{\subseteq}

\newcommand{\iopt}{\ensuremath{O^*}}
\newcommand{\bopt}{\ensuremath{\overline{O}}}

\newcommand{\tX}{\ensuremath{\widetilde{X}}}
\newcommand{\tx}{\ensuremath{\widetilde{x}}}

\newcommand{\tz}{\ensuremath{\widetilde{z}}}
\newcommand{\tf}{\ensuremath{\widetilde{f}}}

\newcommand{\tq}{\ensuremath{\tilde q}}
\newcommand{\tc}{\ensuremath{\tilde c}}
\newcommand{\tg}{\ensuremath{\tilde g}}

\newcommand{\hd}{\ensuremath{\widehat{d}}}
\newcommand{\hx}{\ensuremath{\widehat{x}}}

\newcommand{\hy}{\ensuremath{\widehat{y}}}

\newcommand{\tK}{\widetilde{K}}

\newcommand{\barf}{\ensuremath{\overline f}}
\newcommand{\bx}{\ensuremath{\bar x}}
\newcommand{\by}{\ensuremath{\bar y}}

\newcommand{\bp}{\ensuremath{\bar p}}

\newcommand{\tA}{\ensuremath{\tilde {\mathcal A}}}
\newcommand{\sx}{\ensuremath{x^*}}

\newcommand{\kp}{\ensuremath{\eta}}
\newcommand{\al}{\ensuremath{\alpha}}
\newcommand{\tht}{\ensuremath{\theta}}

\newcommand{\dt}{\ensuremath{\delta}}
\newcommand{\Dt}{\ensuremath{\Delta}}

\newcommand{\w}{\ensuremath{\omega}}
\newcommand{\bo}{\ensuremath{0}} 

\newcommand{\ve}{\ensuremath{\varepsilon}}
\newcommand{\vro}{\ensuremath{\varrho}}

\newcommand{\LW}{\ensuremath{L_{\mathrm{W}}}}

\newcommand*{\rom}[1]{\expandafter\@slowromancap\romannumeral #1@}

\newcommand{\dr}{DR\xspace}
\newcommand{\dcol}{\ensuremath{\D}}
\newcommand{\pdist}{\ensuremath{L}}

\newcommand{\drssc}{\ensuremath{\mathsf{DRSSC}}\xspace}

\newcommand{\drsfl}{\ensuremath{\mathsf{DRSFL}}\xspace}
\newcommand{\drst}{\ensuremath{\mathsf{DRSST}}\xspace}
\newcommand{\drstmon}{\ensuremath{\mathsf{M}\drst}}
\newcommand{\one}{\ensuremath{\mathrm{I}}}
\newcommand{\two}{\ensuremath{\mathrm{II}}}
\newcommand{\dist}{\ensuremath{w}}

\newcommand{\scm}{\ensuremath{\ell}}
\newcommand{\scmdisc}{\scm^{\mathsf{dis}}\xspace}
\newcommand{\scmmax}{\ensuremath{\scm_{\max}}}
\newcommand{\scmmin}{\ensuremath{\scm_{\min}}}
\newcommand{\sgr}{\ensuremath{d}}
\newcommand{\hsgr}{\ensuremath{\widehat{\sgr}}}
\newcommand{\tsgr}{\ensuremath{\widetilde{\sgr}}}
\newcommand{\asym}{\ensuremath{\mathsf{asym}}}
\newcommand{\inpsize}{\ensuremath{\I}}
\newcommand{\Qpcentral}[1][{\pcent}]{Q\ensuremath{{}^{\mathrm{fr}}_{#1}}}
\newcommand{\Qdisc}[1][{\pcent}]{Q\ensuremath{{}_{#1}}}
\newcommand{\bQdisc}[1][{\pcent}]{\ensuremath{\overline{\mathrm{Q}}_{#1}}}
\newcommand{\Rdisc}[1][{\pcent}]{R\ensuremath{{}_{#1}}}
\newcommand{\bRdisc}[1][{\pcent}]{\ensuremath{\overline{\mathrm{R}}_{#1}}}
\newcommand{\fp}[2][{p}]{\ensuremath{f({#1};{#2})}}
\newcommand{\hpcentral}[2][{\pcent}]{\ensuremath{h({#1}\,;{#2})}}
\newcommand{\bhpcentral}[2][{\pcent}]{\ensuremath{\overline{h}({#1}\,;{#2})}}
\newcommand{\thpcentral}[2][{\pcent}]{\ensuremath{\widetilde{h}({#1}\,;{#2})}}
\newcommand{\zpcentral}[2][{\pcent}]{\ensuremath{z({#1}\,;{#2})}}
\newcommand{\zpcentshort}[2][{\pcent}]{\ensuremath{z^{\short}({#1}\,;{#2})}}
\newcommand{\zpcentlng}[2][{\pcent}]{\ensuremath{z^{\lng}({#1}\,;{#2})}}
\newcommand{\hzshort}[1][{}]{\ensuremath{\widehat{z}^{\short,{#1}}}}
\newcommand{\hzlong}[1][{}]{\ensuremath{\widehat{z}^{\lng,{#1}}}}

\newcommand{\htpcent}{\ensuremath{\widehat{p}}}
\newcommand{\htp}{\ensuremath{\widehat{p}}\xspace}

\newcommand{\ppoly}{\htp}
\newcommand{\Qdiscpoly}{\Qdisc[\ppoly]} 
\newcommand{\Qppoly}{\Qpcentral[\ppoly]}

\newcommand{\hppoly}[1]{\hpcentral[\ppoly]{#1}}
\newcommand{\zppoly}[1]{\zpcentral[\ppoly]{#1}}
\newcommand{\zpolylp}[2][\ppoly]{T\ensuremath{{}_{{#1},{#2}}}}
\newcommand{\Asup}{\ensuremath{\A^{\mathrm{sup}}}}

\newcommand{\hh}{\ensuremath{\widehat{h}}}
\newcommand{\bh}{\ensuremath{\overline{h}}}
\newcommand{\tdh}{\ensuremath{\widetilde{h}}}
\newcommand{\bg}{\ensuremath{\overline{g}}}
\newcommand{\lb}{\ensuremath{\mathsf{LB}}}
\newcommand{\lbfl}{\ensuremath{\psi}}
\newcommand{\ub}{\ensuremath{\mathsf{UB}}}

\newcommand{\Eh}[2][A]{\ensuremath{{\textstyle \Exp^h_{#1}}\bigl[#2\bigr]}}
\newcommand{\El}[2][A]{\ensuremath{{\textstyle \Exp^l_{#1}}\bigl[#2\bigr]}}
\newcommand{\hEh}[2][i]{\ensuremath{\widehat{\Exp}^{{#1},h}_{A}\bigl[#2\bigr]}}
\newcommand{\hEl}[2][i]{\ensuremath{\widehat{\Exp}^{{#1},l}_{A}\bigl[#2\bigr]}}

\newcommand{\alg}{\ensuremath{\mathsf{Alg}}}
\newcommand{\maxmin}{$\max$-$\min$\xspace}
\newcommand{\kmaxmin}{$k$-$\max$-$\min$\xspace}
\newcommand{\polyopt}{\ensuremath{\iopt_{\ppoly}}}
\newcommand{\vol}{\ensuremath{\mathsf{vol}}}

\newcommand{\polyalg}{\ensuremath{\mathsf{PolyAlg}}\xspace}
{\theoremstyle{remark} \newtheorem*{ppolyalg}{Algorithm $\polyalg(\kp)$}}
\newcommand{\optp}{\OPT_{\Pi}}

\newcommand{\gxcp}[1][\scmbound]{\ensuremath{\Phi(x,{#1},A)}\xspace}
\newcommand{\scmbound}{\mu}
\newcommand{\SCM}{\ensuremath{\mathcal{L}}}
\newcommand{\SCMbar}{\ensuremath{\overline{\mathcal{L}}}}
\newcommand{\allsets}{unrestricted\xspace}
\newcommand{\kbounded}{\ensuremath{k}-bounded\xspace}
\newcommand{\J}{\ensuremath{S}}
\newcommand{\mustar}{\ensuremath{\mu^*}}

\newcommand{\jbar}{\ensuremath{\overline{j}}}
\newcommand{\tm}{\ensuremath{t}}

\newcommand{\veprime}{\ve}
\newcommand{\hproxy}[2][{\pcent}]{\ensuremath{h^{\mathrm{pr}}({#1}\,;{#2})}}

\newcommand{\Atilde}{\ensuremath{\widetilde{\A}}}
\newcommand{\Afreq}{\ensuremath{\mathcal{A}^{\mathrm{freq}}}}
\newcommand{\Ahatfreq}{\ensuremath{\widehat{\mathcal{A}}^{\mathrm{freq}}}}
\newcommand{\Ahatrare}{\ensuremath{\widehat{\mathcal{A}}^{\mathrm{rare}}}}
\newcommand{\Arare}{\ensuremath{\mathcal{A}^{\mathrm{rare}}}}

\newcommand{\pblocked}{\ensuremath{\overline{p}}}
\newcommand{\pfree}{\ensuremath{\widetilde{p}}}
\newcommand{\Pfree}{\ensuremath{P^{\mathrm{free}}}}
\newcommand{\Phatfree}{\ensuremath{\widehat{P}^{\mathrm{free}}}}
\newcommand{\Qfree}{\ensuremath{Q^{\mathrm{free}}}}

\newcommand{\sgrest}{\ensuremath{\sgr^{\mathrm{est}}}}
\newcommand{\qtilde}{\ensuremath{\widetilde{q}}}
\newcommand{\qx}{\ensuremath{q^*}}

\newcommand{\phat}{\ensuremath{\widehat{p}}}
\newcommand{\maxinprox}[1][x]{K\ensuremath{{}_{#1}}\xspace}
\newcommand{\qstar}{\ensuremath{q^*}}

\newcommand{\stocp}{\textnormal{(2St-P)}\xspace}
\newcommand{\bbeta}{\ensuremath{\overline\beta}}
\newcommand{\Sbx}{\ensuremath{S_{\bx}}}
\newcommand{\new}{\ensuremath{\mathrm{new}}}

\makeatletter
\newcommand{\leqnomode}{\tagsleft@true\let\veqno\@@leqno}
\newcommand{\reqnomode}{\tagsleft@false\let\veqno\@@eqno}
\makeatother

\title{Approximation Algorithms for Distributionally Robust \\ Stochastic Optimization with
  Black-Box Distributions%
\thanks{A preliminary version~\cite{LinharesS19} appeared in the Proceedings of the
  51st ACM Symposium on Theory of Computing (STOC), 2019.}} 
\author{Andr\'e Linhares\thanks{{\tt \{alinhare,cswamy\}@uwaterloo.ca}.
        Dept. of Combinatorics and Optimization, University of Waterloo, Waterloo, ON N2L 3G1.
        Supported in part by NSERC grant 327620-09 and an NSERC Discovery Accelerator
        Supplement award.}
\and
\addtocounter{footnote}{-1}
         Chaitanya Swamy\footnotemark
}

\date{}

\begin{document}

\maketitle
\def\thepage{}
\thispagestyle{empty}

\begin{abstract}
Two-stage stochastic optimization is a widely used framework for modeling uncertainty,
where we have a probability distribution over possible realizations of the data, called
scenarios, and decisions are taken in two stages: we make first-stage decisions knowing
only the underlying distribution and before a scenario is realized, and may take
additional second-stage recourse actions after a scenario is realized. The goal is
typically to minimize the total expected cost. A common criticism levied at this model is
that the underlying probability distribution is itself often imprecise! To address this,
an approach that is quite versatile and has gained popularity in the
stochastic-optimization literature is the {\em distributionally robust 2-stage model}: given a
collection $\dcol$ of probability distributions, our goal now is to minimize the maximum
expected total cost with respect to a distribution in $\dcol$. 

There has been almost no prior work
however on developing approximation algorithms for distributionally robust problems, when
the underlying scenario-set is discrete, as is the case with discrete-optimization
problems. We provide a framework for designing approximation algorithms in such settings
when the collection $\dcol$ is a ball around a central distribution and the central
distribution is accessed {\em only via a sampling black box}. 

We first show that one
can utilize the {\em sample average approximation} (SAA) method---solve the
distributionally robust problem with an empirical estimate of the central distribution---%
to reduce the problem to the case where the central distribution has
{\em polynomial-size} support. This follows because we argue that a distributionally robust
problem can be reduced in a novel way 
to a standard 2-stage problem with bounded inflation factor, which enables one to use the
SAA machinery developed for 2-stage problems. 
Complementing this, we show how to approximately solve a fractional relaxation of the
SAA (i.e., polynomial-scenario central-distribution) problem. Unlike in 2-stage stochastic-
or robust- optimization, this turns out to be quite challenging. We utilize the ellipsoid
method in conjunction with several new ideas to show that this problem can be
approximately solved provided that we have an (approximation) algorithm 
for a certain max-min problem that is akin to, and generalizes, the \kmaxmin problem---%
find the worst-case scenario consisting of at most $k$ elements---encountered in 2-stage
robust optimization. We obtain such a 
procedure for various discrete-optimization problems; by complementing this via  
LP-rounding algorithms that provide {\em local} (i.e., per-scenario)
approximation guarantees, 
we obtain the {\em first} approximation algorithms for the distributionally robust versions
of a variety of discrete-optimization problems including set cover, vertex cover, edge cover, 
facility location, and Steiner tree,
with guarantees that are, except for set cover, within $O(1)$-factors of the
guarantees known for the deterministic version of the problem.  
\end{abstract}

\newpage
\pagenumbering{arabic}
\normalsize

\section{Introduction}
Stochastic-optimization models capture uncertainty by modeling it via a probability
distribution over a collection $\A$ of possible realizations of the data, called 
{\em scenarios}.  
An important and widely used model is the {\em 2-stage recourse model}, where one seeks to
take actions both before and after the data has been realized (stages I and II) so as to
minimize the expected total cost incurred. Many applications come under this setting. An
oft-cited prototypical example is {\em 2-stage stochastic facility location}, wherein one
needs to decide where to set up facilities to serve clients. 
The client-demand pattern is uncertain, but one does have some statistical
information about the demands. One can open some facilities initially, given only the
distributional information about demands; after a specific demand pattern is realized
(according to this distribution), one can take additional recourse actions such as opening
more facilities incurring their recourse costs. The recourse costs are usually higher
than the first-stage costs, as they may entail making decisions in rapid reaction to
the observed scenario (e.g., deploying resources with smaller lead time). 

An issue with the above 2-stage model, which is a common source of criticism, 
is that the distribution modeling the uncertainty is itself often imprecise!
Usually, one models the distribution to be statistically
consistent with some historical data, so we really have a {\em collection} of
distributions, and a more robust approach is to 
{\em hedge against the worst possible distribution}. This gives rise to the
{\em distributionally robust 2-stage model}: 
the setup is similar to that of the 2-stage model, but we now have a collection $\dcol$ of
probability distributions; our goal is to minimize the maximum expected total cost with
respect to a distribution in $\dcol$. Formally, if $X\sse\Rplus^m$ is the set of
first-stage actions and the cost associated with $x\in X$ is $c^\T x$, 
we want to solve the following problem:

\medskip
\noindent \hspace*{-2ex}
\begin{minipage}[t]{0.45\textwidth}
\vspace*{-5ex}
\begin{equation}
\min_{x\in X}\quad c^\T x+\max_{q\in\dcol}\E[A\sim q]{g(x,A)} \tag{DRO} \label{drprob} 
\end{equation}
\end{minipage}
\qquad
where \ \ $g(x,A):=\min_{\text{second-stage actions }z^A}\bigl(\text{cost of }z^A\bigr)$.
\smallskip

Distributionally robust (\dr) stochastic optimization 
is a versatile approach dating back
to~\cite{Scarf58} that has (re)gained interest 
recently in the Operations Research literature, where it is 
sometimes called {\em data-driven} or {\em ambiguous} stochastic optimization 
(see, e.g.,~\cite{GaoK,BertsimasSZ,VanparysEK,EsfahaniK17} and their references). 
The \dr 2-stage model 
also serves to 
nicely interpolate between the extremes of:  
(a) 2-stage stochastic optimization, which optimistically assumes that one knows
the underlying distribution $p$ precisely (i.e., $\dcol=\{p\}$); and  
(b) 2-stage robust optimization, which abandons the distributional view and seeks to
minimize the maximum cost incurred in a scenario, 
thereby adopting the overly
cautious approach of being robust against {\em every} possible scenario regardless of how 
{\em likely} it is for a scenario to materialize;
this can be captured by letting $\dcol=\{\text{all distributions over $\A$}\}$, where $\A$ is the scenario-collection
in the 2-stage robust problem. 
Both extremes can lead to suboptimal decisions: with stochastic optimization, 
the optimal solution for a specific distribution $p$ could be quite suboptimal even for a
``nearby'' distribution $q$;%
\footnote{There are examples where $\|q-p\|_1\leq\ve$ but an optimal
solution for $p$ can be arbitrarily bad when evaluated under $q$.}
with robust optimization, the presence of a single scenario, however unlikely,
\nolinebreak
\mbox{may force certain decisions that are undesirable for all other scenarios.}

Despite its modeling benefits and popularity, 
to our knowledge, there has been almost no prior work on 
developing approximation algorithms for 
{\em \dr 2-stage discrete-optimization}, and, more generally, for \dr 2-stage problems
with a {\em discrete} underlying scenario set (as is the case in discrete optimization).
(The exception is~\cite{AgarwalDSY10}, which we discuss in Section~\ref{relwork}.%
\footnote{Peripherally related is~\cite{WuDX15}, who consider a version of \dr facility location,
where the uncertainty only influences the costs and not the constraints, which yields a
much-simpler and more restrictive model.})

\vspace*{-0.5ex}
\subsection{Our contributions} \label{contr}

\vspace*{-0.5ex}
We initiate a systematic study of distributionally robust discrete 2-stage problems from
the perspective of approximation algorithms. 
We develop a {\em general framework} for
designing approximation algorithms for these problems, when the collection $\dcol$ is a
ball around a central distribution $\pcent$ in the $L_\infty$ metric, $\frac{1}{2}L_1$
metric ({\em total-variation distance}), 
or {\em Wasserstein} metric (defined below). 
(Note that this still allows interpolating between stochastic and robust optimization.)
We make {\em no} assumptions about $\pcent$; it could have exponential-size support,
and our only means of accessing $\pcent$ is via a {\em sampling black box}.%
\footnote{The \dr problem remains challenging even if $\pcent$ has polynomial-size
  support, but $|\A|$ is exponential.} 
We view sampling from the black box as an elementary
operation, so our running time bounds also imply sample-complexity bounds.
Settings where $\dcol$ is a ball in some probability metric arise
naturally when one tries to infer a scenario distribution from observed data 
(see, e.g.~\cite{ErdoganI,EsfahaniK17,ZhaoG15})---%
hence, the moniker data-driven optimization---%
and it has been argued that defining $\dcol$ using the Wasserstein metric has various
benefits~\cite{EsfahaniK17,ZhaoG15,GaoK,VanparysEK}.

{\em We view the frameworks that we develop for \dr discrete 2-stage problems as our
chief contribution, and the techniques that we devise for dealing with 
Wasserstein metrics as the main feature of our work} (see Theorem~\ref{intromainthm} below).
We demonstrate the utility of our frameworks by using them 
to obtain the {\em first} approximation guarantees for the
distributionally robust versions of various discrete-optimization problems such as set 
cover, vertex cover, edge cover, 
facility location, and Steiner tree. 
The guarantees that we obtain 
are, in most cases, within $O(1)$-factors of the guarantees known for the deterministic 
(and 2-stage-\{stochastic, robust\}) counterpart of the problem (see Table~\ref{results}). 

\vspace*{-1.5ex}
\paragraph{Formal model description.}
We study the following {\em distributionally robust 2-stage model}.
We are given an underlying set $\A$ of scenarios, and 
a ball $\dcol=\{q: \pdist(\pcent,q)\leq r\}$ of distributions around a central
distribution $\pcent$ over $\A$
under some metric $\pdist$ on probability distributions. 
We can take {\em first-stage} actions $x\in X\sse\Rplus^m$ 
before a scenario is realized, incurring a first-stage cost $c^\T x$, 
and 
{\em second-stage recourse} actions $z^A$ after a scenario $A\in\A$ is realized;  
the combination
$(x,z^A)$ of first- and second-stage actions for a scenario $A$ must yield a feasible
solution for each scenario $A$. 
Using $A\sim q$ to denote that scenario $A$ is drawn according to distribution $q$,
we want to solve: 
$\min_{x\in X}\ \bigl(c^\T x+\max_{q:\pdist(\pcent,q)\leq r}\E[A\sim q]{\text{cost of $z^A$}}\bigr)$.

We use $\inpsize$ to denote the input size, which always measures the encoding size of
the underlying deterministic problem, along with the first- and second-stage costs and the
radius $r$ of the ball $\dcol$. 
It is standard in the study of 2-stage problems in the CS literature 
to assume that every first-stage action has a corresponding recourse action 
(e.g., facilities may be opened in either stage). 
We use $\ld\geq 1$ to denote an inflation parameter that measures 
the maximum factor by which the cost of a first-stage action increases 
in the second stage.  
We consider the cases where $\pdist$ is the $L_\infty$ metric,
$\|p-q\|_\infty:=\max_{A\in\A}|p_A-q_A|$;  
$\frac{1}{2}L_1$ metric, $\frac{1}{2}\|p-q\|_1:=\frac{1}{2}\sum_{A\in\A}|p_A-q_A|$, which
is the total-variation distance; or a Wasserstein metric. 

To motivate and define the rich class of Wasserstein metrics, 
note that while the choice of $\pdist$ 
is a problem-dependent modeling decision, 
we would like the ball $\dcol$ to contain other ``reasonably similar''
distributions, and 
{\em exclude} completely unrelated distributions, as
the latter could lead to overly-conservative decisions, \`a la robust optimization. 
One way of measuring the similarity between two distributions  
is to see if they 
they spread their probability mass on ``similar'' scenarios. 
Wasserstein metrics capture this viewpoint crisply, and lift an underlying 
{\em scenario metric} $\scm$ to a metric on distributions over scenarios. 
The  {\em Wasserstein distance} between two distributions $p$ and $q$ is the minimal cost 
of moving probability mass to transform $p$ into $q$,
where the cost of moving $\gm_{A,A'}$ mass from scenario $A$ to scenario $A'$ is 
$\gm_{A,A'}\scm(A,A')$. 
(Observe that $\frac{1}{2}L_1$ is the Wasserstein metric with respect to the
\nolinebreak
\mbox{discrete scenario metric: $\scmdisc(A,A')=1$ if $A\neq A'$, and $0$ otherwise.)}  

\medskip \noindent 
{\bf Example: \dr 2-stage facility location (\drsfl)}. 
As a concrete example, consider the \dr 
version of 2-stage facility location. 
We have a metric space
$\bigl(\F\cup\C,\{\dist_{ij}\}_{i,j\in\F\cup\C}\bigr)$, where $\F$ is a set of facilities,
and $\C$ is a set of clients. A scenario is a subset of $\C$ indicating the set of clients
that need to be served in that scenario. (We can model integer demands by creating
co-located clients.) 
We may open a facility $i\in\F$ in stages I or II,
incurring costs of $f_i$ and $f^\two_i$ respectively. 
In scenario $A$, we need to assign every $j\in A$ to a facility $i^A(j)$ opened in stage I
or in scenario $A$; the second-stage cost of scenario $A$ is 
$\sum_{i\text{ opened in scenario }A}f^\two_i+\sum_{j\in A}c_{i^A(j)j}$. 
The goal is to minimize 
$\sum_{i\text{ opened in stage I}}f_i+\max_{q:\pdist(\pcent,q)\leq r}\E[A\sim q]{\text{second-stage cost of }A}$.
Here $\ld:=\max\{1,\max_{i\in\F}f^\two_i/f_i\}$, and
$\inpsize$ is the encoding size of $\bigl(\F,\C,\dist,f,f^\two,r\bigr)$.

We consider two common choices for $\A$: (a) the {\em \allsets} setting: $\A:=2^\C$, 
which is the usual setting in 2-stage stochastic optimization; and 
(b) the {\em \kbounded} setting: $\A=\A_{\leq k}:=\{A\sse\C: |A|\leq k\}$, which is
the usual setup in 2-stage robust optimization for modeling an exponential number of
scenarios~\cite{FeigeJMM05,Khandekar,GuptaNR}.  
These two settings for $\A$ arise for other problems as well (where $\C$ is a suitable
ground set). 

In addition to $\pdist$ being the $L_\infty$ or $\frac{1}{2}L_1$ metrics, we can consider
various ways of defining a scenario metric $\scm$  
in terms of the underlying assignment-cost metric $\dist$ to capture that two scenarios
involving demand locations in the same vicinity are deemed similar; lifting these scenario
metrics to the Wasserstein metric over distributions yields a rich class of
\dr 2-stage facility location models. 
For instance, we can define the {\em asymmetric metric}
$\scm^\asym_\infty(A,A'):=\max_{j'\in A'}\dist(j',A)$, 
where $\dist(j',A):=\min_{j\in A}\dist_{j'j}$, which measures the maximum separation
between clients in $A'$ and locations in $A$ (the resulting Wasserstein metric $\LW$ will
now be an asymmetric metric on distributions). 
(There are other natural scenario metrics: the asymmetric metric
$\scm^\asym_1(A,A'):=\sum_{j'\in A'}\dist(j',A)$, and the symmetrizations of these
asymmetric metrics:

\vspace*{-1ex}
\paragraph{Our results.} 
Our main result pertains to Wasserstein metrics, 
which have a great deal of modeling power. 
Let $\LW$ be the Wasserstein metric with respect to a scenario metric $\scm$. 
To gain mathematical traction, it will be convenient to move to a 
relaxation of the \dr 2-stage problem where we 
allow fractional second-stage decisions. 
Let $g(x,A)$ be the optimal second-stage cost of scenario $A$ given
$x$ as the first-stage actions when we allow {\em fractional} second-stage actions. 
(We will obtain integral second-stage actions by rounding an optimal solution to $g(x,A)$
using an LP-relative $\al$-approximation algorithm for the deterministic problem.) 

We relate the approximability of the \dr problem to that of known tasks in
2-stage-stochastic- and deterministic- optimization, and the following {\em deterministic} 
problem:   

\vspace*{-3.5ex}
\begin{equation*}
g(x,y,A)\ :=\ \max_{A'\in\A} \ g(x,A')-y\cdot\scm(A,A') \qquad
\text{given a first-stage decision $x\in X$, scenario $A\in\A$, $y\geq 0$}. 
\end{equation*}

\vspace*{-1.5ex}
\noindent
Notice that 
$g(x,y,A)$ ties together three distinct  
sources of complexity in the \dr 2-stage problem: the combinatorial complexity of the
underlying optimization problem, captured by $g(x,A')$; the complexity of the
scenario set $\A$; and the complexity of the scenario metric $\scm$, captured by the 
$y\cdot\scm(A,A')$ term. 

\begin{inftheorem}[Combination of Theorems~\ref{mainsaathm} and~\ref{polythm}] 
\label{intromainthm} 
Suppose that we have the following. 
\begin{enumerate}[label=(\arabic*), nosep, labelwidth=\widthof{(3)}, leftmargin=!]
\item \label{gxyapprox} 
A {\em $(\beta_1,\beta_2)$-approximation algorithm for computing $g(x,y,A)$}, which is
an algorithm that given $(x,y,A)\in X\times\R_+\times\A$ returns $\bA\in\A$ such that
$g(x,\bA)-y\cdot\scm(A,\bA)\geq
\max_{A'\in\A}\bigl(\frac{g(x,A')}{\beta_1}-\beta_2\cdot y\cdot\scm(A,A')\bigr)$;

\item \label{lapx} 
A {\em local $\rho$-approximation} algorithm for the underlying 2-stage problem, which
is an algorithm that rounds a fractional first-stage solution to an integral one while
incurring at most a $\rho$-factor blowup in the first-stage cost, and in the cost of each
scenario; and

\item \label{dapx} 
An LP-relative $\al$-approximation algorithm for the underlying deterministic problem.
\end{enumerate} 
Then we can obtain an
$O\bigl(\al\beta_1\beta_2\rho+\ve)$-approximation for the \dr problem in time
$\poly\bigl(\text{input size},\frac{\ld}{\ve}\bigr)$. 
\end{inftheorem} 

Ingredients~\ref{lapx} and~\ref{dapx} can be obtained using
known results for 2-stage-stochastic- and deterministic- optimization; 
ingredient~\ref{gxyapprox} is the new component we need to supply to instantiate
Theorem~\ref{intromainthm} and obtain results for specific \dr 2-stage problems.
(The non-standard notion of approximation for $g(x,y,A)$ is necessary, as
the mixed-sign objective precludes any guarantee under the standard notion of 
approximation; 
see Theorem~\ref{saainapx}.)  
In various settings, 
we show that a $(\beta_1,\beta_1)$-approximation for $g(x,y,A)$ can be
obtained by utilizing results for the simpler \maxmin problem%
---$\max_{A'\in\A}g(x,A')$ (i.e., $g(x,0,A)$)---encountered in 2-stage robust 
optimization (see the proof of Theorem~\ref{gxyalg} in Section~\ref{gxyalgproof}): in
the \kbounded setting, where $\A=\A_{\leq k}$, 
this is called the \kmaxmin problem~\cite{FeigeJMM05,Khandekar,GuptaNR}. 
In particular, this 
applies to the $\frac{1}{2}L_1$-metric, 
as in this case we have $g(x,y,A)=\max\{g(x,A),\max_{A'\in\A}g(x,A')-y\}$.

\begin{infcor} \label{introl1thm}
Consider a \dr 2-stage problem where the Wasserstein metric $\LW$ is the $\frac{1}{2}L_1$
metric.  
Suppose that we have a $\beta$-approximation for the problem $\max_{A'\in\A}g(x,A')$
(given $x\in X$ as input), and
we have ingredients~\ref{lapx} and~\ref{dapx} in Theorem~\ref{intromainthm}.
Then we can obtain an
$O\bigl(\al\beta\rho+\ve)$-approximation for the \dr problem in time
$\poly\bigl(\text{input size},\frac{\ld}{\ve}\bigr)$. 
\end{infcor}

Theorem~\ref{intromainthm} (to a partial extent) and Corollary~\ref{introl1thm} thus
provide novel, useful {\em reductions} from \dr 2-stage optimization to 
2-stage \{stochastic, robust\} (and deterministic) optimization. 
(For instance, \cite{GuptaNRb} 
devise approximations for the \maxmin problem in Corollary~\ref{introl1thm} (i.e.,
$\max_{A'\in\A}g(x,A')$) for scenario sets defined by matroid-independence and/or 
knapsack constraints; 
Corollary~\ref{introl1thm} enables us to export these guarantees to the
corresponding \dr 2-stage problem with the $\frac{1}{2}L_1$ metric.) 
In some cases, we can improve upon the guarantees in Theorem~\ref{intromainthm}.
For certain covering problems, \cite{ShmoysS06} showed how to obtain $\rho=2\al$ via a
decoupling idea; 
by incorporating this idea within our reduction, 
we can improve the guarantee in Theorem~\ref{intromainthm} and obtain 
an $O(\beta_1\beta_2\rho+\ve)$-approximation (see ``Set cover'' in Section~\ref{apps}). 

We demonstrate the versatility of our framework by 
applying Theorem~\ref{intromainthm} and these refinements 
to obtain guarantees for the \dr versions of set cover, vertex cover, edge cover,
facility location, and Steiner tree (Section~\ref{apps}). These constitute the majority of
problems investigated for 2-stage optimization. Our strongest results are for facility
location, vertex cover, and edge cover; for Steiner tree, we obtain results in the
unrestricted setting. 
Table~\ref{results} summarizes these results.

\begin{table}[t!]
\centering
{\small 
\begin{tabular}{|c|c|c|c|c|c|c|}
\hline
\multirow{3}{*}{\begin{tabular}{c}{\bf Problem}\end{tabular}} 
& \multicolumn{5}{c|}{\bf Wasserstein metrics} & 
\multirow{3}{*}{
{\bf \boldmath $L_\infty$, $\A=2^U$}} 
\\ \cline{2-6} 
& \multicolumn{2}{c|}{$\frac{1}{2}L_1$} & \multicolumn{2}{c|}{$\scm^\asym_\infty$ (see \S~\ref{prelim})} & 
\multirow{2}{*}{\begin{tabular}{c} General $\A$, $\scm$ \\ $\beta$=approx. for $g(x,y,A)$\end{tabular} } &  

\\ \cline{2-5}
& 
{$\A=2^U$} & 
{$\A_{\leq k}$}
& \makecell[c]{$\A=2^U$} & \makecell[c]{$\A_{\leq k}$} & & \\ \hline

\makecell[l]{Facility location} & $21.96$ & $196$ & $21.96$ & $196$ &
$O(\beta)$ & $10.98$ \\ \hline
\makecell[l]{Vertex cover} & $16$ & $101.25$ & -- & -- & $O(\beta)$ & $8$ \\ \hline
\makecell[l]{Edge cover} & $12$ & $36$ & -- & -- & $O(\beta)$ &  $6$ \\ \hline
\makecell[l]{Set cover} & $O(\log n)$ & $O(\log^2 n)$ & -- & -- & $O(\beta\log n)$ & $O(\log n)$ \\ \hline
\makecell[l]{Steiner tree} & 160 & * & 160 & * & * & *  \\ \hline
\end{tabular}
\caption{A summary of our results. Recall that $\A_{\leq k}=\{A\sse U:|A|\leq k\}$. We have
  omitted the $O(\ve)$ terms that appear in the factors. The $\scm^\asym_\infty$
  setting does not apply to vertex cover, edge cover, and set cover. The $\beta$-approximation
  for $g(x,y,A)$ is the factor $\beta_1\beta_2$ in Theorem~\ref{intromainthm}. The
  * entries are open questions.\label{results}} 
}
\end{table}

\vspace*{-1ex}
\paragraph{Technical takeaways for \dr problems with Wasserstein metrics.}
The reduction in Theorem~\ref{intromainthm} is obtained by supplementing tools from
2-stage \{stochastic, robust\} optimization 
with various additional ideas. 
Its proof consists of two main components, 
{\em both of which are of independent interest}.  
\begin{description}[style=sameline, leftmargin=0ex, topsep=0.5ex, itemsep=0.5ex, parsep=0ex]
\item[$\bullet$ {\bf Sample average approximation (SAA) for \dr problems.}]\  
In Section~\ref{saa}, we prove that a simple and appealing 
approach in stochastic optimization called the 
SAA method can be applied to {\em reduce the \dr problem to the setting where $\pcent$ has
a polynomial-size support}. 
In the SAA method, we draw some $N$ samples to estimate
$\pcent$ by its empirical distribution $\htp$, and solve the distributionally robust
problem for $\htp$. 
We show that (roughly speaking) by taking $N=\poly\bigl(\text{input size},\frac{\ld}{\ve}\bigr)$
samples, we can ensure 
that a $\beta$-approximate oracle for the SAA objective value can be combined with
a $\rho$-approximation algorithm for the SAA problem, 
to obtain an $O(\beta\rho+\ve)$-approximate solution to the original problem, with high
probability (see Theorem~\ref{mainsaathm}). 
It is well known that $\Omega(\ld)$ samples are needed even for (standard) 2-stage
stochastic problems in the black-box model~\cite{ShmoysS06}.  
Our SAA result substantially expands the scope of problems for which the SAA method is
known to be effective (with $\poly(\text{input size},\ld)$ sample size). Previously, such
results 
were known for the special case of 2-stage stochastic problems~\cite{CharikarCP05,SwamyS12} (see
also~\cite{HomemKS}), and multi-stage stochastic problems with a constant number of
stages~\cite{SwamyS12} (for $\beta,\rho=1$).

Proving our SAA result requires augmenting the SAA machinery for 2-stage stochastic
problems~\cite{CharikarCP05,SwamyS12} with various new 
ingredients to deal with the challenges presented by \dr problems. 
We elaborate in Section~\ref{saa}. 

\item[$\bullet$ {\bf Solving the polynomial-size central-distribution case.}]\ 
Complementing the above SAA result, we show how to approximately solve the \dr
2-stage problem with a polynomial-size central distribution $\ppoly$
(Section~\ref{polysolve}).  
It is natural to move to
a fractional relaxation of the problem, by replacing the first-stage set $X$ by a suitable
polytope $\Pc\supseteq X$. 
In {\em stark contrast} with 2-stage \{stochastic, robust\} optimization, 
where the fractional relaxation of the polynomial-scenario problem immediately gives a
polynomial-size LP and is therefore straightforward to solve in polytime, 
it is substantially more challenging 
to even approximately solve the fractional \dr problem with a polynomial-size central
distribution. In fact, this 
is perhaps the technically more-challenging part of the paper.
The crux of the problem is that, while $\ppoly$ has polynomial-size support, there are
(numerous) distributions $q$ in $\dcol$ 
that have {\em exponential-size support}, and one needs to optimize over such
distributions. 
In particular, if we use duality to reformulate the problem 
$\max_{q:\LW(\ppoly,q)\leq r}\E[A\sim q]{g(x,A)}$ 
as a minimization LP, this leads to an LP with 
{\em an exponential number of both constraints and variables} (see the discussion in
Section~\ref{polysolve}).
Thus, while we started with a 
polynomial-support central distribution, we have ended up in a situation similar to that 
in 2-stage stochastic or robust optimization with an {exponential} number of
scenarios!

To surmount these obstacles, we work with the {\em convex program}
$\min_{x\in\Pc}\hppoly{x}$, and solve this approximately by leveraging the ellipsoid-based
machinery in~\cite{ShmoysS06} (see Theorem~\ref{polythm}). Not surprisingly, this poses
various fresh difficulties, chiefly because we are unable to compute approximate
subgradients as required by~\cite{ShmoysS06}. 
We delve into these issues, and the ideas needed to overcome them in
Section~\ref{polysolve}. 
\end{description}

\vspace*{-1ex}
\paragraph{Approximating \boldmath $g(x,y,A)$.}
We use the following natural strategy: ``guess''
$\scmbound=\scm(A,A^*)$ for the optimal $A^*$, possibly within a $(1+\ve)$-factor, 
and solve the constrained problem (\gxcp):
$\max_{A'\in\A:\scm(A,A')\leq\scmbound}g(x,A')$. 
It is easy to show that a $\beta$-approximation to (\gxcp) 
yields a $\beta(1+\ve)$-approximation for $g(x,y,A)$ (Lemma~\ref{gxyredn}). 
In the unrestricted setting ($\A=2^U$), we will usually be able to solve (\gxcp) exactly,
exploiting the fact that our problems are covering problems.
In the \kbounded setting, we cast (\gxcp) as a \kmaxmin problem (note that $x$ is
integral), and utilize known results for this problem.  

For \drsfl, the result by~\cite{Khandekar} requires creating co-located clients,
which does not work for us. 
We illuminate a novel connection between {\em cost-sharing schemes} and \kmaxmin problems 
by showing 
that a cost-sharing scheme for FL having certain properties can be leveraged 
to obtain an approximation algorithm for \kmaxmin \{integral, fractional\} FL
(see the proof of Theorem~\ref{flasym}). 
In doing so, we also end up improving the approximation factor for \kmaxmin FL
from $10$~\cite{Khandekar} to $6$. 
Whereas cost-sharing schemes have played a role in 2-stage stochastic
optimization, in the context of the boosted-sampling approach of~\cite{GuptaPRS04}, they
have not been used previously for \kmaxmin problems. (The approach in~\cite{GuptaNR} 
has some some similar elements, but there is no explicit use of cost shares.) 
Cost-sharing schemes 
offer a useful tool for designing algorithms
for \kmaxmin problems, that we believe will find further application. 

\vspace*{-1ex}
\paragraph{\dr problems with the \boldmath $L_\infty$ metric.}
For the $L_\infty$ metric (Section~\ref{linfty}), 
we directly consider the fractional relaxation of the problem. As with the Wasserstein
metric, even for a polynomial-scenario central distribution, solving the resulting problem
is quite challenging since it (again) leads to an LP with exponentially many variables and
constraints. We move to a proxy objective that is pointwise close to the true objective,
and show that an $\w$-subgradient of the proxy objective can be computed efficiently at
any point, even for $\w=1/\poly(\text{input size})$. This enables us to use the algorithm 
in~\cite{ShmoysS06} to solve the fractional problem; rounding this solution using a local
approximation algorithm yields results for the \dr discrete 2-stage
problem. 
Table~\ref{results} lists the results we obtain for the $L_\infty$ metric as well.

\vspace*{-1ex}
\subsection{Related work} \label{relwork}
Stochastic optimization is a field with a vast amount of literature 
(see, e.g.,~\cite{BirgeL97,Prekopa95,RuszczynskiS03}), but 
its study from an approximation-algorithms perspective is relatively recent. 
Various approximation results have been obtained in the 2-stage recourse model over the
last 15 years in the CS and Operations-Research (OR) literature (see, e.g.,~\cite{SwamyS06}), 
but more general models, such as distributionally robust stochastic optimization,
have received little or no attention in this regard.

To the best of our knowledge, with the exception of~\cite{AgarwalDSY10}, which we discuss
below, there are no prior approximation algorithms for distributionally robust 2-stage
discrete optimization problems, when the number $|\A|$ of possible scenarios is (finite,
but) exponentially large (even if $\pcent$ has polynomial-size support). 
Much of the work in the stochastic-optimization and OR literature on these problems has
focused on proving suitable duality results that sometimes allow one to reformulate the
\dr problem more compactly. Moreover, in many cases, the results obtained are for 
{\em continuous} scenario spaces and with other assumptions about the recourse costs.
For instance, \cite{EsfahaniK17,GaoK,ZhaoG15,HanasusantoK18} all consider the setting where
$\dcol$ is a ball in the Wasserstein metric, and provide a closed-form description
of the worst-case distribution in $\dcol$, which is then used to reformulate the \dr
problem under further convexity assumptions of the scenario collection $\A$. 
\dr problems have gained attention in recent years due to their 
usefulness in inferring decisions from observed data while avoiding the risk of
overfitting: here $\dcol$ is used to model a class of distributions 
from which the observed data could arise (with high confidence).
Various works have advocated the use of a Wasserstein ball around the empirical
distribution $\htp$ for this purpose~\cite{EsfahaniK17,ZhaoG15,GaoK,VanparysEK}, but there
are no results proving polynomial bounds on the number of samples needed in order to
produce provably-good results. Note that these works, by definition, consider the 
setting where the central distribution has polynomial-size support.
The distributionally robust setting has also been considered for chance-constrained
problems; see, e.g.~\cite{ErdoganI} and the references therein.

The work of~\cite{AgarwalDSY10} in the CS literature on {\em correlation gap} can be
interpreted as studying distributionally robust 
discrete-optimization problems, but in a very different setting where $\dcol$ is not a
ball.  
Instead, $\dcol$ is the collection of distributions that agree with some given
expected values; the correlation gap quantifies the worst-case ratio of the \dr
objective when one chooses the optimal decisions with respect to 
the distribution in $\dcol$ that treats all random variables as independent, versus the
optimum of the \dr problem. Agrawal et al.~\cite{AgarwalDSY10} proved various $O(1)$
bounds on the correlation gap for submodular functions and subadditive functions admitting
suitable cost shares. 
Various other works (see, e.g.,~\cite{DelageY10,Popescu07} and the references therein) have
considered such moment-based collections, but again under continuity and/or convexity
assumptions about the scenario space and/or recourse costs. 

We now briefly survey the work on approximation algorithms under the stochastic- and
robust- optimization models, which the \dr model generalizes.
As noted above, various approximation results have been obtained for 2-stage, and even
multistage problems.
In the black-box model, a common approach is the SAA method, which simply consists of
solving the stochastic-optimization problem for the empirical distribution $\htp$ obtained
by sampling.
The effectiveness of this method has been analyzed both for 2-stage stochastic
problems~\cite{HomemKS,CharikarCP05,SwamyS12} and multi-stage stochastic
problems~\cite{SwamyS12}. 
The sample-complexity bound in~\cite{HomemKS} is a non-polynomial bound for general 
2-stage stochastic problems, 
whereas~\cite{CharikarCP05,SwamyS12} both obtain $\poly(\text{input size},\ld)$ bounds for
structured problems.  
The proof in~\cite{SwamyS12} applies also to structured multistage linear programs,
and~\cite{CharikarCP05} show that even approximate solutions to the 
2-stage SAA problem translate to approximate solutions to the original 2-stage
problem. We build upon the SAA machinery of Charikar et
al.~\cite{CharikarCP05}. Previously, Shmoys and Swamy~\cite{ShmoysS06} showed how to use
the ellipsoid method to solve structured 2-stage linear programs in the black-box model,
and how to round the resulting fractional solution. We utilize their machinery based on
approximate subgradients to solve the polynomial-scenario central-distribution setting.
Approximation algorithms for 2-stage problems have also been developed via combinatorial
means. The prominent technique here is the boosted sampling technique of Gupta et al.%
~\cite{GuptaPRS04}; the survey~\cite{SwamyS06} gives a detailed description of these and
other approximation results for 2-stage optimization.

Two-stage robust optimization where uncertainty is reflected in the constraints and not
the data was proposed in~\cite{DhamdhereGRS05}, who devised approximation algorithms for
various problems in the {\em polynomial-scenario setting}. Notice that it is not
clear how to even specify problems with exponentially many scenarios in the robust
model. Feige et al.~\cite{FeigeJMM05} expanded the model of~\cite{DhamdhereGRS05} by
considering what we call the \kbounded setting, where 
every subset of at most $k$ elements is a scenario. 
Subsequently,~\cite{Khandekar} and~\cite{GuptaNR} expanded the collection of results known
for 2-stage robust problems in the \kbounded setting. 
We utilize results for the closely-related \kmaxmin problem encountered in this setting in 
our work. 

We briefly discuss a few other snippets that consider intermediary approaches
between stochastic and robust optimization.
Swamy~\cite{Swamy11} considers a model for 
{\em risk-averse 2-stage stochastic optimization} that
interpolates between the stochastic and robust optimization approaches.
In the context of {\em online algorithms}, Mirrokni et al.~\cite{Mirrokni12} and
Esfandiari et al.~\cite{Esfandiari15} give online algorithms for allocation problems that
are simultaneously competitive both in a random input model and in an adversarial input
model. Finally, we note that our distributionally robust setting can be seen to be in a
similar spirit as a recent focus in algorithmic mechanism design, 
where one does not assume precise knowledge of the underlying distribution;
rather one (implicitly) has a collection of distributions, 
and one seeks to design mechanisms that work for every distribution in this collection;
see, e.g.,~\cite{HuangMR15}.

\section{Problem definitions, and our general class of \dr 2-stage problems}
\label{prelim}

Recall that we consider settings where we have a ball 
$\dcol=\{q: \pdist(\pcent,q)\leq r\}$ of distributions (over the scenario-collection $\A$)
around a central distribution $\pcent$ under some metric $\pdist$ on distributions, and we
seek to minimize the maximum expected cost with respect to a distribution in $\D$.  
As mentioned earlier, we make no assumptions about $\pcent$, and 
only require the ability to draw samples from $\pcent$. 
The metrics that we consider for $\pdist$ 
are the $L_\infty$ metric, 
$\frac{1}{2}L_1$ metric, 
and the Wasserstein metric. 
We now define Wasserstein metrics precisely.

\begin{definition}[{\bf Wasserstein (a.k.a transportation or earth-mover) distance}] 
\label{transdef}
The Wasserstein distance between two probability distributions $p$ and $q$ over $\A$ is
defined with respect to an underlying metric $\scm$ on $\A$. 
A {\em transportation plan} or {\em flow} from $p$ to $q$ is a vector
$\gamma\in\Rplus^{\A \times \A}$ such that:
(i) $\sum_{A' \in \A}\gamma_{A,A'} = p_A$ for all $A \in \A$; and
(ii) $\sum_{A \in \A} \gamma_{A, A'}=q_{A'}$ for all $A' \in \A$. 
The {\em Wasserstein distance} between $p$ and $q$, denoted $\LW(p,q)$, is the minimum
value of $\sum_{A,A'}\gm_{A,A'}\scm(A,A')$ over all transportation plans from $p$ to $q$. 

If $\scm$ is an asymmetric metric, then $\LW$ is an asymmetric metric; if $\scm$ is a
pseudometric---i.e., $\scm$ satisfies the triangle inequality but $\scm(A,A')$ could be $0$
for $A\neq A'$---then so is $\LW$. 
\end{definition}

In Section~\ref{apps}, 
we consider the \dr versions of set cover (and some special cases), facility location, and
Steiner tree. \dr 2-stage facility location (\drsfl) was defined in Section~\ref{contr};
we define the remaining problems below, and then discuss the general class of \dr 2-stage
problems to which our framework applies.  
Recall that $\inpsize$ denotes the input size.

\begin{enumerate}[label=$\bullet$, topsep=0.5ex, itemsep=0.5ex,
    labelwidth=\widthof{$\bullet$}, leftmargin=!] 
\item {\bf \dr 2-stage set cover (\drssc).} 
We have a collection $\Sc$ of subsets 
over a ground set $U$. A scenario is a subset of $U$ and specifies the set of elements to
be covered in that scenario.  
We may buy a set $S\in\Sc$ in either
stage, incurring costs of $c_S$ and $c^\two_S$ in stages I and II respectively. 
The sets chosen in stage I and in each scenario $A$ must together cover $A$. 
The goal is to choose some first-stage sets
$\Sc^\one\sse\Sc$ and sets $\Sc^A\sse\Sc$ in each scenario $A$ 
so as to minimize
$\sum_{S\in\Sc^\one}c_S+\max_{q:\pdist(\pcent,q)\leq r}\E[A\sim q]{\sum_{S\in\Sc^A}c^\two_S}$. 

We have $\ld:=\max\{1,\max_{S\in\Sc}c^\two_S/c_S\}$, and 
$\inpsize$ is the encoding size of $\bigl(U,\Sc,c,c^\two,r\bigr)$.
We consider the \allsets ($\A=2^U$) and \kbounded ($\A=\{A\sse U:|A|\leq k\}$)
settings. Different scenarios could be quite unrelated, so there does not seem to be a
natural choice for a (non-discrete) scenario-metric; we therefore consider (balls in) the 
$L_\infty$ or $\frac{1}{2}L_1$ metrics. 

\item {\bf \dr 2-stage Steiner tree (\drst).} 
We have a complete graph $G=(V,E)$ with metric edge costs $\{c_e\}_{e \in E}$, root 
$s\in V$, and inflation factor $\ld \ge 1$. A scenario is a subset of nodes 
$A \subseteq V$ (called  {\em terminals}) specifying the nodes that need to be connected
to $s$. We may buy an edge $e \in E$ in stages I or II, incurring costs $c_e$ or
$c^\two_e=\ld c_e$ respectively. The union of the edges $F\sse E$ bought in stage I, and
$F^A\sse E$ bought in scenario $A$, must connect all nodes in $A$ to $s$, and 
we want to minimize 
$\sum_{e\in F}c_e+\max_{q:\pdist(\pcent,q)\leq r}\E[A\sim q]{\sum_{e\in F^A}c^\two_e}$.
(With non-uniform inflation factors for different edges, even 2-stage stochastic Steiner
tree becomes at least as hard as {\em group Steiner tree}~\cite{Ravi:2004gp}.)

Here $\inpsize$ is the encoding size of $(G,c,r)$. 
We obtain results in the \allsets setting, and leave the \kbounded setting for future
work. As with \drsfl, in addition to the $L_\infty$ and $\frac{1}{2}L_1$ metrics, we can
consider scenario metrics defined using $c$ (e.g., $\scm^\asym_\infty$) and the resulting
Wasserstein metrics. 
\end{enumerate}

\paragraph{A general class of \dr 2-stage problems.}
Abstracting away the key properties of \drsfl, \drssc, \drst, we 
now define the generic \dr 2-stage problem that we consider.
As before, $X$ denotes the finite first-stage action set of the discrete problem.
It will be convenient to consider 
the natural fractional relaxation of the \dr problem obtained by 
enlarging the discrete second-stage action set and $X$ to suitable polytopes.
Recall that $g(x,A)$ is the optimal second-stage cost of scenario $A$ given $x$ as the
first-stage decision, when we allow fractional second-stage actions. 
Let $\Pc\sse\Rplus^m$ denote the polytope specifying the fractional first-stage 
decisions, 
with $X=\Pc\cap\Z^m$. 
(For example, for \drssc, $g(x,A)$ is the optimal value of a set-cover LP 
where we may buy sets fractionally in the second stage, and $\Pc=[0,1]^m$.) 
One benefit of moving to the fractional relaxation is that, for every scenario $A$,
$g(x,A)$ is a convex function of $x$, whose value and subgradient can be exactly computed.    

\begin{defn} \label{apsgrad}
Let $f:\R^m\rightarrow\R$ be a function. We say that $\sgr\in\R^m$ is a {\em subgradient} of $f$
at $u\in\R^m$ if we have $f(v)-f(u)\geq\sgr\cdot(v-u)$ for all $v\in\R^m$.
Given $S\sse\R^m$, we say that $\hsgr$ is an {\em $(\w,S)$-subgradient} of $f$ at the point
$u\in S$ if for every $v\in S$, we have $f(v)-f(u)\geq\hsgr\cdot(v-u)-\w f(u)$. 
We abbreviate $(\w,\Pc)$-subgradient to $\w$-subgradient.
\end{defn}

Following~\cite{CharikarCP05,ShmoysS06,SwamyS12}, we 
consider the following generic \dr 2-stage problem \eqref{Qdisc} with discrete first-stage
set $X$, and its
(further) fractional relaxation \eqref{Qpcentral}, and require that they satisfy properties
\ref{p1}--\ref{p6} listed below.
Let $\|u\|$ denote the $L_2$-norm of $u$. 

\vspace*{-1.5ex}
\noindent
\begin{minipage}{0.6\textwidth}
\begin{equation}
\min_{x \in X} \quad 
\hpcentral{x} := c^\T x + \max_{q:\pdist(\pcentral, q) \le r}\E[A \sim q]{g(x,A)} 
\tag{\Qdisc} \label{Qdisc}
\end{equation}
\end{minipage}
\qquad
\begin{minipage}{0.3\textwidth}
\begin{equation}
\min_{x\in\Pc} \hpcentral{x} \tag{\Qpcentral} \label{Qpcentral}
\end{equation}
\end{minipage}

\medskip

In proving their SAA result for 2-stage stochastic problems, \cite{CharikarCP05} 
define properties \ref{p1}, \ref{p2} below to capture the fact that every first-stage
action has a corresponding recourse action that is more expensive by a bounded factor, and  
hence, it is always feasible to not take any first-stage actions. 

\begin{enumerate}[label=(P\arabic*), topsep=0.5ex, itemsep=0.25ex, parsep=0ex, labelwidth=\widthof{(P6)},
    leftmargin=!] 
\item \label{p1}
$\bo\in X$, $c\geq\bo$, $\log|X|=\poly(\inpsize)$, 
and $0\leq g(x,A)\leq g(\bo,A)$ for all $x\in\Pc, A\in\A$.

\item \label{p2}
We know an {\em inflation parameter} $\ld\geq 1$ such that 
$g(\bo,A)\leq g(x,A)+\ld c^\T x$ for all $x\in\Pc, A\in\A$. 
\end{enumerate}

Since we apply the ellipsoid-based machinery in~\cite{ShmoysS06} to solve the
fractional problem with a polynomial-size central distribution, we need
bounds on the feasible region $\Pc$ in terms of enclosing and enclosed balls; this is 
captured by \ref{p3}, which is directly lifted from~\cite{ShmoysS06}. 
Note that the vast majority of 2-stage problems (including \drsfl, \drssc, \drst)
involve $\{0,1\}$ decisions, with $X=\{0,1\}^m$ and so $\Pc=[0,1]^m$, so \ref{p3} is
readily satisfied.
As in~\cite{ShmoysS06}, we need to be able to compute the value and subgradient of 
the recourse cost $g(x,A)$, which is a benign requirement since $g(x,A)$ is the optimal 
value of a polytime-solvable LP in all our applications. 
Whereas~\cite{ShmoysS06} define a syntactic class of 2-stage stochastic LPs 
and show (implicitly) that they satisfy this requirement, 
we explicitly isolate this requirement in \ref{p4}, \ref{p5}. 

\begin{enumerate}[label=(P\arabic*), start=3, topsep=0.5ex, itemsep=0.25ex, parsep=0ex,
    labelwidth=\widthof{(P6)}, leftmargin=!] 
\item \label{p3}
We have positive bounds $R$ and $V\leq 1$ such that 
$\Pc\sse B(\bo,R):=\{x:\|x\|\leq R\}$ and $\Pc$ contains a ball of radius $V$ such that
$\ln\bigl(\frac{R}{V}\bigr)=\poly(\inpsize)$. 

\item \label{p4}
For every $A\in\A$, $g(x,A)$ is convex over $\Pc$, and can be efficiently computed for
every $x\in\Pc$.  

\item \label{p5}
For every $x\in\Pc, A\in\A$, we can efficiently compute a subgradient $\sgr$ of
$g(x,A)$ at $x$ with $\|\sgr\|\leq K$, where $\ln K=\poly(\inpsize)$.
Hence, the Lipschitz constant of $g(x,A)$ is at most $K$ (due to
Definition~\ref{apsgrad}).
\end{enumerate}

Finally, 
we need the following additional mild condition.
 
\begin{enumerate}[label=(P\arabic*), start=6, topsep=0.5ex, itemsep=0.25ex, parsep=0ex,
    labelwidth=\widthof{(P6)}, leftmargin=!] 
\item \label{p6}
When $\pdist$ is the Wasserstein metric with respect to a scenario metric $\scm$, we
know $\tau\geq 1$ with $\ln\tau=\poly(\inpsize)$ such that 
$g(x,A')-g(x,A)\leq\tau\cdot\scm(A,A')$ for all $x\in\Pc$ and all $(A,A')$ with
$\scm(A,A')>0$.
\end{enumerate}

\smallskip \noindent
As noted above, \ref{p1}--\ref{p5} are gathered from~\cite{CharikarCP05,ShmoysS06},
and hold for all the 2-stage problems considered in the CS literature 
(see~\cite{SwamyS12,DhamdhereGRS05,FeigeJMM05,Khandekar,GuptaNR});
\ref{p6} is a new requirement, but 
is also rather mild and holds for all the problems we consider.
\ref{p1}, \ref{p2} 
and \ref{p6} are
used to prove that SAA works for the \dr problem under the Wasserstein metric
(Section~\ref{saa}). \ref{p3}--\ref{p5} pertain to the fractional relaxation, and are
utilized to show that one can efficiently solve the SAA problem approximately
(Section~\ref{polysolve}).  

A solution to \eqref{Qdisc} needs to be rounded to yield integral second-stage actions: 
any LP-relative $\al$-approximation algorithm for the deterministic version of the
problem can be used to obtain recourse actions for each scenario $A$ having cost at
most $\al\cdot g(x,A)$. 
To round a fractional solution to \eqref{Qpcentral}, 
we utilize a local approximation algorithm for the 2-stage problem: we say that
$\alg$ is a {\em local $\rho$-approximation algorithm} for \eqref{Qpcentral} if, given any
$x\in\Pc$, it returns an integral solution $\tx\in X$ and implicitly specifies integral
recourse actions $\tz^A$ for every $A\in\A$, such that $c^\T\tx\leq\rho(c^\T x)$
and $\text{(cost of $\tz^A$)}\leq\rho g(x,A)$ for all $A\in\A$. 
An $\al$-approximate solution to \eqref{Qpcentral} combined with a local
$\rho$-approximation yields an $\al\rho$-approximate solution to the discrete \dr
2-stage problem. 
Local approximation algorithms exist for various 2-stage problems%
---e.g., set cover, vertex cover, facility location~\cite{ShmoysS06}---%
with approximation factors that are comparable to the approximation factors
known for their deterministic counterparts.

\section{Distributionally robust problems under the Wasserstein metric} \label{wasserstein}
We now focus on the \dr 2-stage problem \eqref{Qdisc} when $\pdist$ is the
Wasserstein metric $\LW$ with respect to a metric $\scm$ on scenarios. 
Plugging in the definition of $\LW$ (with respect to scenario metric $\scm$),
we can rewrite \eqref{Qdisc} as follows.

\vspace*{-3.5ex}
\noindent
\begin{minipage}[t]{0.6\textwidth}
\leqnomode
\begin{equation}
\tag{\Qdisc}
\min_{x \in X}\ \hpcentral{x} := c^\T x + \zpcentral{x}, 
\qquad \text{where}\ \zpcentral{x}\ :=\ 
\end{equation}
\reqnomode
\end{minipage}
\hspace*{-0.4in}
\begin{minipage}[t]{0.45\textwidth}
\begin{alignat}{2}
\max & \ \ & \sum_{A,A'}\gm_{A,A'}&g(x,A') \tag{T${}_{\pcent,x}$} \label{zdefn} \\
\text{s.t.} && \sum_{A'} \gm_{A,A'} & \leq \pcent_A \quad \ \forall A\in\A 
\label{pbnd} \\
&& \sum_{A,A'}\scm(A,A')&\gm_{A,A'} \leq r \label{rbnd} \\[-2ex]
&& \gm & \geq 0. \label{noneg}
\end{alignat}
\end{minipage} 

\medskip
Let $\iopt:=\min_{x\in X}\hpcentral{x}$ denote the optimal value of \eqref{Qdisc}.
We note that a naive, simplistic approach that ignores 
the uncertainty in the underlying distribution,
and only considers the central distribution $\pcent$, 
yields (expectedly) poor bounds. 
Suppose $\bx$ is an $\al$-approximate solution for the 2-stage problem 
$\min_{x\in X}\bigl(c^\T x+\E[A\sim\pcent]{g(x,A)}\bigr)$. 
Given \ref{p6}, one can show that 
$\zpcentral{\bx}\leq\E[A\sim\pcent]{g(\bx,A)}+\tau\cdot r$ (and is at least
$\E[A\sim\pcent]{g(\bx,A)}$), which implies $\hpcentral{\bx}\leq\al\cdot\iopt+\tau\cdot r$,
but this is too weak a guarantee since $\tau\cdot r$ could be quite large
compared to $\iopt$. 

In Section~\ref{saa}, we work with \eqref{Qdisc} and show that the
SAA approach can be used to reduce 
to the case where the central distribution has {\em polynomial-size} support. 
In Section~\ref{polysolve}, we show how to approximately solve the polynomial-size
support case by applying the ellipsoid method to its (further) relaxation
\eqref{Qpcentral}, where we replace
$X$ with $\Pc$. 
Here, we utilize a local approximation algorithm to move from $\Pc$ to $X$, and thereby 
interface with, and complement, the SAA result for \eqref{Qdisc} proved in
Section~\ref{saa}. 
This result applies more generally, even when $\scm$ is not a metric; we only require
that $\scm(A,A)=0$ for all $A\in\A$. (If $\scm$ is not a metric, the Wasserstein
distance with respect to $\scm$ need not yield a metric on distributions.)

In Section~\ref{apps}, we consider various combinatorial-optimization
problems, and utilize the above results in conjunction to obtain the {\em first}
approximation results for the \dr versions of these problems.

\subsection{A sample-average-approximation (SAA) result for distributionally robust problems} 
\label{saa} 
The SAA approach is the following simple, intuitive
idea: draw some $N$ samples from $\pcent$, estimate $\pcent$ by the empirical distribution
$\htpcent$ induced by these samples, and solve the {\em SAA problem} (\Qdisc[\htpcent]).
We prove the following SAA result. 
For any $\ve\leq\frac{1}{3}$, if we construct $O\bigl(\frac{1}{\ve}\bigr)$ 
SAA problems, each using 
$\poly\bigl(\inpsize,\frac{\ld}{\ve},\log(\frac{1}{\kp})\bigr)$ independent samples, 
and if we have a $\beta$-approximation algorithm for computing the objective value of the SAA
problem at any given point, then we can utilize $\rho$-approximate solutions to these SAA
problems to obtain a solution $\hx\in X$ satisfying 
$\hpcentral{\hx}\leq 4\beta\rho\bigl(1+O(\ve)\bigr)\cdot\iopt+2\beta\rho\kp$
with high probability; 
Theorem~\ref{mainsaathm} gives the precise statement. 

The proof has several ingredients. 
There are two main approaches~\cite{CharikarCP05,SwamyS12} for showing that the SAA method
with a polynomial number of samples works for stochastic-optimization problems. 
Charikar et al.~\cite{CharikarCP05} prove the following SAA result for 2-stage problems.
\begin{theorem}[\cite{CharikarCP05}] \label{ccpsaa}
Consider a 2-stage problem \stocp\textnormal{:}
\hfill $\min_{x\in\tX}\ \bigl(\fp{x}:=\tc^\T x+\E[A\sim p]{\tg(x,A)}\bigr)$, \hfill 
with scenario set $\tA$, where $(\tX,\tc,\tg,\tA)$ satisfy \ref{p1}, \ref{p2} with
inflation parameter $\Ld$. 
With probability at least $1-\dt$, any optimal solution to the SAA problem constructed
using $\poly\bigl(\log|\tX|,\frac{\Ld}{\ve},\log(\frac{1}{\dt})\bigr)$ samples is a
$(1+\ve)$-approximate solution to \stocp. 
More generally, 
there is a way of using an $\al$-approximation algorithm for the SAA problem, 
in conjunction with 
a $\beta$-approximate objective-value oracle for the SAA problem,
to obtain an $\bigl(\al\beta+O(\ve)\bigr)$-approximate solution to
\stocp with high probability. 
\end{theorem} 

Note that \eqref{Qdisc} is {\em not} a standard 2-stage stochastic-optimization problem
because constraint \eqref{rbnd} couples the various scenarios, 
which prevents us from applying Theorem~\ref{ccpsaa} to \eqref{Qdisc}. 
The SAA result in Swamy and Shmoys~\cite{SwamyS12} applies to the fractional relaxation of 
the problem, 
and works whenever the objective functions of the SAA and original problems satisfy 
a certain ``closeness-in-subgradients'' property. 
A subgradient of $\hpcentral{\cdot}$ at a point $x\in\Pc$ is obtained from 
the optimal distribution $q$ 
to the inner maximization problem in \eqref{Qdisc}.
This is however an exponential-size object and utilizing this to prove
closeness-in-subgradients seems quite daunting.

Our first insight is that 
we can decouple the scenarios by {\em Lagrangifying} constraint
\eqref{rbnd} using a dual variable $y\geq 0$. 
By standard duality arguments, this leads to the following reformulation of
\eqref{Qdisc}.
\begin{gather}
\min_{x\in X} \quad \Bigl[c^\T x +
\underbrace{\min_{y\geq 0} \Bigl(ry+
\max\ \Bigl\{\sum_{A,A'}\gm_{A,A'}(g(x,A')-y\cdot\scm(A,A')): \ \ 
\gm\geq 0, \ \ 
\sum_{A'}\gm_{A,A'}\leq\pcent_A\ \ \forall A\in\A\Bigr\}\Bigr)}_{\text{\small{$\zpcentral{x}$}}}\Bigr]
\notag \\
\hspace*{-0.2in}
\text{which simplifies to} \quad
\min_{x\in X, y\geq 0} \ \ \hpcentral{x,y} := c^\T x+ry+
\E[A\sim\pcent]{\max_{A'\in\A}\bigl(g(x,A')-y\cdot\scm(A,A')\bigr)}. 
\tag{\Rdisc} \label{Rdisc}
\end{gather}

Recall that $g(x,y,A):=\max_{A'\in\A}\bigl(g(x,A')-y\cdot\scm(A,A')\bigr)$. 
Let $\scmmax:=\max_{A,A'}\scm(A,A')$. 
The chief benefit of the reformulation \eqref{Rdisc} is that we can view \eqref{Rdisc} as
a 2-stage problem: the first-stage action-set is $X\times\Rplus$, and the optimal
second-stage cost of scenario $A$ under first-stage actions $(x,y)$ is given by
$g(x,y,A)$. This makes it more amenable to utilize the SAA machinery developed for
2-stage problems. We can exploit
\ref{p6} to show that we may limit $y$ to the range $[0,\tau]$ in \eqref{Rdisc}, 
and use \ref{p2} to bound the inflation factor of \eqref{Rdisc}. 

\begin{lemma} \label{lem:xy_to_x} \label{xytox} \label{ybnd}
For any $x\in X$, there exists $y\in[0,\tau]$ such that $\hpcentral{x}=\hpcentral{x,y}$.
Hence, $x\in X$ is an $\al$-approximate solution to \eqref{Qdisc} iff 
$\exists y\in[0,\tau]$ such that $(x,y)$ is an $\alpha$-approximate solution to
\eqref{Rdisc}.  
\end{lemma}

\begin{proof} 
The second statement is immediate from the first one since \eqref{Qdisc} and \eqref{Rdisc}
have the same optimal values. So we focus on showing the first statement.

Consider any $x\in X$. There exists $y^*\geq 0$ such that $\hpcentral{x}=\hpcentral{x,y^*}$.
If $y^*\leq\tau$, then we are done. So suppose $y^*>\tau$. 
We argue that $\hpcentral{x^*,\tau}\leq\hpcentral{x,y^*}$. This completes the proof since
we also have $\hpcentral{x}\leq\hpcentral{x,y}$ for all $y\geq 0$.
Clearly, $c^\T x+r\tau\leq c^\T x+ry^*$. 
If $A'\in\A$ is such that $g(x,y^*,A)=g(x,A')-y^*\cdot\scm(A,A')$, then it must be that
$\scm(A,A')=0$. Otherwise, 
$g(x,A')-y^*\cdot\scm(A,A')<g(x,A')-\tau\cdot\scm(A,A')\leq g(x,A)$, 
where the last inequality follows from \ref{p6}.
This contradicts the choice of $A'$. Therefore, we have
$g(x,y^*,A)=\max_{A'\in\A:\scm(A,A')=0}g(x,A')=g(x,\tau,A)$, completing the proof.
\end{proof}

\begin{lemma} \label{rdiscld}
For the 2-stage problem \eqref{Rdisc}, 
we can set the parameter $\Ld$ in Theorem~\ref{ccpsaa} to be
$\max\bigl\{\ld,\frac{\scmmax}{r}\bigr\}$. 
\end{lemma}

\begin{proof} 
Consider any $x\in X$, $y\geq 0$, and $A\in\A$. Let $A'\in\A$ be such that $g(0,0,A)=g(0,A')$.
Then 
$$
g(0,0,A)-g(x,y,A)\leq g(0,A')-\bigl(g(x,A')-y\cdot\scm(A,A')\bigr)
\leq\ld c^\T x+y\cdot\scmmax\leq\max\Bigl\{\ld,\tfrac{\scmmax}{r}\Bigr\}(c^\T x+ry).
$$
The second inequality above follows from \ref{p2}.
\end{proof}

Given Lemmas~\ref{ybnd} and~\ref{rdiscld}, by suitably discretizing 
$[0,\tau]$, 
one can use Theorem~\ref{ccpsaa} to show that: if we construct the SAA problem   
$\min_{x\in X}\hpcentral[\htp]{x}\ \equiv\ \min_{x\in X,y\in[0,\tau]}\hpcentral[\htp]{x,y}$
using $\poly\bigl(\inpsize,\frac{\ld}{\ve},\log\tau,\frac{\scmmax}{r}\bigr)$ samples, 
and can compute (approximately) the SAA objective value $\hpcentral[\htp]{x,y}$ at any
given point, then, with high probability, one can translate an $\al$-approximate solution
to the SAA problem to an $O(\al+\ve)$-approximate solution to \eqref{Qdisc}. 
But this result does not quite suit our purposes due to various reasons. 

The term $\frac{\scmmax}{r}$ could be rather large, and is not $\poly(\inpsize,\ld)$, so 
this does not yield polynomial sample complexity.%
\footnote{The problem persists even if we utilize the closeness-in-subgradients machinery
in~\cite{SwamyS12} to the fractional version of \eqref{Rdisc}. 
This would involve estimating $\E[A\sim p]{\scm(A,\pi(x,y,A))}$ to within an $\ve r$
term, where $\pi(x,y,A)=\argmax_{A'\in\A}\bigl(g(x,A')-y\cdot\scm(A,A')\bigr)$, 
which requires $O\left(\frac{\scmmax}{\ve r}\right)$ samples.}
Moreover
it seems difficult to compute the SAA objective value $\hpcentral[\htp]{x,y}$, or even
approximate it. 
This difficulty arises because 
computing $g(x,y,A)$ 
encompasses the \nphard \kmaxmin problem encountered in 2-stage robust optimization, and 
furthermore, the mixed-sign objective in $g(x,y,A)$ makes it hard to even approximate
$g(x,y,A)$ (see Theorem~\ref{saainapx}).

We need various ideas to circumvent these issues.
We show that we can eliminate the dependence on $\frac{\scmmax}{r}$ altogether at
the expense of a slight deterioration in the approximation ratio when moving from the SAA 
to the original problem. 
The $\frac{\scmmax}{r}$ term arises because $g(0,0,A)$ might be attained by a scenario
$A'$ where $\scm(A,A')\approx\scmmax$ (see the proof of Lemma~\ref{rdiscld}). 
Our crucial second insight is that we can eliminate this and reduce the sample complexity
to $\poly(\inpsize,\ld)$,  
by specifically imposing that we
never encounter $(A,A')$ pairs with $\scm(A,A')>M:=\ld r$; 
we call such pairs 
{\em long} edges, and the remaining pairs {\em short} edges.
Any $\gm$ satisfying \eqref{rbnd} can send at most $\frac{r}{M}=\frac{1}{\ld}$ flow on
the long edges. Motivated by this, we ``decompose'' $\zpcentral{x}$ into $\zpcentshort{x}$
and $\zpcentlng{x}$, which are (roughly speaking) the contribution from the short and long
edges respectively. 
(This decomposition is akin to the division of low-
and high- cost scenarios 
used by~\cite{CharikarCP05} to prove Theorem~\ref{ccpsaa}, but there are significant
technical differences, which complicate things for us, as we discuss below.) 
We define $\zpcentshort{x}$ and $\zpcentlng{x}$ as follows.  
\begin{alignat*}{1}
\zpcentshort{x}\ &:=\ \max\ \biggl\{\sum_{A,A'}\gm_{A,A'}g(x,A'): \quad 
\eqref{pbnd}, \eqref{rbnd}, \eqref{noneg}, \ \ \gm_{A,A'}=0\ \ \text{if}\ \ \scm(A,A')>M\biggr\} \\
\zpcentlng{x}\ &:=\ \max\ \biggl\{\sum_{A,A'}\gm_{A,A'}g(x,A'): \quad 
\eqref{pbnd}, \eqref{rbnd}, \eqref{noneg}, 
\quad \sum_{A,A'}\gm_{A,A'}\leq\tfrac{1}{\ld}\biggr\}.
\end{alignat*}

\begin{lemma} \label{lem:decomposition} \label{lem:high_neg} \label{zproxy}
For every central distribution $p$, and every $x \in\Pc$, we have 
$\hpcentral[p]{x}\leq c^\T x+\zpcentshort[p]{x}+\zpcentlng[p]{0}\leq 2\hpcentral[p]{x}$.
\end{lemma}

\begin{proof} 
We prove this by showing that: 
(i) $\zpcentral[p]{x}\le\zpcentshort[p]{x}+\zpcentlng[p]{x}\le 2\zpcentral[p]{x}$; and
(ii) $\zpcentlng[p]{x}\leq\zpcentlng[p]{0}\leq\zpcentlng[p]{x}+ c^\T x$.
Given these bounds, the upper bound on $\hpcentral[p]{x}$ follows from the upper
bounds on $\zpcentral[p]{x}$ and $\zpcentlng[p]{x}$ in parts (i) and (ii) respectively. 
For the other direction, we have 
$$ 
c^\T x+\zpcentshort[p]{x}+\zpcentlng[p]{0}
\leq c^\T x+\zpcentshort[p]{x}+\zpcentlng[p]{x}+ c^\T x
\leq 2 c^\T x+2\zpcentral[p]{x}= 2\hpcentral[p]{x} ,
$$
where the first and second inequalities follow from the second inequalities of parts (ii)
and (i) respectively.

Part (ii) follows from property \ref{p2}. For any feasible solution $\gm$ to the
optimization problem defining $\zpcentlng[p]{0}$ (and $\zpcentlng[p]{x}$), we have
$$
\sum_{A,A'}\gm_{A,A'}g(x,A')\leq\sum_{A,A'}\gm_{A,A'}g(0,A')
\leq\sum_{A,A'}\gm_{A,A'}g(x,A')+(\ld c^\T x)\sum_{A,A'}\gm_{A,A'}
\leq\sum_{A,A'}\gm_{A,A'}g(x,A')+ c^\T x.
$$

We now prove part (i). It is clear from the definition that
$\zpcentshort[p]{x},\zpcentlng[p]{x}\leq\zpcentral[p]{x}$, so the second inequality holds. 
For the first inequality, consider any feasible solution $\gm$ to \eqref{zdefn}. 
Let $\gm^\short$ be the restriction of $\gm$ to the short edges, along with $0$s for the
long edges. Similarly, let $\gm^\lng$ be the restriction of $\gm$ to the long edges,
along with $0$s for the short edges. Then $\gm^\short$ and $\gm^\lng$ are feasible
solutions to the optimization problems defining $\zpcentshort[p]{x}$ and
$\zpcentlng[p]{x}$ respectively. This yields the first inequality in (i).
\end{proof}

Given Lemma~\ref{zproxy}, 
we focus on the {\em thresholded proxy problem} \eqref{bQdisc} below, 
and its reformulation obtained (as before) by Lagrangifying \eqref{rbnd} and simplifying.  

\leqnomode
\noindent
\begin{minipage}{0.4\textwidth}
\begin{equation}
\min_{x\in X}\ \bhpcentral{x} := c^\T x+\zpcentshort{x}, \tag{\bQdisc} \label{bQdisc}
\end{equation}
\end{minipage}
\quad
\begin{minipage}{0.6\textwidth}
\begin{equation}
\tag{\bRdisc} \label{bRdisc}
\min_{x\in X,y\geq 0} \ \bhpcentral{x,y} := c^\T x+ry+
\E[A\sim\pcent]{\bg(x,y,A)}, 
\end{equation}
\end{minipage}
\reqnomode
\smallskip

\noindent
where $\bg(x,y,A):=\max_{A'\in\A:\scm(A,A')\leq M}\bigl(g(x,A')-y\cdot\scm(A,A')\bigr)$.
After suitably discretizing the $y$-interval $[0,\tau]$, we obtain that the 2-stage
problem \eqref{bRdisc} satisfies \ref{p1} and \ref{p2} 
with inflation parameter $\Ld=\ld$. 
So Theorem~\ref{ccpsaa} applied to \eqref{bRdisc} 
suggests an improved $\poly\bigl(\inpsize,\frac{\ld}{\ve}\bigr)$ sample complexity,
but two sources of difficulty remain. 

First, 
while we would like to consider the proxy problem (\bRdisc[\ppoly]), which is the SAA
version of \eqref{bRdisc}, 
we are in fact solving the true) SAA problem (\Qdisc[\ppoly]) 
approximately.
Whereas $\hpcentral[p]{x}$ and $\bhpcentral[p]{x}+\zpcentlng[p]{0}$ are pointwise close,
$\zpcentlng[p]{0}$ could be significant compared to $\zpcentral[p]{x}$ (as indicated by the
factor-$2$ loss in Lemma~\ref{zproxy}).
Therefore, 
an $\al$-approximation to (\Qdisc[\ppoly]) 
does not yield an $O(\al)$-approximation to (\bQdisc[\ppoly]) (or equivalently,
(\bRdisc[\ppoly])).  
We will in fact {\em not} be able to obtain an approximate solution to
(\bRdisc[\ppoly]),
and so it is unclear why transferring approximation guarantees from 
(\bRdisc[\ppoly]) to \eqref{bRdisc} (and hence \eqref{bQdisc}) is helpful. 
That is, the artifact we encounter is that the 2-stage SAA problem that has bounded
inflation factor is {\em not} the one that we are able to approximate.
(Note that Theorem~\ref{ccpsaa} is not equipped to deal with this issue since its starting
point is an approximate solution to the SAA problem.)

The way around this is to realize
that 
our goal is to evaluate the quality of the SAA solution for the 
original problem \eqref{Qdisc}, and not \eqref{bRdisc}. 
In 2-stage stochastic optimization, the contribution $f_h(p)$ from high-cost scenarios to
the total expected cost is linear in $p$, 
which provides a handle on how to relate $f_h(\pcent)$ and $f_h(\ppoly)$. 
In our case, 
the contribution $\zpcentlng[p]{0}$ is nonlinear in $p$, and we need to derive new 
insights to reason about how this changes when we move from
$\pcent$ to its empirical estimate $\htpcent$; we then proceed by carefully adapting the 
ideas in~\cite{CharikarCP05}. We explain this in more detail under ``Overview'' in
Appendix~\ref{mainsaaproof}. 

Second, we (still) do not have an approximate value oracle for 
$\hpcentral[\htpcent]{x,y}$ (or $\bhpcentral[\htpcent]{x,y}$). 
However, we will show in Section~\ref{polysolve} (see Lemma~\ref{gxytoapxsub})
that if we have the non-standard type of approximation for $g(x,y,A)$ mentioned in
Theorem~\ref{intromainthm}, 
then one can obtain an approximate value oracle for
$\hpcentral[\htpcent]{x}$. 
While this is not the same as a value oracle for $\hpcentral[\htpcent]{x,y}$, we show that 
this nevertheless suffices. 
 
Combining these ingredients yields the following theorem,
which is the main result of this section. Recall that $\iopt:=\min_{x\in X}\hpcentral{x}$,
and $\log|X|$ and $\log\tau$ are $\poly(\inpsize)$.

\begin{theorem} \label{mainsaathm}
Let $\ve\leq\frac{1}{3}$, $\kp>0$.
Consider $k=\frac{2}{\ve}\log\bigl(\frac{1}{\dt}\bigr)$ SAA problems with
objective functions 
$\hpcentral[\htpcent^i]{x}:=c^\T x+\zpcentral[\htpcent^i]{x}$,
for $i=1,\ldots,k$, where each $\htpcent^i$ is an empirical estimate of $\pcent$
constructed using 
$N=\poly(\frac{\ld}{\ve},\log|X|,\log(\frac{\tau}{\kp}),\log(\frac{1}{\dt})\bigr)$
independent samples. Suppose that for every $i=1,\ldots,k$, we have 
a solution $\hx^i\in X$ and an estimate $f^i$, 
such that: 
(S1) $\hpcentral[\htpcent^i]{\hx^i}\leq\beta f^i$; and
(S2) $f^i\leq\rho\cdot\min_{x\in X}\hpcentral[\htpcent^i]{x}$ (where $\beta,\rho\geq 1$). 
Let $j=\argmin_{i=1,\ldots, k}f^i$ and $\hx=\hx^j$. 
Then, \mbox{$\hpcentral{\hx}\leq 4\beta\rho\bigl(1+O(\ve)\bigr)\iopt+2\beta\rho\kp$}
with probability at least $1-3\dt$. 
\end{theorem}

The mixed (i.e., multiplicative + additive) guarantee obtained above can be turned into a 
purely multiplicative guarantee if we have a lower bound $\lb$ on $\iopt$
with $\log\bigl(\frac{1}{\lb}\bigr)=\poly(\inpsize)$.
We show 
that such a lower bound can indeed be obtained under some very mild assumptions
(Lemma~\ref{lbnd}). 

The proof of Theorem~\ref{mainsaathm} is further complicated due to
the peculiarities of the estimates 
$f^i$ that we have for $\hpcentral[\htpcent^i]{\hx^i}$. 
Note that (S1), (S2) only imply that $\hx^i$ is a $\beta\rho$-approximation to the SAA
problem, and $\hpcentral[\htpcent^i]{\hx^i}\in[f^i/\rho,\beta f^i]$, 
so a statement of the form in Theorem~\ref{ccpsaa} 
would yield an inferior approximation bound of $O(\beta^2\rho^2+\ve)$.
Instead, we need to adapt the arguments of \cite{CharikarCP05} to suit the numerous
peculiarities of our setting. The proof is therefore somewhat technical and
we defer this to Appendix~\ref{mainsaaproof}.  

\medskip
We remark that the proxy problem \eqref{bQdisc} (or \eqref{bRdisc}) is used 
{\em only in the analysis}. One takeaway  
here is
that we derive a {\em substantially improved sample-complexity bound by taking a slight
hit in the approximation ratio} when moving from the SAA to the original problem. 
This is a novel, nuanced result regarding the effectiveness of the SAA method for
\dr 2-stage problems. 
We do not know of any other
setting where one obtains drastically improved sample complexity by settling
for a worse than $(1+\ve)$-factor (but still $O(1)$) loss when moving from the SAA
to the original problem. (In particular, no such result is known for standard 2-stage problems.)

\subsection{Solving distributionally robust problems for polynomial-support central distributions} 
\label{polysolve} 
We now show how to approximately solve the distributionally robust problem \eqref{Qdiscpoly}
when the central distribution $\ppoly$ has polynomial-size. 
This will allow us to solve the SAA problem(s) constructed in Section~\ref{saa}, and
complement Theorem~\ref{mainsaathm}. 
Let $\Asup$ denote the support of $\ppoly$. So we have

\vspace*{-3.5ex}
\noindent
\begin{minipage}[t]{0.35\textwidth}
\leqnomode
\begin{equation}
\tag{\Qdiscpoly} \label{Qdiscpoly}
\min_{x \in X}\ \hppoly{x} := c^\T x + \zppoly{x}, 
\end{equation}
\reqnomode
\end{minipage}
\begin{minipage}[t]{0.65\textwidth}
\begin{alignat*}{2}
\text{where}\ \zppoly{x}\ :=\ 
\max & \ \ & \sum_{(A,A')\in\Asup\times\A}\gm_{A,A'}&g(x,A') \tag{T${}_{\ppoly,x}$} \label{zpolydefn} \\
\text{s.t.} && \sum_{A'} \gm_{A,A'} & \leq \ppoly_A \quad \ \forall A\in\Asup \\
&& \sum_{A,A'}\scm(A,A')&\gm_{A,A'} \leq r \\[-2ex]
&& \gm & \geq 0. 
\end{alignat*}
\end{minipage}

\medskip
We consider the fractional relaxation of \eqref{Qdiscpoly}, where we replace $X$ with
its relaxation $\Pc$ to obtain \mbox{(\Qppoly):\ $\min_{x\in\Pc}\hppoly{x}$.} 
As noted earlier, unlike the case 
with 2-stage \{stochastic, robust\} optimization, 
where the fractional relaxation of the polynomial-scenario problem gives a
polynomial-size LP and is therefore straightforward to solve in polytime, 
it is substantially more challenging 
to even approximately solve the fractional \dr polynomial-scenario problem. 
In particular, reformulating $\zppoly{x}$ (and hence (\Qppoly)) as a minimization LP leads
to an LP with exponential number of constraints and variables.  
The issue is that 
\eqref{zpolydefn} involves an {\em exponential} number of
$\gm_{A,A'}$ variables.  
So if we reformulate $\zppoly{x}$ as a minimization LP by taking the dual of \eqref{zpolydefn}
(and replacing $g(x,A')$ by its LP formulation), we obtain an exponential number of constraints
(due to the $\gm_{A,A'}$ variables), and an exponential number of variables (needed 
to encode the LP for $g(x,A')$, for each $A'\in\A$).
(An exception to all this is the \allsets setting (i.e., $\A=2^U$ for some set $U$) with
the discrete scenario metric $\scm$ (so $\LW$ is the $\frac{1}{2}L_1$-metric),
under the assumption that $g(x,A)\leq g(x,A')$ for all $x$,
$A\sse A'$, which holds for covering problems. 
Here, we can reformulate $\zppoly{x}$ as a polynomial-size minimization LP and
hence, obtain a compact LP for 
(\Qppoly), and round 
its optimal solution using a local approximation algorithm. 
Theorem~\ref{polythm_collapsible} proves a more general result along these lines.)   

To overcome these obstacles, we work with the {\em convex program} given by
(\Qppoly). Recall that $g(x,y,A):=\max_{A'\in\A}\bigl(g(x,A')-y\cdot\scm(A,A')\bigr)$, 
where $x\in\Pc$, $y\geq 0$, and $A\in\A$.
We show that the complexity of solving \eqref{Qdiscpoly} is tied to the problem of
finding a near-optimal solution to $g(x,y,A)$. However, as noted earlier, under the
standard notion of approximation, it is impossible to obtain any approximation guarantee 
due to the mixed-sign objective in $g(x,y,A)$ (see Theorem~\ref{saainapx}). To evade this
difficulty, we consider the following non-standard notion of approximation for $g(x,y,A)$.

\begin{definition} \label{gxyapx}
We say that $\alg$ is a {\em $(\beta_1,\beta_2)$-approximation algorithm} for $g(x,y,A)$, where
$\beta_1,\beta_2\geq 1$, if it returns a scenario $\bA\in\A$ such that 
\nolinebreak
\mbox{$g(x,\bA)-y\cdot\scm(A,\bA)\geq\frac{g(x,A')}{\beta_1}-\beta_2\cdot y\cdot\scm(A,A')$ for all
$A'\in\A$.} 
\end{definition}

Recall that a {\em local $\rho$-approximation} for (\Qppoly) 
is an algorithm that given $x\in\Pc$, returns an integral solution $\tx\in X$ and integral 
recourse actions $\tz^A$ for every $A\in\A$ (implicitly), such that $c^\T\tx\leq\rho(c^\T x)$
and $\text{(cost of $\tz^A$)}\leq\rho g(x,A)$ for all $A\in\A$. 
The main result of this section, which is used to interface with 
Theorem~\ref{mainsaathm}, is as follows. 

\begin{theorem} \label{polythm}
Suppose that we have a polytime separation oracle for $\Pc$, a local
$\rho$-approximation algorithm for (\Qppoly), and a $(\beta_1,\beta_2)$-approximation algorithm for
$g(x,y,A)$ for any $(x,y,A)\in X\times\Rplus\times\A$. 
For any $\ve>0$, in $\poly\bigl(\inpsize,\log(\frac{1}{\ve})\bigr)$ time, we can compute
$\tx\in X$ and an estimate $\tf$ of $\hppoly{\tx}$ such that: 
$\tf\leq\hppoly{\tx}\leq\beta_1\beta_2\cdot\tf$, and
$\tf\leq\rho(1+\ve)\cdot\min_{x\in X}\hppoly{x}$. 
\end{theorem}

We prove the above theorem by utilizing the ellipsoid method. For this, we need to be able
to compute a subgradient of the 
objective function $\hppoly{x}$. Shmoys and Swamy~\cite{ShmoysS06} showed that it suffices
to have $\w$-subgradients (Definition~\ref{apsgrad}). 
We show that a near-optimal solution to \eqref{zpolydefn} yields an
approximate subgradient of $\hppoly{x}$ (Lemma~\ref{apxsub}), and we can obtain such a
solution to \eqref{zpolydefn} using a $(\beta_1,\beta_2)$-approximation to $g(x,y,A)$
(Lemma~\ref{gxytoapxsub}). 
Recall from properties \ref{p4}, \ref{p5} that for every $A\in\A$, the function
$g(\bullet,A)$ is convex, and at every $x\in\Pc$, 
$A\in\A$, we can efficiently compute $g(x,A)$, and a subgradient $\sgr^{x,A}$ with 
$\|\sgr^{x,A}\|\leq K$, where $\ln K=\poly(\inpsize)$.
The proof of Lemma~\ref{gxytoapxsub} appears after the proof of Theorem~\ref{polythm},
right before Section~\ref{hardness}.

\begin{lemma} \label{apxsub} \label{lem:sub1}
Let $x\in\Pc$, and $\gamma$ be a $\beta$-approximate solution to \eqref{zpolydefn}. Then
$d := c + \sum_{(A,A')\in\Asup\times\A}\gamma_{A,A'}\sgr^{x,A'}$
is a $\bigl(1 - \frac{1}{\beta}\bigr)$-subgradient of $\hppoly{.}$ at $x$.
\end{lemma}

\begin{proof} 
Consider any $x'\in\Pc$. Since $\gm$ is a feasible solution to (\zpolylp{x'}),
we have $\hppoly{x'}\geq c^\T x'+\sum_{(A,A')\in\Asup\times\A}\gm_{A,A'}g(x',A')$.
Let $\gm^*$ be an optimal solution to \eqref{zpolydefn}.
Since $\gm$ is a $\beta$-approximate solution to \eqref{zpolydefn}, we have
\begin{equation*}
\begin{split}
\zppoly{x}=\sum_{(A,A')\in\Asup\times\A}\gm^*_{A,A'}g(x,A')
&=\tfrac{1}{\beta}\cdot\quad\sum_{\mathclap{(A,A')\in\Asup\times\A}}\ \gm^*_{A,A'}g(x,A')
+\Bigl(1-\tfrac{1}{\beta}\Bigr)\cdot\quad\sum_{\mathclap{(A,A')\in\Asup\times\A}}\ \gm^*_{A,A'}g(x,A') \\
&\leq\sum_{\mathclap{(A,A')\in\Asup\times\A}}\ \gm_{A,A'}g(x,A')
+\Bigl(1-\tfrac{1}{\beta}\Bigr)\cdot\quad\sum_{\mathclap{(A,A')\in\Asup\times\A}}\ \gm^*_{A,A'}g(x,A').
\end{split}
\end{equation*}
Therefore,
\begin{equation*}
\begin{split}
\hppoly{x'}-\hppoly{x}
& \geq c^\T(x'-x)+\quad\sum_{\mathclap{(A,A')\in\Asup\times\A}}\ \gm_{A,A'}\Bigl(g(x',A')-g(x,A')\Bigr)
-\Bigl(1-\tfrac{1}{\beta}\Bigr)\cdot\quad\sum_{\mathclap{(A,A')\in\Asup\times\A}}\ \gm^*_{A,A'}g(x,A') \\
& \geq c^\T(x'-x)+\quad\sum_{\mathclap{(A,A')\in\Asup\times\A}}\ \gm_{A,A'}\sgr^{x,A'}\cdot (x'-x)
-\Bigl(1-\tfrac{1}{\beta}\Bigr)\hppoly{x} \\
&=d^\T(x'-x)-\Bigl(1-\tfrac{1}{\beta}\Bigr)\hppoly{x}. 
\end{split} 
\end{equation*}
The second inequality follows since $\sgr^{x,A'}$ is a subgradient of $g(\cdot,A')$ at $x$.
\end{proof}

\begin{lemma} \label{lem:beta-dual} \label{gxytoapxsub}
Let $x\in\Pc$. Suppose we have a $(\beta_1,\beta_2)$-approximation algorithm for
$g(x,y,A)$ for all $y\geq 0$ and all $A\in\A$. 
Then, (i) we can compute a
$\beta_1\beta_2$-approximate solution $\gamma$ to \eqref{zpolydefn}; 
(ii) hence, $f=c^\T x+\sum_{(A,A')\in\Asup\times\A}\gm_{A,A'}g(x,A')$ satisfies
$f\leq\hppoly{x}\leq\beta_1\beta_2\cdot f$.
\end{lemma}

The ellipsoid-based algorithm in~\cite{ShmoysS06} (and for convex 
optimization in general) has two phases: one where we use approximate
subgradients to obtain a polynomial number of feasible points such that at least one of
them is a near-optimal solution, and the other, where we choose the best among these
feasible points. 
In the first phase, starting with an ellipsoid that contains the entire feasible region, 
at each step, we add a cut (i.e., a hyperplane) passing through the center $\bx$ of the
current ellipsoid to chop off a half-ellipsoid that does not contain points of interest.  
If $\bx$ is infeasible, we use a violated inequality to obtain such a cut. Otherwise, we 
find an $\w$-subgradient $\hsgr$ of $\hppoly{\bullet}$ at $\bx$ and use the cut
$\hsgr^\T(y-\bx)\leq 0$; the definition of $\w$-subgradient ensures that any point $y$
discarded by this cut has $\hppoly{y}\geq(1-\w)\hppoly{x}$. We continue this until the
volume of the current ellipsoid becomes sufficiently small, which happens after a
polynomial number of iterations.
The first phase can be executed using $\w$-subgradients, for an
arbitrary $\w$. Shmoys and Swamy~\cite{ShmoysS06} showed that the second phase can be
implemented even without having an (approximate) objective-function oracle (which can be
hard to obtain with exponentially many scenarios) 
provided that we have $\w$-subgradients for sufficiently small $\w$ ($=1/\poly(\inpsize)$).  

Computing $\w$-subgradients efficiently for such small $\w$ would require an FPTAS for
\eqref{zpolydefn}. But, in general, the optimization problems $g(x,y,A)$ and
\eqref{zpolydefn} are complicated problems that can capture the \apx-hard 
\kmaxmin problem---$\max_{A\sse U:|A|\leq k} g(x,A)$---%
encountered in 2-stage robust optimization~\cite{FeigeJMM05,GuptaNR,Khandekar} 
(see Theorem~\ref{saainapx}). 
rules out an FPTAS for \eqref{zpolydefn};
moreover, the approximation we can obtain for $g(x,y,A)$ will naturally depend on the
application.  
We sidestep this difficulty by noting that Lemma~\ref{gxytoapxsub} (ii) 
gives a $\beta_1\beta_2$-approximate value oracle for 
$\hppoly{x}$, which can be used to implement the second phase. 

A final difficulty that remains is that for our applications (see Section~\ref{apps}), we
will only be able to approximate $g(x,y,A)$ for integral $x$ (as is the case with robust \kmaxmin
problems); indeed Theorem~\ref{polythm} only assumes that we have an approximation
algorithm for computing $g(x,y,A)$ when $x\in X=\Pc\cap\Z^m$.
However, we need to 
add an $\w$-subgradient cut passing through the center $\bx$ of our current ellipsoid,
which will typically not be integral; so we will not be able to use
Lemmas~\ref{gxytoapxsub} and~\ref{apxsub} to obtain an $\w$-subgradient at $\bx$.
To bypass this difficulty, we use the unorthodox approach of generating a cut from a point  
{\em different} from the ellipsoid-center $\bx$. 
We round $\bx$ to $\tx\in X$ using our local
approximation algorithm, and use Lemma~\ref{apxsub} at $\bx$, but with an approximate
solution to (\zpolylp{\tx}) (obtained by approximating $g(\tx,y,A)$), 
to compute a vector $\sgr$; 
we add the cut $\sgr^\T(x-\bx)\leq 0$. While $\sgr$ need not
be an $\w$-subgradient at $\bx$, we argue that this cut is still valid, in that any
point $x'$ cut off by the inequality has $\hppoly{x'}$ large compared to $\hppoly{\tx}$.

\begin{lemma} \label{validcut}
Let $\bx\in\Pc$ and $\tx\in X$ be obtained by rounding $\bx$ using a local
$\rho$-approximation algorithm. Let $\gm$ be a $\beta$-approximate 
solution to $\mathrm{(T_{\ppoly,\tx})}$, 
and let $\tsgr=c+\sum_{(A,A')\in\Asup\times\A}\gm_{A,A'}\sgr^{\bx,A'}$.
If $x'\in\Pc$ is such that $\tsgr^\T(x'-\bx)\geq 0$, then 
$\hppoly{x'}\geq
\frac{1}{\rho}\cdot\bigl(c^\T\tx+\sum_{(A,A')\in\Asup\times\A}\gm_{A,A'}g(\tx,A')\bigr)
\geq\frac{1}{\beta\rho}\cdot\hppoly{\tx}$.
\end{lemma}

\begin{proof} 
Define $f(x):=c^\T x+\sum_{(A,A')\in\Asup\times\A}\gm_{A,A'}g(x,A')$ for all
$x\in\Pc$. Clearly, $f(x)\leq\hppoly{x}$ for all $x\in\Pc$. 
Also, since we use a local approximation algorithm to obtain $\tx$, we have 
$f(\bx)\geq f(\tx)/\rho$.
By mimicking the proof of Lemma~\ref{apxsub}, we have that
$c+\sum_{(A,A')\in\Asup\times\A}\gm_{A,A'}\sgr^{x,A'}$ is a 
subgradient of $f$ at $x$.
We have $\hppoly{x'}-f(\bx)\geq f(x')-f(\bx)\geq\tsgr^\T(x'-\bx)\geq 0$. So
$\hppoly{x'}\geq f(\bx)\geq f(\tx)/\rho$. Finally, $f(\tx)\geq\hppoly{\tx}/{\beta}$ by
Lemma~\ref{gxytoapxsub} (ii).
\end{proof}

We describe below the algorithm \polyalg leading to Theorem~\ref{polythm}. 
By \ref{p3}, $\Pc\sse B(\bo,R)$, and contains a ball of
radius $V\leq 1$, where $\ln\bigl(\frac{R}{V}\bigr)$, $\ln K$ are $\poly(\inpsize)$. 
Lemma~\ref{apxsub} implies that the Lipschitz constant of $\hppoly{.}$ is at most
$K':=\|c\|+K$, so $\ln K'=\poly(\inpsize)$.
To utilize \polyalg to obtain Theorem~\ref{polythm}, we require a lower bound $\lb$ on
$\polyopt:=\min_{x\in X}\hppoly{x}$ with $\log\bigl(\frac{1}{\lb}\bigr)=\poly(\inpsize)$.  
Under a standard, rather mild assumption (that originated in~\cite{ShmoysS06}), we argue
that we can either compute such a lower bound, or determine that $x=\bo$ is an optimal
solution (Lemma~\ref{lbnd}), and show that this suffices. 
Call a scenario $A$ a ``null 
scenario'' if $g(x,A)=g(0,A)$ for all $x\in\Pc$ (e.g., $A=\es$ in \drssc). 
We assume that in every non-null scenario $A$, we have
$c^\T x+g(x,A)\geq 1$ for all $x\in\Pc$. 
We assume that we are given $\scmmax=\max_{A,A'}\scm(A,A')$ (or an upper bound on it) in
the input.

{\small \vspace{5pt} \hrule 

\begin{ppolyalg}

Require: separation oracle for $\Pc$, local $\rho$-approximation algorithm $\B$, 
and a $(\beta_1,\beta_2)$-approximation algorithm $\alg$ for $g(x,y,A)$ for all 
$(x,y,A)\in X\times\Rplus\times\A$. 

\noindent
Output: $\tx\in X$ and $\tf$ satisfying: 
$\tf\leq\hppoly{\tx}\leq\beta_1\beta_2\tf$, and 
$\tf\leq\rho\bigl(\min_{x\in X}\hppoly{x}+\kp\bigr)$.

\begin{enumerate}[label=A\arabic*., topsep=0.5ex, itemsep=0ex, labelwidth=\widthof{A3.},
    leftmargin=!]
\item Set $k\assign 0,\ \bx_0\assign \bo,\ 
\mu\assign\min\bigl\{1,\frac{\kp}{2K'R}\bigr\},\ 
N\assign\lceil 2m^2\ln\bigl(\frac{2R}{\mu V}\bigr)\rceil$. 
Let $E_0\assign B(\bo,R)$ and  $\Pc_0\assign\Pc$. 

\item For $i=0,\ldots, N$ do the following. \mbox{(We maintain that $E_i$ is an 
ellipsoid centered at $\bx_i$ containing $\Pc_k$.)} 
\begin{enumerate}[label=\alph*), ref=\theenumi\alph*), topsep=0ex, itemsep=0ex,
    labelwidth=\widthof{c)}, leftmargin=!]  
\item If $\bx_i\notin\Pc_k$, let $a^\T x\leq b$ be an inequality that is satisfied by all
$x\in\Pc_k$ but violated by $\bx_i$. (This is either obtained from a   separation oracle
for $\Pc$, or from inequalities added in prior iterations.) 
Let $H$ be the halfspace $\{x\in\R^m: a\cdot (x-\bx_i)\leq 0\}$.

\item \label{polyalgcut}
If $\bx_i\in\Pc_k$, let $\tx_k\in X$ be obtained by rounding $\bx_i$ using $\B$.
Use Lemma~\ref{gxytoapxsub} and $\alg$ to obtain a $\beta_1\beta_2$-approximate solution
$\gm$ to (\zpolylp{\tx_k}) (which has polynomial-size support). Define 
$\tsgr_k:=c+\sum_{(A,A')\in\Asup\times\A}\gm_{A,A'}\sgr^{\bx_i,A'}$, and
$\tf_k:=c^\T\tx_k+\sum_{(A,A')\in\Asup\times\A}\gm_{A,A'}g(\tx_k,A')$.  
If $\tsgr_k=\bo$, then return $\tx_k$ and $\tf_k$.
Otherwise, let $H$ denote the halfspace $\{x\in\R^m: \tsgr_k^\T (x-\bx_i)\leq 0\}$. 
Set $\Pc_{k+1}\assign\Pc_k\cap H$, and $k\assign k+1$.  

\item Set $E_{i+1}$ to be the ellipsoid of minimum volume containing the half-ellipsoid
$E_i\cap H$, and let $\bx_{i+1}$ be its center.
\end{enumerate}

\item Let $k\assign k-1$. 
Let $j=\argmin_{i=0,\ldots,k}\tf_i$. Return $\tx_j$ and $\tf_j$.
\end{enumerate} 
\end{ppolyalg}
\hrule
}

\smallskip

\begin{lemma} \label{lbnd}
Suppose that we have a $(\beta_1,\beta_2)$-approximation for $g(\bo,0,A'')$ for some
scenario $A''\in\A$.  
We can efficiently determine that either $\lb=\frac{r}{\beta_1\scmmax}$ is a lower bound
on $\min_{x\in\Pc}\hpcentral[p]{x}$ for every distribution $p$, 
or $x=\bo$ is an optimal solution to $\min_{x\in\Pc}\hpcentral[p]{x}$ for every 
distribution $p$. 
\end{lemma}

\begin{proof} 
We first show that $\min_{x\in X}\hpcentral[p]{x}\geq\frac{r}{\scmmax}$ for every $p$, if
$\A$ contains any non-null scenario. Otherwise $x=\bo$ is an optimal solution to
$\min_{x\in X}\hppoly{x}$ for every $p$. 
Note that a non-null scenario $A$ must satisfy $g(0,A)\geq 1$.

Say $A^*\in\A$ is a non-null scenario.  
Fix any $x\in X$. There is a feasible solution $\gm$ to
(\zpolylp[p]{x}) that sends at least $\frac{r}{\scmmax}$ flow to $A^*$, i.e.,
$\sum_{A\in\Asup}\gm_{A,A^*}\geq\frac{r}{\scmmax}$. 
So $\zpcentral[p]{x}\geq\frac{r}{\scmmax}\cdot g(x,A^*)$, and
so $\hpcentral[p]{x}\geq\frac{r}{\scmmax}$, since $c^\T x+g(x,A^*)\geq 1$ as $A^*$ is a
non-null scenario. This holds for every $p$. 

If all scenarios in $\A$ are null scenarios, then $\zpcentral[p]{x}=\zpcentral[p]{\bo}$
for all $x\in X$, since $g(x,A)=g(\bo,A)$ for all $A\in\A$ and $x\in X$. Hence,
$\hpcentral[p]{\bo}=\min_{x\in X}\hpcentral[p]{x}$; again, this holds for all $p$.

We use the $(\beta_1,\beta_2)$-approximation algorithm for $g(0,0,A'')$ to obtain a
scenario $\bA\in\A$. Therefore, we have 
$g(0,\bA)\geq\frac{1}{\beta_1}\cdot\max_{A\in\A}g(0,A)$. So if $g(0,\bA)<\frac{1}{\beta_1}$,
then $g(0,A)<1$ for all $A\in\A$, which means that all scenarios in $\A$ are null
scenarios, and we return $x=\bo$ as an optimal solution. Otherwise, we return the lower
bound $\lb$. To see why $\lb$ is a valid lower bound, when
$g(0,\bA)\geq\frac{1}{\beta_1}$, note that there are two cases. If $\A$ contains a
non-null scenario then we have established that $\frac{r}{\scmmax}$ is a lower bound.
Otherwise, we have established that $x=\bo$ is an optimal solution; there is a feasible
solution to (\zpolylp[p]{\bo}) that sends at least $\frac{r}{\scmmax}$ to $\bA$, so 
$\hpcentral[p]{\bo}\geq g(0,\bA)\cdot\frac{r}{\scmmax}=\frac{r}{\beta_1\scmmax}$.
\end{proof}

\begin{proofof}{Theorem~\ref{polythm}}
We first apply Lemma~\ref{lbnd} to either determine that $x=\bo$ is an optimal solution,
or obtain a lower bound $\lb=\frac{r}{\beta_1\scmmax}$ on $\polyopt$.
If Lemma~\ref{lbnd} returns $x=\bo$ as an optimal solution, then we
use Lemma~\ref{gxytoapxsub} and $\alg$ to obtain a $\beta_1\beta_2$-approximate solution
$\gm$ to (\zpolylp{0}). 
We return $x=0$ as the optimal solution, and $\sum_{(A,A')\in\Asup\times\A}\gm_{A,A'}g(0,A')$
as an estimate of $\hppoly{\bo}$, which is a suitable estimate due to
Lemma~\ref{gxytoapxsub} (ii). 

So suppose Lemma~\ref{lbnd} returns the lower bound $\lb$. 
We run Algorithm \polyalg with $\kp=\ve\lb$. 
By Lemma~\ref{gxytoapxsub} (ii), we immediately obtain
that $\tf_l\leq\hppoly{\tx_l}\leq\beta_1\beta_2\cdot\tf_l$ for all $l=1,\ldots,k$.

We re-work the arguments in Lemma 4.5 from~\cite{ShmoysS06}.
For $S\sse\R^m$, let $\vol(S)$ denote the volume of $S$. Let $\vol_m$ denote the volume of
the unit ball (in the $L_2$-norm) in $\R^m$. It is well known that
$\frac{\vol(E_{i+1})}{\vol(E_i)}\leq e^{-1/2m}$ for every $i=0,\ldots,N$ (see, e.g.,
\cite{GrotschelLS88}). 

Let $\sx\in X$ be an optimal solution to $\min_{x\in X}\hppoly{x}$. 
Recall that $\mu=\min\bigl\{1,\frac{\kp}{2K'R}\bigr\}$. 
If $\tsgr_l\cdot(\sx-\bx_l)\geq 0$ for some $l$ (this includes the case when
$\tsgr_l=\bo$), then Lemma~\ref{validcut} shows that 
$\tf_l\leq\rho\cdot\hppoly{\sx}$. 
Otherwise consider the affine transformation
$T$ defined by $T(x)=\mu I_m(x-\sx)+\sx=\mu x+(1-\mu)\sx$ where $I_m$ is the $m\times m$ 
identity matrix, and let $W=T(\Pc)$, so $W\sse\Pc$ is a shrunken version of $\Pc$.
By properties of affine transformations, we have 
$\vol(W)=\mu^m\vol(\Pc)\geq(\mu V)^m\vol_m$, where the last inequality follows since $\Pc$
contains a ball of radius $V$.  
For any $x'=T(x)\in W$, we have $\|x'-\sx\|=\mu\|x-\sx\|\leq\frac{\kp}{K'}$ since
$x,\sx\in B(\bo,R)$; so $\hppoly{x'}\leq\hppoly{\sx}+\kp$ since $\hppoly{\cdot}$ has
Lipschitz constant at most $K'$.
The volume of the ball $E_0=B(\bo,R)$ is $R^m\vol_m$. Therefore,
$$
\vol(\Pc_k)\leq\vol(E_N)\leq e^{-N/(2m)}\,\vol(E_0)\leq 
\Bigl(\tfrac{\mu V}{2}\Bigr)^m\vol_m<\vol(W).$$ 
So there must be a point $x'\in W$ that lies on a boundary of $\Pc_k$ generated by a
hyperplane $\hsgr_l\cdot(x-\bx_l)=0$. This implies (by Lemma~\ref{validcut}) that 
$$
\tf_l\leq\rho\cdot\hppoly{x'}\leq\rho\bigl(\hppoly{\sx}+\kp\bigr)
\leq\rho\cdot\hppoly{\sx}+\rho\ve\cdot\lb
\leq\rho(1+\ve)\hppoly{\sx}
$$
where the last inequality follows since $\lb$ is a lower bound on $\polyopt$.
\end{proofof}

\begin{proofof}{Lemma~\ref{gxytoapxsub}}
Part (ii) follows immediately from part (i) and the definition of $\hppoly{x}$. 
We focus on proving part (i).
We consider the dual of \eqref{zpolydefn}, and show that a
$(\beta_1,\beta_2)$-approximation algorithm $\alg$ for $g(x,y,A)$ yields an 
{\em approximate} separation oracle for the dual. The dual of \eqref{zpolydefn} is as
follows. 
\begin{alignat}{2}
\min & \quad & \sum_{A\in\Asup}\ppoly_A&\tht_A+ry \tag{D} \label{dual} \\
\text{s.t.} && \tht_A & \geq g(x,A')-y\cdot\scm(A,A') \qquad \forall A\in\Asup, A'\in\A 
\label{dgxy} \\
&& \tht, y & \geq 0. \label{dnonneg}
\end{alignat}
Notice that \eqref{dual} is an LP (since $x$ is fixed) with an exponential number of
constraints, but a polynomial number of variables.
It is evident that $\alg$ yields some type of approximate separation oracle for
\eqref{dual}. Using a standard technique in approximation algorithms, we prove
that \eqref{dual}, and the primal \eqref{zpolydefn}, can be solved approximately (see,
e.g.,~\cite{JainMS03,FriggstadS15}). 

Define $\Qc(\nu):=\{(\tht,y): \eqref{dgxy},\ \eqref{dnonneg},\ \sum_{A\in\Asup}\ppoly_A\tht_A+ry\leq\nu\}$. 
Note that $\OPT_{\ref{dual}}$ is the smallest $\nu$ such that $\Qc(\nu)\neq\es$. 
We use 
$\alg$ to give an approximate separation oracle in the following sense. Given $\nu, (\tht,y)$,
we either show that $(\beta_1\tht,\beta_1\beta_2 y)\in\Qc(\beta_1\beta_2\nu)$, or we
exhibit a hyperplane separating $(\tht,y)$ from $\Qc(\nu)$. Thus, for a fixed $\nu$, in
polynomial time, the ellipsoid method either certifies that $\Qc(\nu)=\es$, or returns a
point $(\tht,y)$ with $(\beta_1\tht,\beta_1\beta_2 y)\in\Qc(\beta_1\beta_2\nu)$. The
approximate separation oracle proceeds as follows. 
We first check if $\sum_{A\in\Asup}\ppoly_A\tht_A+ry\leq\nu$ and \eqref{dnonneg} hold, and if
not, use the appropriate inequality as the separating hyperplane. 
Next, for every $A\in\Asup$, we run $\alg$ for the point $(x,y,A)$. If in
this process, we ever obtain a scenario $\bA$ such that $g(x,\bA)-y\cdot\scm(A,\bA)>\tht_A$
then we return $\tht_A\geq g(x,\bA)-y\cdot\scm(A,\bA)$ as the separating
hyperplane. Otherwise, for all $A\in\Asup$ and $A'\in\A$, we have 
$$
\tht_A\geq g(x,\bA)-y\cdot\scm(A,\bA)\geq g(x,A')/\beta_1-\beta_2y\cdot\scm(A,A').
$$
This implies that $(\beta_1\tht,\beta_1\beta_2 y)\in\Qc(\beta_1\beta_2\nu)$. 

It is easy to find an upper bound $\ub$ with $\log\ub$ polynomially bounded such that
$\Qc(\ub)\neq\es$. For a given $\e>0$, we use binary search in $[0,\ub]$ to find $\nu^*$
such that the ellipsoid method when run for $\nu^*$ (with the above separation oracle),
returns a solution $(\tht^*,y^*)$ with 
$(\beta_1\tht^*,\beta_1\beta_2 y^*)\in\Qc(\beta_1\beta_2\nu^*)$, and when run for
$\nu^*-\e$ 
certifies that $\Qc(\nu^*-\e)=\es$. 
So $\OPT_{\ref{dual}}\leq\beta_1\beta_2\nu^*$. 
For $\nu^*-\e$, we obtain a polynomial-size certificate for the emptiness of
$\Qc(\nu^*-\e)$. This consists of the polynomially many violated
inequalities returned by the separation oracle during the execution of the ellipsoid
method, and the inequality $\sum_{A\in\Asup}\ppoly_A\tht_A+ry\leq\nu^*-\e$.
By duality (or Farkas' lemma), this means that if we restrict \eqref{zpolydefn} to only 
use the $\gm_{A,A'}$ variables corresponding to (the polynomially-many) violated inequalities of type \eqref{dgxy}
returned during the execution of the ellipsoid method,  we can obtain a polynomial-size
feasible solution $\overline{\gm}$ to \eqref{zpolydefn} whose value is at least
$\nu^*-\e$. If we take $\e$ to be $1/\exp(\inpsize)$ (so the binary search still takes
polynomial time),
this also implies that
$\overline{\gm}$ has value at least $\nu^*\geq\OPT_{\ref{dual}}/(\beta_1\beta_2)$.
\end{proofof}

\subsubsection{Hardness results for the SAA problem} \label{hardness}

First, observe that 
for the \dr 2-stage problem $\min_{x\in X}\hppoly{x}$, where $\ppoly$ has 
{\em polynomial-size} support, 
if we set $r=\scmmax$, then $\zppoly{\bo}=\max_{A\in\A}g(\bo,A)$, 
so that computing $\zppoly{\bo}$ is equivalent to the \maxmin problem $\max_{A\in\A}g(\bo,A)$. 

\begin{theorem} 
\label{saainapx} 
Consider the \dr 2-stage problem $\min_{x\in X}\hppoly{x}$, where 
the support of $\ppoly$ is a {\em polynomial-size} subset of 
$\A_{\leq k}=\{A\sse U: |A|\leq k\}$.
Consider the following two settings. 
\begin{enumerate}[label=(B\arabic*), topsep=0ex, itemsep=0.5ex, labelwidth=\widthof{(B2)},
    leftmargin=!] 
\item the \kbounded setting with the $\frac{1}{2}L_1$ metric;
\item the \allsets setting 
with scenario metric $\scm$ given by: $\scm(A,A)=0$ for all $A\in\A$; 
for $A\neq A'\in\A$, we have $\scm(A,A')=1$ if $|A|,|A'|\leq k$, and $Z$ otherwise, 
where $\frac{Z}{2}$ is an upper bound on $g(\bo,U)$.
\end{enumerate}
Assume that $g(\bo,\es)=0$,  
the \kmaxmin problem $(\Pi):\ \max_{A\sse U:|A|\leq k} g(\bo,A)$, is \nphard, and the
optimum value of $(\Pi)$ is at least $1$.  
We have the following hardness results in both settings, assuming {\em P}$\neq$\np.
\begin{enumerate}[label=(\alph*), topsep=0ex, itemsep=0.5ex, labelwidth=\widthof{(c)},
    leftmargin=!] 
\item No polytime multiplicative approximation is possible for computing 
  $g(\bo,y,\es)$, given $y\geq 0$ as input.  
\item By choosing $\htp$ suitably, the hardness result in (a) carries over to the problem
  of computing $\E[A\sim\htp]{g(\bo,y,A)}$, given $y\geq 0$ as input. 
\item One can choose $r$, $\htp$ so that the problem of computing
  $\zppoly{\bo}$ is at least as hard as $(\Pi)$.
\end{enumerate}
\end{theorem}

\begin{proof} 
Part (b) follows from part (a) by simply taking $\htp$ to be
the distribution that puts a weight of $1$ on the scenario $\es$; then
$\E[A\sim\htp]{g(\bo,y,A)}=g(\bo,y,\es)$, so the hardness result in part (a) carries over. Let $A^*\in\A_{\leq k}$ be an optimal solution to $(\Pi)$, and $\optp=g(\bo,A^*)$ be its objective 
value.  

\vspace*{-1ex}
\paragraph{Part (a).}
We consider the setting (B1) first.
Clearly, $g(\bo,y,\es)$ also seeks to find an optimum of $(\Pi)$. 
By exploiting the mixed-sign objective, we can argue that any multiplicative
approximation would allow us to decide if $\optp > T$ by setting $y$ appropriately, which
is \npcomplete. More precisely, suppose we have a $\beta$-approximation algorithm for
$g(\bo,y,\es)$. Then, we can decide if $\optp > T$ for a given number $T \ge 0$ as follows. Set
$y=T$, and run the $\beta$-approximation algorithm. 
If $\optp> T$, then 
\[
g(0, y, \es) \ge g(0, A^*) - y \cdot \scm(\es, A^*) > T - T \cdot 1 = 0 ,
\]
so the approximation algorithm would return a
solution with positive value. 
If instead we have $\optp\le T$, then for every scenario $A' \in \A_{\le k}$ with $A' \neq \es$, we have $
g(0, A') - y \cdot \scm(\es, A') = g(0, A') - T \le 0 $. 
Since we also have $g(0, \es) - y \cdot \scm(\es, \es) = 0 - T \cdot 0 = 0$, we conclude that $g(0, y, \es) = 0$, and so the approximation algorithm must return a
solution with value $0$.
So we can distinguish between $\optp > T$ and $\optp \le T$.

Now consider the setting (B2). Again, suppose we are given $T \ge 0$ and we want to
decide if $\optp > T$. We may assume that $T\geq 1$, as otherwise the answer is
yes. Again take $y=T$. If $\optp > T$, then scenario $A^*$ 
satisfies $g(\bo,A^*)-y\cdot\scm(\es,A^*) > T - T \cdot 1 \ge 0$, so a multiplicative approximation for
$g(\bo,y,\es)$ must return a solution with positive objective value. If $\optp\le T$, then we
claim that $g(\bo,y,\es)=0$, and so the approximation algorithm must return a solution
with objective value 0. Thus, we can distinguish between $\optp\geq T$ and $\optp<T$. 
To prove the claim, we have $g(\bo,\es)-y\cdot\scm(\es,\es)=0$.
For every $A'\in\A_{\leq k}$, we have 
$g(\bo,A')-y\cdot\scm(\es,A')\leq T-T\cdot 1=0$. For every $A'\notin\A_{\leq k}$, we
have $g(\bo,A')-y\cdot\scm(\es,A')\leq Z/2- T Z\leq 0$. 

\vspace*{-1ex}
\paragraph{Part (c).}
For the setting (B1), we simply set $r=\scmmax$ (and $\htp$ to be arbitrary). Then, we
have $\zppoly{\bo}=\max_{A'\in\A_{\leq k}}g(\bo,A')$, which is exactly the same as problem
$(\Pi)$. 

For the setting (B2), we set $r=1$ and take $\htp$ to be the distribution that puts weight
of 1 on $\es$. We claim that $\zppoly{\bo}$ is again the same as
problem $(\Pi)$. 
Setting $\gm_{\es,A^*}=1$ and $\gm_{A,A'}=0$ everywhere else
gives a feasible solution to (\zpolylp{\bo}) of objective value $\optp$.
Let $\gm^*$ be an optimal solution to (\zpolylp{\bo}). Let $\al$ be the amount of flow
sent by $\gm^*$ on $(\es,A')$ pairs with $\scm(\es,A')=Z$. Let $\tht=\gm^*_{\es,\es}$. 
The flow on the remaining $(\es,A)$ pairs has volume $1-\al-\tht$, contributes at most
$(1-\al-\tht)\optp$ to the objective, and has $\scm$-cost $1-\al-\tht$.
So we have $\al\cdot Z+(1-\al-\tht)\leq 1$ and $\optp\leq\al\cdot\frac{Z}{2}+(1-\al-\tht)\optp$, 
which implies that $(\al+\tht)\bigl(\optp-\frac{1}{2}\bigr)\leq 0$. Since $\optp\geq 1$ by
assumption, we have that $\al+\tht=0$, and hence $\gm^*$ has objective value $\optp$.
\end{proof}

\subsubsection{Refinements: \boldmath formulating (\Qppoly) as a compact LP in special
  cases} \label{refine}
We say that the set of scenarios $\A$ is 
{\em collapsible under the scenario metric $\scm$} if for every scenario $A \in \A$,  
we can efficiently compute a polynomial-size collection of scenarios $\phi(A)$ such that
for every $x\in\Pc$, $y\geq 0$, we have
$g(x,y,A)=\max_{A'\in\phi(A)}\bigl(g(x,A')-y\cdot\scm(A,A')\bigr)$. 
For example, if $\A=2^U$ for a ground set $U$, $\scm$ is the discrete scenario metric, and 
$g(x,A)\leq g(x,A')$ for all $x$, $A\sse A'$, then $\A$ is collapsible under $\scm$ since
$g(x,y,A)$ is attained by scenarios $A$ or $U$, for all $(x,y,A)\in\Pc\times\R_+\times\A$. 
We show that if $\A$ is collapsible under $\scm$ then (\Qppoly) can be cast as a
polytime-solvable LP, 
and its optimal solution can be rounded using an algorithm that is weaker than a
local approximation algorithm. 
(Note also that in this special case,
we have a simple, application-independent polytime algorithm for computing $g(x,y,A)$
exactly.)  

A {\em restricted local $\rho$-approximation algorithm} takes as input a point 
$x\in\Pc$ {\em and a set of scenarios $\Atilde \subseteq \A$}, and returns an integral solution
$\tx\in X$ and integral recourse actions $\tz^A$ for every $A\in\Atilde$ (possibly
specified implicitly), such that $c^\T\tx\leq\rho(c^\T x)$ 
and $\text{(cost of $\tz^A$)}\leq\rho g(x,A)$ for all $A\in\Atilde$. 
(A local $\rho$-approximation algorithm is a special case of this.) 
This weaker notion will be crucial for the Steiner-tree application in Section~\ref{apps}. 

\begin{theorem} \label{polythm_collapsible}
Suppose that $\A$ is collapsible under the scenario metric $\scm$, 
and $g(x,A)$ is the optimal value of a polytime-solvable LP for all
$(x,A)\in\Pc\times\A$. 
Suppose that we have a polytime separation oracle for $\Pc$, and a restricted
local $\rho$-approximation algorithm for (\Qppoly). 
Then, in $\poly(\inpsize)$ time, we can compute: 
\begin{enumerate}[label=(\alph*), topsep=0.5ex, itemsep=0.5ex] 
\item 
an optimal solution $\bx\in\Pc$ to $\min_{x\in\Pc}\hppoly{x}$, and its objective value
$\hppoly{\bx}$;
\item 
$\tx\in X$, and its objective value $\hppoly{\tx}$, satisfying
$\hppoly{\tx}\leq\rho\cdot\min_{x\in\Pc}\hppoly{x}$.
\end{enumerate}
\end{theorem}

\begin{proof} 
We reformulate $\zppoly{x}$ as an LP. The dual of \eqref{zpolydefn} is as
follows. 
\begin{alignat}{2}
\min & \quad & \sum_{A\in\Asup}\ppoly_A&\tht_A+ry \tag{D${}_{\ppoly,x}$} \label{dual_before_collapsing} \\
\text{s.t.} && \tht_A & \geq g(x,A')-y\cdot\scm(A,A') \qquad \forall A\in\Asup, A'\in\A 
\label{dgxyagain} \\
&& \tht, y & \geq 0. \notag 
\end{alignat}

Since by assumption $\A$ is collapsible under the scenario metric $\scm$, the
exponentially many constraints in \eqref{dgxyagain} can be collapsed to the polynomially
many constraints: 
\begin{equation} 
\tht_A \geq g(x,A')-y\cdot\scm(A,A') \qquad \forall A\in\Asup, A'\in\phi(A).
\label{dgxypoly}
\end{equation}
Suppose that $g(x,A)$ is captured by the polytime-solvable LP: 
$\min\ c^A \cdot z_A\ \text{s.t.}\ (x, z_A) \in \mathcal{F}(A)$, where $\mathcal{F}(A)$ is a polytope (over which we can optimize linear functions efficiently). Then, incorporating this in the
above constraints, we obtain the following LP-formulation for $\min_{x\in\Pc}\hppoly{x}$.
\begin{alignat}{3}
\min & \quad & c^\T x+\sum_{A\in\Asup}&\ppoly_A\tht_A+ry \tag{DR-LP} \label{drlp} \\
\text{s.t.}
&& \tht_A & \geq c^{A'} \cdot z_{A'}-y\cdot\scm(A,A') \quad && \forall A\in\Asup, A' \in \phi(A) \notag \\ 
&& \quad (x, z_{A'}) & \in \mathcal{F}(A') && \forall A' \in \cup_{A \in \Asup} \phi(A) \notag \\
&& \tht, y &\geq 0, \quad x\in\Pc. \notag
\end{alignat}
Since we have polytime separation oracles for the polytopes $\Pc$ and $\{\mathcal{F}(A')\}$, we can efficiently compute an
optimal solution $\bx$ for \eqref{drlp} using the ellipsoid method. This proves part (a). 

Part (b) follows from part (a) by applying the restricted local $\rho$-approximation
algorithm with the scenario set $\Atilde := \cup_{A \in \Asup} \phi(A)$ to round $\bx$ and
obtain $\tx\in X$. As shown above, we can efficiently compute $\zppoly{\tx}$, and hence
$\hppoly{\tx}$, by solving an LP. 
Observe that if $(\theta^*, y^*)$ is an optimal solution to (D$_{\ppoly,\bx}$), then
$(\rho\theta^*,\rho y^*)$ satisfies constraints \eqref{dgxypoly},
which implies that $\zppoly{\tx}\leq\rho\cdot\zppoly{\bx}$. Since we also have $c^\intercal \tx
\le \rho c^\intercal \bx$, this implies $\hppoly{\tx}\leq\rho\cdot\hppoly{\bx}$. 
\end{proof}

\subsection{Applications to distributionally robust combinatorial optimization} 
\label{apps} 
We now 
apply our framework---i.e., Theorems~\ref{mainsaathm} and~\ref{polythm}---for handling
general \dr 2-stage problems to obtain the {\em first}
approximation guarantees for the \dr versions of various
combinatorial-optimization problems (under the Wasserstein metric) such as 
set cover, vertex cover, edge cover, facility location, and Steiner tree. 
Except for set cover, our approximation factors are within $O(1)$ factors of
the guarantees known for the deterministic counterparts of these
problems.
In order to apply Theorems~\ref{mainsaathm} and~\ref{polythm} for a specific
problem, we need to do the following. 
\begin{enumerate}[1., topsep=0.5ex, itemsep=0.25ex, parsep=0ex, labelwidth=\widthof{2.},
    leftmargin=!] 
\item Verify that properties \ref{p1}--\ref{p6} hold. This is usually quite immediate. 
\ref{p1}--\ref{p3} follow from the problem definition (in most cases $X=\{0,1\}^m$,
$\Pc=[0,1]^m$), with $\ld$ being the maximum factor by which the cost of a first-stage
action increases in the second stage. 
\ref{p4}, \ref{p5} follow from prior work~\cite{ShmoysS06,SwamyS12} as
the underlying 2-stage problem falls into the class of 2-stage programs considered therein. 
\ref{p6} 
can usually be satisfied by taking $\tau=\ub/(\min_{A,A':\scm(A,A')>0}\scm(A,A'))$,
for a suitable upper bound $\ub$ on $\max_{A\in\A}g(\bo,A)$.

\item Furnish the following algorithms.
\begin{enumerate}[(a), nosep, labelwidth=\widthof{(c)}, leftmargin=!]
\item An LP-relative $\al$-approximation algorithm for the deterministic
counterpart, so as to round $g(x,A)$ and obtain integral second-stage decisions:
we simply plug in known approximation results.

\item A local $\rho$-approximation algorithm for the 2-stage problem:
we have $\rho=2\al$ for set cover, vertex cover, and edge cover~\cite{ShmoysS06}, 
and $\rho=O(1)$ for facility location~\cite{ShmoysS06}. (For Steiner tree, 
we use Theorem~\ref{polythm_collapsible} in place of Theorem~\ref{polythm}; see below.)

\item A $(\beta_1,\beta_2)$-approximation algorithm for computing $g(x,y,A)$, where 
$(x,y,A)\in X\times\Rplus\times\A$. This is a new component that we need to devise, whose
design will depend on the scenario set $\A$ and the scenario metric $\scm$ (and of course the
underlying problem).
For various problems, we show how to obtain such an approximation by building upon results
known for \kmaxmin problems. We defer the proof of Theorem~\ref{gxyalg} to
the end of this section (Section~\ref{gxyalgproof}).

\begin{theorem} \label{gxyalg}
For the \kbounded setting with $\scm$ being the discrete metric, for any 
$(x,y,A)\in X\times\Rplus\times\A$, we can obtain 
$(\beta,1)$-approximation algorithms for computing $g(x,y,A)$, where $\beta$ is:
\linebreak
(a) $O(\log n)$ for set cover; 
(b) $\frac{2e}{e-1}$ for vertex cover; and
(c) $2$ for edge cover.
\end{theorem}  
\end{enumerate}
\end{enumerate}
Theorems~\ref{mainsaathm} and~\ref{polythm} then show that, for any $\ve>0$, we can obtain a
solution to the distributionally robust discrete 2-stage problem (i.e., integral first-
and second-stage decisions) of cost at most $4\al\rho\beta_1\beta_2\bigl(1+O(\ve)\bigr)$
times the optimum in $\poly\bigl(\inpsize,\frac{\ld}{\ve}\bigr)$ time (and hence, sample
complexity). 

In certain cases, we can obtain improved guarantees
by exploiting the fact that 
the fractional SAA problem, $\min_{x\in\Pc}\hppoly{x}$, can be solved in a better way,
without resorting to a local approximation algorithm.  
The most generic such setting is the unrestricted setting when the scenario collection
$\A=2^U$ is collapsible under the scenario metric. 
This includes the following natural choices of the scenario metric. 

\begin{lemma} \label{collapsiblelem}
Suppose that for all $x\in \Pc$, and all $A\sse A'$, we have $g(x,A)\leq g(x,A')$. Then
the collection of scenarios $\A = 2^U$ is collapsible under: 
(i) the discrete metric $\scmdisc$; and
(ii) the asymmetric metric $\scm^\asym_\infty(A,A') = \max_{j' \in A'}\dist(j',A)$, 
where $\dist$ is a metric on $U$.   
\end{lemma}

\begin{proof} 
Let $A \in \A$ be an arbitrary scenario. If $\scm$ is the discrete metric $\scmdisc$, we take
$\phi(A) := \{A, U\}$. If $\scm$ is the asymmetric metric $\scm^\asym_\infty$, we take
$\phi(A) := \left\{\{j \in U : \min_{k \in A}\dist_{kj} \le \scmbound\}: \scmbound \in \SCM
\right\}$, where $\SCM:=\{\dist_{jj'}: j,j' \in U\}$ is the set of all distances between two
elements of the ground set. Note that in both settings, if we choose an arbitrary pair
$(x, \mu) \in \Pc \times \SCM$, the collection of scenarios $\phi(A)$ contains the
(unique) maximal solution $A'$ for the constrained problem
\eqref{Was_oracle_simplified}. By the monotonicity property of the second-stage costs
$g(\cdot, \cdot)$ imposed in the lemma statement, $A'$ is optimal for
\eqref{Was_oracle_simplified}. By Lemma \ref{gxyredn}, it follows that $\phi (A)$ contains
an optimal solution for the unconstrained problem $g(x, y, A)$ for every pair $(x, y) \in
\Pc \times \Rplus$, and so $\A$ is collapsible under $\scm$. 
\end{proof}

The condition on $g$ in Lemma~\ref{collapsiblelem} holds for all our applications, since
they are covering problems. 
Thus, in the \allsets setting with Wasserstein metric corresponding to the scenario
metrics in Lemma~\ref{collapsiblelem}, 
Theorem~\ref{polythm_collapsible} combined with Theorem~\ref{mainsaathm} 
yields an improved $4\al\rho\bigl(1+O(\ve)\bigr)$-approximation, using 
a {\em restricted} local $\rho$-approximation algorithm,
a weaker requirement that is crucial for Steiner tree. 
There are other, orthogonal benefits 
that result from achieving a better approximation for the fractional SAA problem than that 
given by Theorem~\ref{polythm}.  
These require taking a different route than Theorem~\ref{mainsaathm} to transfer
approximation guarantees from the SAA problem to the original problem. We discuss these in
the context of the specific problems to which they apply. 

\subsubsection{Set cover} 
The \dr version was defined in Section~\ref{prelim}. 
Recall that an instance is given by $\bigl(U,\Sc,\{c_S,c^\two_S\}_{S\in\Sc}\bigr)$, where
$\Sc\sse 2^U$ and $c, c^\two$ denote the first- and second-stage costs respectively. 
Let $n=|U|$. We have $\al=O(\log n)$, and 
$\rho=2\al$. Different scenarios could be quite unrelated, so there does not seem 
to be a natural choice for $\scm$ other than the discrete metric $\scmdisc$; we therefore
consider the $\frac{1}{2}L_1$-metric. We can take $\tau=\sum_{S\in\Sc}c^\two_S$. 
Instantiating the above results yields an $O(\log^2n)$-approximation in the \allsets  
setting, and an $O(\log^3 n)$-approximation in the \kbounded setting (using
Theorem~\ref{gxyalg} (a)). But we can do better and improve these guarantees by an 
$O(\log n)$ factor.

By incorporating a decoupling idea of~\cite{ShmoysS06} in our ellipsoid-based algorithm
(in a manner similar to~\cite{FeigeJMM05} in their work on 2-stage robust set cover), we
can avoid the use of local approximation algorithm in Algorithm \polyalg, and instead use
a $(\beta_1,\beta_2)$-approximation algorithm for $g(x,y,A)$ more directly. 

\begin{theorem} \label{drscpolythm}
Consider the fractional SAA \drssc problem: $\min_{x\in\Pc}\hppoly{x}$.
Suppose that we have a $(\beta_1,\beta_2)$-approximation algorithm for
$g(\bo,y,A)$ for any $(y,A)\in\Rplus\times\A$. 
For any $\ve>0$, in $\poly\bigl(\inpsize,\log(\frac{1}{\ve})\bigr)$ time, we can compute 
$\bx\in\Pc$, and an estimate $\barf$ of $\hppoly{\bx}$, satisfying
$\hppoly{\bx}\leq\beta_1\beta_2\barf$ and
$\barf\leq 2(1+\ve)\cdot\min_{x\in\Pc}\hppoly{x}$. 
\end{theorem}

We complement Theorem~\ref{drscpolythm} with an analogue of Theorem~\ref{mainsaathm},
to transfer approximation guarantees from the fractional SAA problem,
$\min_{x\in\Pc}\hppoly{x}$, to the original fractional problem,
$\min_{x\in\Pc}\hpcentral{x}$. 

Note that by Lemma~\ref{lbnd}, we can find in polytime (under very mild assumptions)
a lower bound $\lb$ (independent of $p$) on the optimal value of
$\min_{x\in\Pc}\hpcentral[p]{x}$ such that
$\log\bigl(\frac{1}{\lb}\bigr)=\poly(\inpsize)$, or determine if $x=\bo$ is 
an optimal solution to $\min_{x\in\Pc}\hpcentral[p]{x}$ for every distribution
$p$. In the latter case, there is nothing to be done, so assume otherwise. 

\begin{theorem} \label{fracsaathm}
Let $\ve\leq\frac{1}{3}$, $\kp>0$.
Let \eqref{Qpcentral}: $\min_{x\in\Pc}\hpcentral{x}$, be the fractional version of 
a \dr problem satisfying properties \ref{p1}--\ref{p6}. Let $\lb$ be a lower bound 
on $\min_{x\in\Pc}\hpcentral[p]{x}$ for all $p$. Consider
$k=\frac{2}{\ve}\log\bigl(\frac{1}{\dt}\bigr)$ SAA problems with objective functions  
$\hpcentral[\htpcent^i]{x}:=c^\T x+\zpcentral[\htpcent^i]{x}$,
for $i=1,\ldots,k$, where each $\htpcent^i$ is an empirical estimate of $\pcent$
constructed using 
$N=\poly(\frac{\ld}{\ve},\log(\frac{\tau R}{V\lb}),\log(\frac{1}{\dt})\bigr)$
independent samples. Suppose that for every $i=1,\ldots,k$, we have 
a solution $\bx^i\in\Pc$ and an estimate $\barf^i$ of $\hpcentral[\htpcent^i]{\bx^i}$
satisfying $\hpcentral[\htpcent^i]{\bx^i}\leq\bbeta\cdot\barf^i$ and
$\barf^i\leq\rho\cdot\min_{x\in\Pc}\hpcentral[\htpcent^i]{x}$ (where $\bbeta,\rho\geq 1$).
Let $j=\argmin_{i=1,\ldots, k}\barf^i$ and $\bx=\bx^j$. 
Then, 
$\hpcentral{\bx}\leq 4\bbeta\rho\bigl(1+O(\ve)\bigr)\cdot\min_{x\in\Pc}\hpcentral{x}$ 
with probability at least $1-3\dt$. 
\end{theorem}

Before proving Theorems~\ref{drscpolythm} and~\ref{fracsaathm}, we state the results that
follow from these (and other prior results).
Combining Theorems~\ref{polythm_collapsible} (a) and~\ref{fracsaathm}, and a local
$\rho$-approximation algorithm (where $\rho=O(\log n)$), we obtain an 
$O(\log n)$-approximation in the \allsets setting.
Combining Theorems~\ref{gxyalg} (a), \ref{drscpolythm}, and~\ref{fracsaathm}, and a local 
$\rho$-approximation algorithm, we obtain an $O(\log^2 n)$ in the \kbounded setting. 

\begin{proofof}{Theorem~\ref{fracsaathm}} 
The proof follows by suitably discretizing $\Pc$ and applying Theorem~\ref{mainsaathm} to
the discretized version of $\Pc$. 
By Lemma~\ref{apxsub}, for every distribution $p$, we have that the Lipschitz constant of
$\hpcentral[p]{x}$ is at most $K':=\|c\|+K$, and $\ln K'=\poly(\inpsize)$. 
Recall that by \ref{p3}, $\Pc$ is contained in the ball $B(\bo,R)$, and contains a ball of
radius $V\leq 1$ such that $\ln\bigl(\frac{R}{V}\bigr)=\poly(\inpsize)$.  
We discretize $\Pc$ as in~\cite{SwamyS12}. 
Let $\Delta=\frac{\ve \cdot \lb \cdot V}{8K'R\sqrt{m}}$, and consider the grid 
$\G=\{x\in\Pc: x_i = n_i\Delta,\ \ n_i \in \mathbb{Z}_+ \text{ for all }i = 1,\dots,m \}$.%
\footnote{Note that $V$ needs to be a part of the specification of the grid size;
  otherwise, a ``flat'' $\Pc$ could evade the grid across arbitrarily large distances.}
As shown in~\cite{SwamyS12}, we have: 
(i) $|\G| \le \bigl(\frac{2R}{\Delta}\bigr)^m$, and so 
$\log |\G| = \poly\bigl(\inpsize,\log(\frac{1}{\ve\cdot\lb})\bigr)$; and
(ii) for any $x\in\Pc$, letting $\phi(x)$ denote the point in $\G$ closest to $x$ in
Euclidean distance, we have 
$\bigl\|x-\phi(x)\bigr\|\leq\frac{\ve\lb}{K'}$, and
hence, $\bigl|\hpcentral[p]{x}-\hpcentral[p]{\phi(x)}\bigr|\leq\ve\lb$.

Let $N$, the number of samples used to construct each empirical estimate $\htp^i$, be
as given by Theorem~\ref{mainsaathm}, when we apply it taking $X$ to be the grid
$\G$---i.e., we are considering the \dr 2-stage problem $\min_{x\in\G}\hpcentral{x}$---%
and $\kp=\ve\lb$. Note that properties~\ref{p1}--\ref{p6} hold for this \dr problem (since
by assumption they hold for the \dr problem $\min_{x\in X}\hpcentral{x}$). 

To apply Theorem~\ref{mainsaathm} with $X=\G$, we also need to supply the points $\hx^i$
and the estimates $f^i$ as required by the theorem statement. We set $\hx^i=\phi(\bx^i)$,
and $f^i=\max\{\barf^i,\lb\}$ for all $i=1,\ldots,k$. 
We show that these satisfy properties (S1) and (S2) in the statement of
Theorem~\ref{mainsaathm}, with $\beta=\bbeta(1+\ve)$. 
To see this, consider any $i=1,\ldots,k$.
We have 
$$
\hpcentral[\htpcent^i]{\hx^i} \le \hpcentral[\htpcent^i]{\bx^i} + \ve\lb \le
\bbeta\cdot\barf^i+\ve\lb\leq(\bbeta+\ve)f^i\leq\beta f^i
$$
and since $\lb\leq\min_{x\in\Pc}\hpcentral{x}$, we have
$f^i \le \rho \min_{x\in \Pc}\hpcentral[\htpcent^i]{x} \le
\rho \min_{x\in \G}\hpcentral[\htpcent^i]{x}$. 
Moreover, the index $j$, which is a
minimizer of the $\{\barf_i\}$ estimates, is also a minimizer for the new estimates
$\{f^i\}$.   
So applying Theorem~\ref{mainsaathm}, 
we obtain that with probability at
least $1 - 3\delta$, 
\[
\hpcentral{\hx^j}\leq 4\beta\rho\bigl(1+O(\ve)\bigr)\min_{x\in \G}\hpcentral{x}+2\beta\rho\kp.
\]
Note that 
$\min_{x\in \G}\hpcentral{x} \le\min_{x\in\Pc}\hpcentral{\phi(x)} \le\min_{x\in\Pc}\hpcentral{x} + \ve\lb$.
Therefore, we have
\begin{equation*}
\begin{split}
\hpcentral{\bx} & \leq\hpcentral{\hx^j}+\ve\lb
\leq 4\beta\rho\bigl(1+O(\ve)\bigr)\min_{x\in\Pc}\hpcentral{x}+4\beta\rho\bigl(1+O(\ve)\bigr)\ve\lb
+2\beta\rho\kp+\ve\lb \\
& \leq 4\bbeta\rho\bigl(1+O(\ve)\bigr)\min_{x\in\Pc}\hpcentral{x}. \qedhere
\end{split}
\end{equation*}
\end{proofof}

\subsubsection*{Proof of Theorem~\ref{drscpolythm}}
Let $\bigl(U,\Sc,\{c_S,c^\two_S\}_{S\in\Sc}\bigr)$ be the \dr set cover instance being
solved. For any point $\bx \in \Pc$, let
$\Sbx := \{e \in U: \sum_{S \in \Sc:e \in S} \bx_S \ge 1 / 2\}$ be the set of elements
covered to an extent of at least $1/2$ by the first-stage sets. 

The improvement comes from a better way of generating a cut passing through the center
$\bx$ of the current ellipsoid, when $\bx\in\Pc$.
Instead of rounding $\bx$ 
to $\tx\in X$ using a local $\rho$-approximation algorithm and using approximate solutions to
$g(\tx,y,A)$ to generate a suitable cut at $\bx$ in step \ref{polyalgcut} of Algorithm
\polyalg, we do the following. 
Since elements in $\Sbx$ are mostly covered by $\bx$, and the remaining elements are
mostly uncovered, intuitively only these remaining elements should matter.
Indeed, we argue that approximate solutions to 
$\max_{A'\in\A}\bigl(g(\bo,A'\sm\Sbx)-y\cdot\scm(A,A')\bigr)$ 
can be used to obtain a suitable cut at $\bx$. Note that this problem can be cast as
$g(\bo,y,A)$ for a modified instance where we add $\Sbx$ to our set-system, with costs
$c_{\Sbx}=c^\two_{\Sbx}=0$.  
Thus, we avoid the $\rho$-factor loss that was incurred earlier due to the local
approximation. 

Consider the following LP.
\begin{alignat*}{2}
\max & \quad & \sum_{(A,A')\in\Asup\times\A}\gm_{A,A'}&g(0,A' \setminus \Sbx) \tag{W${}_{\bx}$} \label{Wbx} \\
\text{s.t.} && \sum_{A'} \gm_{A,A'} & \leq \ppoly_A \qquad \forall A\in\Asup \\
&& \sum_{A,A'}\scm(A,A')\gm_{A,A'} & \leq r \\[-2ex]
&& \gm & \geq 0. 
\end{alignat*}
We prove analogues of Lemmas~\ref{gxytoapxsub} and~\ref{validcut} showing that one can
compute an approximate solution to \eqref{Wbx} using an approximation algorithm for
$g(\bo,y,A)$ (Lemma~\ref{newgxytoapxsub} (i)), which allows us to both approximate
$\hppoly{\tx}$ for a related point $\tx$ (Lemma~\ref{newgxytoapxsub} (ii)), and obtain a
suitable cut passing through $\bx$ (Lemma~\ref{newvalidcut}). 

\begin{lemma} \label{gxytoapxsub_setcover} \label{newgxytoapxsub}
Let $\bx\in\Pc$ and $\tx := (\min\{2\bx_S, 1\})_{S \in \Sc}$. Suppose we have a
$(\beta_1,\beta_2)$-approximation algorithm for $g(\bo,y,A)$ for all 
$(y,A)\in\Rplus\times\A$.  Then, (i) we can compute a
$\beta_1\beta_2$-approximate solution $\gamma$ to \eqref{Wbx};
(ii) hence, letting $\tf = 2 c^\T \bx+ \sum_{(A,A')\in\Asup\times\A}\gm_{A,A'}g(\bo,A'\setminus \Sbx)$, we have  
$\hppoly{\tx}\leq \beta_1\beta_2 \cdot \tf$.
\end{lemma}

\begin{proof}
Consider the instance of \dr set cover obtained from the original instance
$\bigl(U,\Sc,\{c_S,c^\two_S\}_{S\in\Sc}\bigr)$ by adding the set $\Sbx$ to $\Sc$, with
costs $c_{\Sbx} = c^\two_{\Sbx} = 0$. Let $\{g^\new(x, A)\}_{x \in \Pc, A \in \A}$ denote the
second-stage costs for this new instance of \dr set cover. Note that, for every scenario
$A \in \A$, we have $g^\new(\bo, A) = g(\bo, A\setminus \Sbx)$. Therefore, if we were to write
the LP ($\text{T}_{\ppoly, 0}$) for this modified instance of \dr set cover (i.e.,
($\text{T}_{\ppoly, 0}$) with $g$ substituted by $g^\new$), we would obtain 
\eqref{Wbx}. This means that we can obtain a $\beta_1\beta_2$-approximate solution $\gm$ to
\eqref{Wbx} by applying Lemma~\ref{lem:beta-dual} (i) to the modified instance (using the
$(\beta_1, \beta_2)$-approximation algorithm for $g(\bo,y,A)$ 
given to us, also applied to the modified instance). This proves (i). 

To prove (ii), let $\gm^*$ be an optimal solution of (T${}_{\ppoly,\tx}$). We obtain
\begin{alignat*}{2}
\hppoly{\tx} & = c^\T \tx+\sum_{(A,A')\in\Asup\times\A}\gm^*_{A,A'}g(\tx,A') 
\le 2 c^\T\bx+\sum_{(A,A')\in\Asup\times\A}\gm^*_{A,A'}g(0,A' \setminus \Sbx) \\ 
& \le 2 c^\T \bx+ \beta_1\beta_2 \sum_{(A,A')\in\Asup\times\A}\gm_{A,A'}g(0,A' \setminus
\Sbx) \le \beta_1\beta_2 \cdot \tf \ . 
\end{alignat*}
The first inequality follows because $\tx \le 2\bx$ and, for every scenario $A' \in\A$, we
have $g(\tx,A') \le g(\bo,A'\setminus\Sbx)$. The latter inequality holds because every
feasible fractional second-stage solution for scenario $A'\sm\Sbx$ with $x=\bo$ as the
first-stage solution,
covers all elements of $A'\setminus\Sbx$ fully, and hence, combined with $\tx$, fully
covers all elements of $A'$; therefore, it yields feasible fractional second-stage actions
for scenario $A'$ given the first-stage actions $\tx$.
The second inequality above follows because $\gm$ is a 
$\beta_1\beta_2$-approximate solution for \eqref{Wbx}. The final inequality uses the fact
that $\beta_1,\beta_2 \ge 1$.
\end{proof}

\begin{lemma} \label{validcut_setcover} \label{newvalidcut}
Let $\bx\in\Pc$ and $\tx := (\min\{2\bx_S, 1\})_{S \in \Sc}$. Let $\gm$ be a $\beta$-approximate 
solution to the LP
\eqref{Wbx}, and let $\tsgr=c+\sum_{(A,A')\in\Asup\times\A}\gm_{A,A'}\sgr^{\bx,A'\sm\Sbx}$.
If $x'\in\Pc$ is such that $\tsgr^\T(x'-\bx)\geq 0$, then 
$\hppoly{x'}\geq \frac{1}{2}\bigl(2 c^\T \bx+ 
\sum_{(A,A')\in\Asup\times\A}\gm_{A,A'}g(\bo,A'\setminus \Sbx)\bigr) \ge 
\frac{1}{2\beta} \cdot \hppoly{\tx}$. 
\end{lemma}

\begin{proof}
Consider the function $f(x) = c^\T x + \sum_{(A,A')\in\Asup\times\A}\gm_{A,A'}g(x,A'
\setminus \Sbx)$ defined over $\Pc$. Note that $\gm$ is feasible for the LP
(T${}_{\ppoly,x'}$) and $g(x',A' \setminus \Sbx) \le g(x', A')$ for every scenario 
$A' \in\A$, which implies $\hppoly{x'} \ge f(x')$. 
By mimicking the proof of Lemma~\ref{apxsub}, we have that
$\tsgr$ is a subgradient of $f$ at $\bx$. So 
\[
\hppoly{x'} \ge f(x') \ge f(\bx) + \tsgr^\T(x' - \bx) \ge f(\bx) \ .
\]
Now, note that for every scenario $A' \in \A$, we have 
$g(\bx,A' \setminus \Sbx) \ge\frac{1}{2}g(\bo,A' \setminus \Sbx)$. 
This is because if $z$ is a feasible second-stage solution to scenario $A'\sm\Sbx$ given
$\bx$ as the first-stage actions,
then it covers elements of $A'\sm\Sbx$ to an extent of at least $\frac{1}{2}$, and so
$(\min\{2z_S, 1\})_{S\in\Sc}$ is a feasible second-stage solution for $A' \setminus\Sbx$
given $\bo$ as the first-stage actions. So we obtain  
\begin{alignat*}{2}
\hppoly{x'} & \ge f(\bx) \ge 
c^\T \bx + \frac{1}{2} \sum_{(A,A')\in\Asup\times\A}\gm_{A,A'}g(\bo,A' \setminus \Sbx) \\
& = \frac{1}{2}\Bigl(2 c^\T \bx+ \sum_{(A,A')\in\Asup\times\A}\gm_{A,A'}g(0,A' \setminus\Sbx)\Bigr) 
\ge \frac{1}{2\beta}\cdot\hppoly{\tx},
\end{alignat*}
where the last inequality follows from Lemma~\ref{gxytoapxsub_setcover} (ii).
\end{proof}

We now exploit Lemmas~\ref{gxytoapxsub_setcover} and \ref{validcut_setcover} to obtain
Theorem~\ref{fracsaathm}. We do so by mimicking the proof of Theorem~\ref{polythm}, and
pointing out the changes to Algorithm \polyalg and its analysis.
Let $\alg$ be a $(\beta_1,\beta_2)$-approximation algorithm for $g(\bo,y,A)$ for all
$(y,A)\in\Rplus\times\A$. 
As before, we start by using Lemma~\ref{lbnd}, either certifying
that $x =\bo$ is an optimal solution to (\Qppoly) (in which case we return $x = \bo$, and an
estimate of $\hppoly{0}$ computed via Lemma~\ref{lem:beta-dual}), or that 
$\min_{x\in \Pc}\hppoly{x} \ge \lb$, where $\lb = \frac{r}{\beta_1 \scm_{\max}}$. Suppose we are in
the latter case. 
We run Algorithm \polyalg with parameter $\kp=\ve\lb$, but 
modify step \ref{polyalgcut} 
as follows.

\begin{itemize}
\item If $\bx_i\in\Pc_k$, let $\tx_k := (\min\{2\bx_{i,S}, 1\})_{S \in \Sc}$.
Use Lemma~\ref{gxytoapxsub_setcover} and $\alg$ to obtain a $\beta_1\beta_2$-approximate
solution $\gm$ to (W${}_{\bx_i}$) (which has polynomial-size support). Define 
$\tsgr_k:=c+\sum_{(A,A')\in\Asup\times\A}\gm_{A,A'}\sgr^{\bx_i,A' \setminus S_{\bx_i}}$, and
$\tf_k:=2 c^\T \bx_i+ \sum_{(A,A')\in\Asup\times\A}\gm_{A,A'}g(\bo,A' \setminus S_{\bx_i})$.  
If $\tsgr_k=\bo$, then return $\tx_k$ and $\tf_k$.
Otherwise, let $H$ denote the halfspace $\{x\in\R^m: \tsgr_k^\T (x-\bx_i)\leq 0\}$. 
Set $\Pc_{k+1}\assign\Pc_k\cap H$, and $k\assign k+1$.
\end{itemize}

By Lemma~\ref{newgxytoapxsub} (ii), we immediately obtain
that $\hppoly{\tx_l}\leq\beta_1\beta_2\cdot\tf_l$ for all $l=1,\ldots,k$.
Let $\sx \in \Pc$ be an optimal solution to $\min_{x\in \Pc}\hppoly{x}$. 
We show that there exists an index $l$ such that
$\tf_l\leq 2(1+\ve)\cdot\hppoly{\sx}$. We have two cases to consider. 

\begin{itemize}
\item \underline{Case 1}: we have $\tsgr_l\cdot(\sx-\bx_l)\geq 0$ for some $l$ (this
  includes the case where $\tsgr_l = 0$). Then Lemma~\ref{validcut_setcover} shows that 
$\tf_l\leq2\cdot\hppoly{\sx}$.

\item \underline{Case 2}: we have $\tsgr_l\cdot(\sx-\bx_l)< 0$ for all $l$. In this case,
as argued in the proof Theorem~\ref{polythm}, 
  we can show that there must be a point $x' \in \Pc$ such that
  $\hppoly{x'}\leq\hppoly{\sx}+\kp$ and $\hsgr_l\cdot(x'-\bx_l)=0$ for some $l$. Using
  Lemma~\ref{validcut_setcover} again, we obtain
  $\tf_l\leq2\cdot\hppoly{x'}\leq2\bigl(\hppoly{\sx}+\kp\bigr) 
=2\cdot\hppoly{\sx}+2\ve\cdot\lb\leq 2(1+\ve)\hppoly{\sx}$. \hfill \qed
\end{itemize}

\subsubsection{Vertex cover} 
This is the special case of set cover where we want to cover
edges of a graph by vertices, and we again consider the $\frac{1}{2}L_1$-metric. 
We have $\al=2$, $\rho=2\al$, so we obtain approximation factors of
$\bigl(4\rho+O(\ve)\bigr)=\bigl(16+O(\ve)\bigr)$ in the \allsets setting (using
Theorems~\ref{polythm_collapsible} (a) and~\ref{fracsaathm}), and \linebreak
$\bigl(4\rho\al\cdot\frac{2e}{e-1}+O(\ve)\bigr)=\bigl(101.25+O(\ve)\bigr)$ 
in the \kbounded setting (via Theorems~\ref{gxyalg} (b), \ref{polythm}, 
and~\ref{mainsaathm}).

\subsubsection{Edge cover} 
This is the special case of set cover where we want to cover
vertices of a graph by edges, and we again consider the $\frac{1}{2}L_1$-metric. 
We have $\al=\frac{3}{2}$,
$\rho=2\al$, so we obtain approximation factors of
$\bigl(12+O(\ve)\bigr)$ in the \allsets setting (via Theorems~\ref{polythm_collapsible}
(a) and~\ref{fracsaathm}), and $\bigl(36+O(\ve)\bigr)$
in the \kbounded setting (via Theorems~\ref{gxyalg} (c), \ref{polythm}, 
and~\ref{mainsaathm}).

\subsubsection{Facility location} 
The \dr version (\drsfl) was defined in Section~\ref{prelim}.
Recall that an instance is given by the tuple \linebreak
$\bigl(\F,\C,\{\dist_{ij}\}_{i,j\in\F\cup\C},\{f_i,f^\two_i\}_{i\in\F}\bigr)$,
where $\F$, $\C$ are the facility and client-sets respectively, $\dist$ is the underlying 
metric, and $f,f^\two$ are the first- and second-stage facility-opening costs.
We have $\al=1.488$~\cite{Li}. 
Shmoys and Swamy~\cite{ShmoysS06} showed that an LP-relative $\vro$-approximation for
deterministic FL having a certain ``demand-obliviousness'' property can be turned into a
$(\vro+\al)$-approximation algorithm for 2-stage FL. If the $\vro$-approximation algorithm
has the property that it returns a solution where every cost component of the rounded
solution---i.e., the facility cost, and {\em each} client's assignment cost---is at most
$\vro$ times the corresponding cost component of the fractional solution, then the
resulting algorithm is a local approximation algorithm. Using the deterministic
$4$-approximation algorithm of~\cite{ShmoysTA97} gives a local $\rho$-approximation with
$\rho=5.488$.  

As noted in Section~\ref{prelim}, besides the 
discrete scenario metric, we could define various other 
natural scenario metrics here in terms of the metric $\dist$ and
obtain a rich class of \dr models under the Wasserstein metric. 
We consider one such setting: the asymmetric metric given by
$\scm^\asym_\infty(A,A'):=\max_{j'\in A'}\dist(j',A)$. 

\begin{theorem} \label{flasym}
For \drsfl with $\scm$ being either the discrete metric $\scmdisc$ or the asymmetric metric
$\scm^\asym_\infty$, 
there is a $(6,1)$-approximation for computing $g(x,y,A)$ in the 
\kbounded setting, for any $(x,y,A)\in X\times\Rplus\times\A$, 
\end{theorem}

For 
the Wasserstein metric with respect to both the discrete metric and $\scm^\asym_\infty$,
we can take 
$\tau=\bigl(\sum_{i\in\F}f^\two_i+\sum_{i\in\F,j\in\C}\dist_{ij}\bigr)/(\min_{i,j:\dist_{ij}>0}\dist_{ij})$. 
We obtain the following approximation guarantees for \drsfl with the Wasserstein metric
corresponding to the above scenario metrics:
(i) $\bigl(4\rho+O(\ve)\bigr)=\bigl(21.96+O(\ve)\bigr)$ in the \allsets setting (using 
Theorems~\ref{polythm_collapsible} (a) and~\ref{fracsaathm}); and 
(ii) $\bigl(24\rho\al+O(\ve)\bigr)=\bigl(196+O(\ve)\bigr)$ in the \kbounded setting
(using Theorems~\ref{flasym}, \ref{polythm}, and~\ref{mainsaathm}).

\subsubsection*{Proof of Theorem~\ref{flasym}} \label{flasym-proof}
Fix $(x,y,A)\in X\times\Rplus\times\A$, where $\A=\A_{\leq k}:=\{A\sse\C:|A|\leq k\}$.
Fix $\scm$ to be either the discrete scenario metric $\scmdisc$ or the asymmetric metric $\scm^\asym_\infty$.
Since $\scm(A,A')$ takes polynomially-many values, by Lemma~\ref{gxyredn} (i),
it suffices to give a $6$-approximation for the constrained problem
\eqref{Was_oracle_simplified}: $\max_{A'\in\A:\scm(A,A')\leq\scmbound}g(x,A')$.

With both scenario metrics, this amounts to approximating the \kmaxmin fractional facility
location problem for an underlying facility-location instance 
$\bigl(\F,\C',\{\dist_{ij}\}_{i,j\in\F\cup\C'},\{\tf_i\}_{i\in\F}\bigr)$,
where $\tf_i=0$ if $x_i=1$, and is $f^\two_i$ otherwise.
If $\scm = \scmdisc$ and $\scmbound>0$, then $\C'=\C$ (if $\scmbound=0$, the
optimum of the constrained problem is $g(x,A)$);
if $\scm=\scm^\asym_\infty$, then $\C':=\{j\in\C: \dist(j,A)\leq\scmbound\}$.

\vspace*{-1ex}
\paragraph{A \boldmath $6$-approximation algorithm for \kmaxmin facility location.}
We now devise an algorithm for the \mbox{\kmaxmin fractional facility-location} problem
corresponding to a facility-location instance  
(such as the one obtained above) 
\mbox{$\bigl(\F,\C',\{\dist_{ij}\}_{i,j\in\F\cup\C'},\{\tf_i\}_{i\in\F}\bigr)$.}

Khandekar et al.~\cite{Khandekar} give a $10$-approximation for the version of \kmaxmin
integral FL, where a scenario may place an {\em arbitrary} number of co-located
clients at a location in $\C'$ (and the {\em total} number of clients must be at most $k$).%
\footnote{Since the gap between the integral and fractional optimal
values for FL is at most $\al=1.488$~\cite{Li}, a $\beta$-approximation for the integral
(resp. fractional) version implies an $\al\beta$-approximation for \kmaxmin fractional
(resp. integral) facility location.}  
However, in our setting, we may place at most one client
at any location in $\C'$, so the algorithm in~\cite{Khandekar} does not work for our
purposes. (Clearly, our setting is more general, since we can encode the scenario-setting
of~\cite{Khandekar} by creating $k$ co-located copies at every $j\in\C'$.)
As noted earlier, we can model more-general settings, where clients have (integer)
demands, by creating a fixed number of co-located clients at locations in $\C'$; but,
here again, we have a constraint that {\em limits} the number of co-located clients at any 
$j\in\C'$. 

We therefore need to develop new techniques to devise an approximation algorithm for
\kmaxmin fractional FL. The key tool that we exploit here is that of 
{\em cost-sharing schemes}.
We uncover a novel connection between cost-sharing schemes and \kmaxmin problems 
by demonstrating that one can exploit a cost-sharing scheme for FL having certain
properties to obtain an approximation algorithm for \kmaxmin \{integral, fractional\} FL. 
Our result also improves the approximation factor for \kmaxmin integral FL
from $10$ to $6$. 

A cost-sharing method is a function $\xi:2^{\C'}\times\C'\rightarrow\R_+$, where $\xi(\J,j)$
for $j\in\J$, intuitively gives the contribution of $j$ towards the cost incurred in
satisfying the client-set $\J$ (i.e., the cost of opening facilities and assigning clients
in $\J$ to these open facilities). P\'al and Tardos~\cite{PalT} devised a cost-sharing method 
$\xi$ satisfying the following properties. For sets $\J,T\sse\C'$, define
$\xi(\J,T):=\sum_{j\in T}\xi(\J,j)$.
\begin{enumerate}[label=$\bullet$, topsep=1ex, itemsep=0ex]
\item $\xi(S,j)=0$ if $j\notin S$.
\item (Competitiveness) For every $\J\subseteq\C'$, we have $\xi(\J, \J) \le g(x,\J)$.
\item (Cost-recovery) For every $\J\subseteq \C'$, we have $\xi(\J, \J) \ge g(x,\J)/3$.
\item (Cross-monotonicity) For all $\J_1\sse \J_2 \subseteq\C'$ and every client $j \in\C'$,
  we have $\xi(\J_2, j) \le \xi(\J_1, j)$. 
\end{enumerate}

We will prove an additional useful property about $\xi$, for which we very briefly
describe how $\xi$ is computed. 
For every $\J\sse\C'$ and $i\in\F$, we compute a certain {\em time} $\tm(\J, i) \ge 0$. The
cost-share of a client $j\in\J$ is then defined as 
$\xi(\J, j) := \min_{i \in \F} \max\{t(\J, i), \dist_{ij}\}$. 
The function $\tm(\cdot,\cdot)$ satisfies the following property: for every set 
$\J \subseteq \C'$, every client $j \not \in \J$, and every facility $i \in \F$, we have
$\tm(\J + j, i) \le \tm(\J, i)$. Further, if this inequality is strict, then 
$\tm(\J + j, i) \ge \dist_{ij}$. 

\begin{lemma}
\label{lem:insertion}
Consider $\J \subseteq \C'$ and two clients $j_1\in \J$ and $j_2 \not \in J$.
Then $\xi(\J+j_2,j_1)\geq\min\bigl\{\xi(\J,j_1),\xi(\J+j_2,j_2)\bigr\}$.
\end{lemma}

\begin{proof}
By cross-monotonicity, we have $\xi(\J + j_2, j_1) \le \xi(\J, j_1)$. If this holds at
equality, then the result follows immediately. So assume otherwise.
By the way in which the cost-shares are defined, 
$\xi(\J+j_2,j_1)<\xi(\J,j_1)$ implies that 
$\xi(\J + j_2, j_1) = t(\J + j_2, i)$ for some facility $i$ and 
$\tm(\J + j_2, i) < \tm(\J, i)$. This implies that $\tm(\J+ j_2, i) \ge\dist_{ij_2}$, and
it follows that $\xi(\J + j_2, j_2) \le \max\{\tm(\J + j_2, i), \dist_{ij_2}\} = \tm(\J +
j_2, i) = \xi(\J + j_2, j_1)$.  
\end{proof}

We may assume that $k \le |\C'|$ (otherwise, we simply set $k = | \C'|$). 
Consider the following simple greedy algorithm. Initialize $t\gets 0$, $\J_0\gets\emptyset$. 
For $t=1,\ldots,k$, we find $\jbar \gets \argmax_{j\in\C'\setminus \J_{t-1}} \xi(\J_{t-1} + j, j)$, and
set $\J_t\gets\J_{t-1}\cup\{\jbar\}$.

Let $O^*\in\A$ be such that $g(x,O^*)=\max_{A\in\A}g(x,A)$. 
We claim that $\xi(\J_k,\J_k)\geq\xi(\J_k\cup O^*,\J_k\cup O^*)/2$. This will complete the
proof since this implies that 
$$
g(x,\J_k)\geq\xi(\J_k,\J_k)\geq\frac{\xi(\J_k\cup O^*,\J_k\cup O^*)}{2}
\geq\frac{g(x,\J_k\cup O^*)}{6}\geq\frac{g(x,O^*)}{6}.
$$
In fact~\cite{PalT} show a stronger form of cost-recovery, namely, that there is an
integer solution $z^\J$ feasible for scenario $\J$ given first-stage decisions $x$ such
that $\xi(\J,\J)\geq\bigl(\text{cost of }z^\J\bigr)/3$ for every $\J\sse\C'$, and using
this in the above chain of inequalities shows that $S_k$ yields a $6$-approximation also
for \kmaxmin{} {\em integral} facility location.

We now prove the above claim.
For any $t=1,\ldots,k$, we show that
$\xi(\J_t,j)\geq\lbfl_t$ for all $j\in\J_t$, where $\lbfl_t:=\max_{j'\in\C'\sm\J_{t-1}}\xi(\J_{t-1}+j',j')$.
We prove this by induction on $t$. Note that $\lbfl_t\geq\lbfl_{t+1}$ due to
cross-monotonicity, and since $\C'\sm\J_{t-1}\supseteq\C'\sm\J_t$.
The statement is clearly true for $t=1$. Suppose this is true for index $t$, and
consider index $t+1$. Consider any $j\in\J_{t+1}$. Let $\jbar$ be the element added to
$\J_t$ in iteration $t+1$. By definition, $\xi(\J_{t+1},\jbar)=\lbfl_{t+1}$.
If $j\in\J_t$, then 
$\xi(\J_{t+1},j)\geq\min\{\xi(\J_t,j),\xi(\J_{t+1},\jbar)\}\geq\min\{\lbfl_t,\lbfl_{t+1}\}=\lbfl_{t+1}$, where the second inequality follows from the induction hypothesis. 
Thus, for every $j\in\J_{t+1}$, we have $\xi(\J_{t+1},j)\geq\lbfl_{t+1}$.
This completes the induction step.

Therefore, by repeatedly using cross-monotonicity, we have
\begin{equation*}
\begin{split}
\xi(\J_k,\J_k) & \geq k\cdot\lbfl_k
\geq k\cdot\max_{j\in O^*\sm\J_k}\xi(\J_{k-1}+j,j)
\geq k\cdot\max_{j\in O^*\sm\J_k}\xi(\J_{k}\cup O^*,j) \\
& \geq k\cdot\frac{\xi(\J_{k}\cup O^*,O^*\sm\J_k)}{|O^*\sm\J_k|} 
\geq \xi(\J_{k}\cup O^*,O^*\sm\J_k)=\xi(\J_k\cup O^*,\J_k\cup O^*)-\xi(\J_k\cup O^*,\J_k) \\
& \geq \xi(\J_k\cup O^*,\J_k\cup O^*)-\xi(\J_k,\J_k).
\end{split}
\end{equation*}
The first inequality follows from the statement proved in the previous paragraph; the
second is simply because we restricted $\C'\sm\J_{k-1}$ to $O^*\sm\J_k$; the
third follows from cross-monotonicity; the fourth is because we replaced $\max$ by an
average and all cost shares are nonnegative; the fifth is because $|O^*|\leq k$; and the
last inequality is again due to cross-monotonicity.
\hfill \qed

\subsubsection{Steiner tree} 
The \dr version (\drst) was defined in Section~\ref{prelim}.
Recall that an instance is given by $\bigl(G=(V,E),c,s,\ld\bigr)$, where $(G,c)$ is a
metric, $s$ is the root, and $c_e, c^\two_e=\ld c_e$ are the costs
of buying edge $e$ in stages I and II respectively.

We do not have a local approximation algorithm for \drst, but there is a restricted local
$O(1)$-approximation algorithm for a {\em monotone} version of \drst,
wherein we require that in every scenario $A$, the path from each node $v\in A$ to
the root $s$ consists of a segment starting at $v$ comprising edges bought in scenario $A$,
followed by a segment ending at $s$ comprising first-stage edges. (Thus, in effect, the
first-stage edges $F$ should form a tree containing $s$.)
This monotonicity property was stipulated by~\cite{GuptaRS,DhamdhereGRS05} in the context
of 2-stage \{stochastic, robust\} Steiner tree respectively, where they show that imposing
this condition only incurs a factor-$2$ loss. 
We argue that the same holds in the \dr setting. Thus, by utilizing the restricted local
$10$-approximation algorithm devised by~\cite{SGupta} for this monotone 2-stage Steiner
tree problem in Theorem~\ref{polythm_collapsible}, and the well-known LP-relative
$2$-approximation for Steiner tree, 
we obtain the following results for the \allsets setting. 

\begin{theorem} \label{drstthm}
\drst admits a $(160 +O(\ve))$-approximation algorithm in the \allsets setting with
the scenario metrics $\scmdisc$ and $\scm^\asym_\infty$ (defined with respect to
the metric $c$ on $V$). 
\end{theorem}

\subsubsection*{Proof of Theorem~\ref{drstthm}}
For \drst, the discrete first-stage action set is $X=\{0,1\}^E$.
We first show that imposing the monotonicity condition incurs a factor-$2$ loss for the
\dr problem. Recall that the monotonicity condition states that in every scenario $A$, the
path from a node $v\in A$ to the root $s$ consist of a segment of second-stage edges 
starting at $v$ followed by a segment of first-stage edges ending at $r$; we call such a
path a monotone path.
For $x=\chi^F\in X$, we say that $x+\chi^{F^A}$ contains a $v$-$s$ path (respectively a 
monotone $v$-$s$) path, if $F\cup F^A$ contains a $v$-$s$ path (respectively a monotone
$v$-$s$ path). 
We want to compare the following two \dr 2-stage Steiner tree problems. 
\begin{alignat}{1}
\min_{x \in X} & \quad c^\T x + \max_{q:\pdist(\pcentral, q) \le r}
{\textstyle \Exp_{A \sim q}}\biggl[\underbrace{\min_{F^A\sse E}\ 
\Bigl\{c(F^A):\ x+\chi^{F^A}\text{ contains a $v$-$s$ path }
\forall v\in A\Bigr\}}_{\text{\small{$\gdisc(x,A)$}}}\biggr]
\tag{\drst} \label{DRST} \\[0.5ex]
\min_{x \in X} & \quad c^\T x + \max_{q:\pdist(\pcentral, q) \le r}
{\textstyle \Exp_{A \sim q}}\biggl[\underbrace{\min_{F^A\sse E}\ 
\Bigl\{c(F^A):\ x+\chi^{F^A}\text{ contains a monotone $v$-$s$ path }
\forall v\in A\Bigr\}}_{\text{\small{$\gdiscmon(x,A)$}}}\biggr] 
\tag{\drstmon} \label{DRSTmon}
\end{alignat}

\begin{lemma}[\cite{DhamdhereGRS05}] \label{stageItreelem}
For every first-stage decision $\bx \in X$, there exists $\tx \in X$ such that 
$c^\T \tx\le 2 c^\T \bx$ and $\gdiscmon(\tx, A) \le 2 \gdisc(\bx, A)$ for every set 
$A\sse V$. 
\end{lemma}

\begin{corollary} \label{reductreelem}
Consider the \dr problems \eqref{DRST} and \eqref{DRSTmon} for an arbitrary scenario
collection $\A$. If $\tx$ is an $\al$-approximate solution to \eqref{DRSTmon}, then it is 
a $(2\al)$-approximate solution to \eqref{DRST}.
\end{corollary}

\begin{proof}
By applying Lemma~\ref{stageItreelem} to an optimal solution to \eqref{DRST}, 
we infer that $\OPT_{\drstmon}\leq 2\OPT_{\drst}$.
Note that for every scenario $A \in \A$, we have $\gdisc(\tx, A) \le \gdiscmon(\tx, A)$ by
definition. 
It follows that the objective value of $\tx$ in \eqref{DRST} is no
larger than its objective value in \eqref{DRSTmon}, which by assumption is at most
$\alpha\cdot\OPT_{\drstmon}\leq 2\al\cdot\OPT_{\drst}$. 
\end{proof}

Gupta et al.~\cite{GuptaRS} consider the following integer program (IP) for $\gdiscmon(x,A)$.
For notational simplicity, we assume that $s\notin A$; clearly, this can always be ensured
without changing the problem.
We have variables $\{z^A_e\}_{e \in E}$ to indicate the edges bought in stage
II. To encode the requirement that there is a monotone $v$-$s$ path for every $v\in A$, 
we bidirect the edges to obtain the set of arcs $\Ebi$,
and use flow variables $\{\fIAv_e\}_{e \in \Ebi}$ and $\{\fIIAv_e\}_{e\in\Ebi}$ to
specify the segments of $v$'s path comprising first-stage and second-stage edges. For a vertex $v \in V$, let $\deltain(v)$ (respectively $\deltaout(v)$) denote the arcs of $\Ebi$ entering (respectively leaving) $v$. For an arc $e \in \Ebi$, we abuse notation and use $x_e$ to denote the component of $x$ corresponding to the undirected version of $e$.
\begin{alignat}{3}
\min & \quad & \sum_{e \in E}c^\two_e z^A_e & \notag \\
\text{s.t.} && \sum_{e \in \deltaout(v)} (\fIAv_e + \fIIAv_e) & - \sum_{e \in \deltain (v)} (\fIAv_e + \fIIAv_e) \ge 1 \qquad  && \forall v \in A \label{ccut} \\
&& \sum_{e \in \deltaout(u)}\negthickspace\negthickspace (\fIAv_e + \fIIAv_e) & = 
\sum_{e \in \deltain(u)}\negthickspace\negthickspace (\fIAv_e+\fIIAv_e)  \qquad &&
\forall v \in A, u \in V \setminus \{s,v\} \label{cflowcons} \\
&& \fIAv_e \le x_e, \quad 
\fIIAv_e & \le z^A_e \qquad && \forall v \in A, \forall e \in \Ebi \label{cdom} \\
&& \sum_{e \in \deltain(u)} \fIAv_e & \le \sum_{e \in \deltaout(u)} \fIAv_e \qquad &&
\forall v \in A, u \in V \setminus \{s,v\} \label{cmono} \\
&& z^A_e & \in \{0, 1\} && \forall e \in E \label{cbinz} \\
&& \fIAv_e,\fIIAv_e & \in \{0, 1\} \qquad && \forall v \in A, e \in \Ebi \label{cbinI}
\end{alignat}
Constraints~\eqref{ccut} and~\eqref{cflowcons} enforce that 
$f^{\one, A, v}+f^{\two, A, v}$ sends one unit of flow from $v$ to $s$ for every terminal
$v\in A$ (so it dominates a directed $v\leadsto s$ path), and 
\eqref{cdom} enforces that this flow is supported on edges bought in stages I and II.
Constraints~\eqref{cmono} encode the monotonicity requirement on the $v$-$s$ path.

Letting $g(x,A)$ denote the optimal value of the LP-relaxation obtained by relaxing the
integrality constraints \eqref{cbinz}, \eqref{cbinI} to nonnegativity constraints, the \dr
2-stage Steiner problem (with fractional second-stage decisions) we consider is:
$\min\ \bigl(\hpcentral{x}:=c^\T x+\max_{q:\LW(\pcentral,q)\leq r}\E[A\sim q]{g(x,A)}\bigr)$;
we call this {\em monotone \drst}. By the discussion in the beginning of Section~\ref{apps}, properties~\ref{p1}--\ref{p6} hold for monotone \drst, setting $\ld = \max_{e \in E} c^\two_e / c_e$ and $\tau = \sum_{e \in E} c^\two_e/\min_{e \in E : c_e > 0} c_e$.

Recall that we are in the \allsets setting (so $\A=2^V$), and $\LW$ is the Wasserstein 
metric with respect to the discrete scenario metric $\scmdisc$ or the asymmetric metric
$\scm^\asym_\infty$. 
The set of scenarios is collapsible under both these scenario metrics by Lemma~\ref{collapsiblelem}. 
Gupta et al.~\cite{GuptaRS} presented a restricted local $20$-approximation algorithm for
monotone \drst, and the approximation factor was improved to $10$ by~\cite{SGupta}.
Therefore, 
utilizing Theorems~\ref{mainsaathm} and~\ref{polythm_collapsible}, taking $\rho = 10$ 
and $\alpha=2$ (and $\beta=1$ in Theorem~\ref{mainsaathm}),  
we obtain an $\bigl(80+O(\ve)\bigr)$-approximation for \eqref{DRSTmon}. This yields 
a $\bigl(160+O(\ve)\bigr)$-approximation for \drst (using
Lemma~\ref{reductreelem}). 
\hfill \qed

\subsubsection{Proof of Theorem~\ref{gxyalg}} \label{gxyalgproof}
We first give a reduction, showing that one can approximate $g(x,y,A)$ under very general
settings 
provided that we have a (standard) approximation algorithm for a certain constrained
problem. 

\begin{lemma} \label{gxyredn}
Let $\A$ be any scenario set, and $\scm:\A\times\A\rightarrow\Rplus$ be any function
satisfying $\scm(A,A)=0$ for all $A\in\A$. Fix $x\in X$, and scenario $A\in\A$. 
Consider the constrained problem:
\[
\tag{\gxcp}
\label{Was_oracle_simplified}
\max_{A' \in \A :\scm(A, A') \le \scmbound} g(x,A').
\]
Suppose that we have a $\beta$-approximation algorithm $\alg$ for \eqref{Was_oracle_simplified}.
Let $\SCM:=\{\scm(A,A'): A,A'\in\A\}$.

\noindent
(i) We can compute a $(\beta,1)$-approximation to $g(x,y,A)$ using $|\SCM|$ calls to
$\alg$. 

\noindent
(ii) For any $\ve>0$, we can compute a $(\beta,1+\ve)$-approximation to $g(x,y,A)$ using 
$O\bigl(\log_{1+\ve}(\frac{\scmmax}{\scmmin})\bigr)$ calls to $\alg$, where
$\scmmax:=\max_{A,A'}\scm(A,A')$ and $\scmmin:=\min_{A,A':\scm(A,A')>0}\scm(A,A')$.
\end{lemma}

\begin{proof}
The proof is based on a standard idea of enumerating over all $\scm(A,A')$ values.
For $\mu\in \SCM$, let $A_\mu\in\A$ denote the scenario output by $\alg$
for \eqref{Was_oracle_simplified}.

For part (i), we do the following. We compute $A_\mu$ for all $\mu\in\SCM$.
Let $\mustar := \argmax_{\mu \in \SCM} \bigl(g(x, A_\mu) - y\cdot\scm(A, A_\mu)\bigr)$. We
return $A_{\mustar}$. 
To show that this yields a $(\beta,1)$-approximation for computing $g(x,y,A)$,
consider any $A'\in \A$, and let $\mu' = \scm(A, A')$. We have 
\[
g(x,A_{\mustar}) - y\cdot\scm(A, A_{\mustar}) \ge g(x,A_{\mu'}) - y\cdot\scm(A,A_{\mu'}) 
\ge \frac{g(x,A')}{\beta} - y\cdot\scm(A,A').
\]
The first inequality follows from the definition of
$\mustar$, and the second follows since $A_{\mu'}$ is a $\beta$-approximate solution for
($\gxcp[\mu']$).  

For part (ii), we enumerate values in $[\scmmin,\scmmax]$ in powers of $(1+\ve)$. More
precisely, define
$\SCMbar := \{0\} \cup \bigl\{(1+\ve)^i \scmmin: i =0,\dots,\ceil{\log_{1+\ve}{\frac{\scmmax}{\scmmin}}}\bigr\}$. 
Note that $|\SCMbar| = O\bigl(\log_{1+\ve}({\frac{\scmmax}{\scmmin}})\bigr)$.
We now compute $A_\mu$ for all $\mu\in\SCMbar$.
Let $\mustar := \argmax_{\mu \in \SCMbar} \bigl(g(x, A_\mu) - y\cdot\scm(A, A_\mu)\bigr)$. We
return $A_{\mustar}$. 
Consider any $A'\in \A$. By construction of $\SCMbar$, there
is some $\mu'\in\SCMbar$ such that $\scm(A,A')\leq\mu'\leq(1+\ve)\scm(A,A')$.
Again, by the definition of $\mustar$, and since $A_{\mu'}$ is a $\beta$-approximate
solution for ($\gxcp[\mu']$), we have 
\[
g(x,A_{\mustar}) - y\cdot\scm(A, A_{\mustar}) \ge g(x,A_{\mu'}) - y\cdot\scm(A,A_{\mu'}) 
\geq \frac{g(x,A')}{\beta}-y\cdot\mu'
\ge \frac{g(x,A')}{\beta} - (1+\ve)y\cdot\scm(A,A'). \qedhere
\]
\end{proof}

We now consider the setting in Theorem~\ref{gxyalg}, namely, the \kbounded setting
with $\scm$ being the discrete metric, i.e., $\A=\{A\sse U:|A|\leq k\}$ for some ground
set $U$, and $\scm(A,A')=1$ if $A\neq A'$, and $0$ otherwise. 

Fix $x\in X$ and a scenario $A\in\A$.
By Lemma~\ref{gxyredn}, it suffices to give an approximation algorithm for the
constrained problem \eqref{Was_oracle_simplified}. When $\scmbound=0$, the optimum of the
constrained problem is simply $g(x,A)$ (which is easy to compute), and otherwise, the
constrained problem simplifies to $\max_{A'\in\A}g(x,A')$. So it suffices to obtain a
$\beta$-approximation to this latter problem, which is what we focus on in the sequel.

\vspace*{-1ex}
\paragraph{Part (a) of Theorem~\ref{gxyalg}.}
Gupta et al.~\cite{GuptaNR} give an $O(\log n)$-approximation algorithm for
\kmaxmin set cover, wherein the goal is to choose a set $A\in\A$ so as to maximize the
cost of an optimal integral set-cover for $A$. It is implicit in their analysis%
\footnote{See Theorem 4.2 and Claim 4.3 in~\cite{GuptaNR}; Theorem 4.2 proves that the
optimal fractional cost of the set-cover instance $(S,\F)$ is at most $c(\Phi^*)+12T^*$.}
that this also yields an $O(\log n)$-approximation for \kmaxmin fractional set cover,
where we seek to maximize the cost of an optimal {\em fractional} set cover. 

This immediately implies an $O(\log n)$-approximation for $\max_{A'\in\A}g(x,A')$ as
follows. 
Consider the set cover instance with ground set $U$, and set-costs given by
$w_S=0$ if $x_S=1$, and $w_S=c^\two_S$ otherwise. The \kmaxmin fractional set cover for
this instance is precisely the problem $\max_{A'\in\A}g(x,A')$.
So we obtain an $O(\log n)$-approximation to $\max_{A'\in\A}g(x,A')$. 

\vspace*{-1ex}
\paragraph{Part (b) of Theorem~\ref{gxyalg}.}
The problem $\max_{A'\in\A}g(x,A')$ can be viewed as \kmaxmin
fractional vertex cover, where the cost $w_v$ of a vertex $v$ is $0$ if $x_v=1$, and
$c^\two_v$ otherwise. Feige et al.~\cite{FeigeJMM05} give a $\frac{2e}{e-1}$-approximation
algorithm for \kmaxmin fractional vertex cover, so we obtain a
$\bigl(\frac{2e}{e-1},1\bigr)$-approximation for $\max_{A'\in\A}g(x,A')$. 

\vspace*{-1ex}
\paragraph{Part (c) of Theorem~\ref{gxyalg}.}
The problem $\max_{A'\in\A}g(x,A')$ can be viewed as \kmaxmin
fractional edge cover, where the cost $w_e$ of an edge $e$ is $0$ if $x_e=1$, and
$c^\two_e$ otherwise. Feige et al.~\cite{FeigeJMM05} give a $2$-approximation
algorithm for \kmaxmin fractional edge cover, so we obtain a
$(2,1)$-approximation for $\max_{A'\in\A}g(x,A')$. 
\hfill \qed

\section{Distributionally robust problems under the $L_\infty$-metric} \label{linfty}
We now focus on the \dr 2-stage problem \eqref{Qdisc}, and its fractional relaxation
\eqref{Qpcentral}, in the \allsets setting (so $\A=2^U$, for some $U$) when $\pdist$ is
the $L_\infty$-metric.  
Note that since the $L_\infty$-distance between two probability
distributions is at most $1$, we can assume without loss of generality that $r \le 1$. 
We devise an algorithm that, given any $\ve>0$, 
runs in time $\poly\bigl(\inpsize,\frac{\ld}{r\ve}\bigr)$,
and returns a
$\bigl(2+O(\ve)\bigr)$-approximate solution to the fractional relaxation
\eqref{Qpcentral}. 
Combining this with a local $\rho$-approximation algorithm, we obtain
a $\rho(2+O(\ve))$-approximation for the \dr discrete 2-stage problem (i.e., with discrete
first- and second- stage actions). This leads to the {\em first} guarantees for the \dr
versions of set cover, vertex cover, edge cover, and facility location under the $L_\infty$-metric
(Theorem~\ref{linftyapps}).

At a high level, our approach is as follows. We first show how to obtain a suitable convex
proxy function $\hproxy{x}$ 
that is pointwise close to the objective function $\hpcentral{x}$ so that one can cast the
problem of minimizing $\hproxy{x}$ as a standard 2-stage problem. 
Instead of utilizing the SAA approach to move to an SAA-version of $\hproxy{x}$
with a polynomial-size central distribution, show that a near-optimal solution to
the SAA problem translates to a near-optimal solution to the original problem, and finally
show how to approximately solve the SAA problem (which is again challenging 
since this does not reduce to a polynomial-size LP), it is simpler
to directly solve the proxy problem, $\min_{x\in\Pc}\hproxy{x}$, using the
approximate-subgradient based machinery in~\cite{ShmoysS06}. 
We show that, under the assumption that $g(x,A)\leq g(x,A')$ for all $x$,
$A\sse A'$, which holds for all our applications, one can compute 
an $\w$-subgradient of $\hproxy{x}$ efficiently in time
$\poly\bigl(\inpsize,\frac{\ld}{\w}\bigr)$, and hence can directly use the ellipsoid-based
approach in~\cite{ShmoysS06} to 
obtain a solution $\bx\in\Pc$ such that 
$\hproxy{\bx}\leq\bigl(1+O(\ve)\bigr)\min_{x\in\Pc}\hproxy{x}+\kp$. 
This in turn implies that 
$\hpcentral{\bx}\leq\bigl(2+O(\ve)\bigr)\min_{x\in\Pc}\hpcentral{x}+\kp$. 
We can fold the additive error into the multiplicative error by obtaining a lower
bound on the optimum.

\begin{theorem} \label{linftysolve}
Let $\ve\leq\frac{1}{3}$.
Suppose that for all $x\in \Pc$, and all $A\sse A'$, we have $g(x,A)\leq g(x,A')$. In the
\allsets setting ($\A = 2^U)$ under the $L_\infty$ metric, 
we can compute a solution $\bx\in\Pc$ satisfying 
$\hpcentral{x}\leq\bigl(2+O(\ve)\bigr)\min_{x\in\Pc}\hpcentral{x}$
with probability at least $1-\dt$, in time 
$\poly\bigl(\inpsize,\frac{\ld}{\ve r},\log(\frac{1}{\dt})\bigr)$.
\end{theorem}

\begin{theorem} \label{linftyapps}
We obtain the following approximation factors for the \dr discrete 2-stage problems in the
\allsets setting under the $L_\infty$ metric:  
(a) $O(\log n)$ for set cover; (b) $8+O(\ve)$ for vertex cover; (c) $6+O(\ve)$ for edge
cover; and (d) $10.98+O(\ve)$ for facility location.
\end{theorem}

\begin{proof} 
This follows by rounding the solution returned by Theorem~\ref{linftysolve}, because, as
noted in Section~\ref{apps}, we have local approximation algorithms with guarantees of   
(a) $O(\log n)$ for set cover (where $n=|U|$); (b) $4$ for vertex cover; (c) $3$ for edge
cover; and (d) $5.488$ for facility location. 
\end{proof}

In the sequel, we focus on proving Theorem~\ref{linftysolve}.
We first work our way towards defining the proxy function that we use. Note that for every distribution $q$ with $L_\infty(\pcentral, q) \le r$, we must have $q_A \ge \max\{\pcentral_A-r , 0\}$ for every scenario $A \in \A$. We refer to the right side of this inequality as the {\em blocked mass} in scenario $A$. The remainder of the probability mass $\pcentral_{A}$ (i.e., the difference $\pcentral_A$ and the blocked massed) may be moved to other scenarios, and hence we call it the {\em free mass} in scenario $\A$. Separating the blocked mass and the free mass of all the scenarios, we obtain a decomposition $\pcentral = \pblocked + \pfree$, where $\pblocked_A = \max\{\pcentral_A-r , 0\}$ and $\pfree_A = \pcentral_A - \pblocked_A = \min\{\pcentral_A, r\}$ for every scenario $A \in \A$.

\paragraph{Estimating \boldmath $\Pfree$.}
To define our proxy function, we will need an estimate of $\Pfree$ that is accurate within
a $(1+\ve)$ factor. Lemma~\ref{lem:free_lb} shows that $\Pfree\geq r$, which suggests that
such an estimate can be obtained with high probability using
$\poly\bigl(\frac{1}{r\ve}\bigr)$ samples.
We prove a few simple results below leading up to this (Lemma~\ref{lem:compute_Pfree}).

\begin{lemma}
\label{lem:free_lb}
We have $\Pfree \ge r$.
\end{lemma}

\begin{proof}
If there exists a scenario $A \in \A$ with $\pfree_A \ge r$, then we have $\Pfree \ge \pfree_A \ge r$. Otherwise, we have $\Pfree = \sum_{A \in \A} \pfree_A = \sum_{A \in \A} \pcentral_A = 1 \ge r$. 
\end{proof}

We partition the set of scenarios $\A$ into a set of {\em frequent scenarios} $\Afreq := \{A \in \A : \pcentral_A
\ge r\}$ and a set of {\em rare scenarios} $\Arare := \{A \in \A : \pcentral_A < r\}$. Note that
    $|\Afreq|\leq\frac{1}{r}$, and $\pfree_A =\pcentral_A$ for every scenario $A \in \Arare$. 

\begin{lemma}
\label{lem:est_free}
Consider a partition $\A = \Ahatfreq \cup \Ahatrare$ of the scenarios, with $\Afreq
\subseteq \Ahatfreq$ (and hence $\Ahatrare \subseteq \Arare$). Let $\phat$ be a
probability distribution such that $\sum_{A \in \Ahatfreq} |\phat_A - \pcentral_A| \le
\frac{1}{4} \veprime r$. Let $\Qfree := \sum_{A \in \Ahatfreq} \min\{\phat_A, r\} +
\sum_{A \in \Ahatrare} \phat_A$ and $\Phatfree := \min\left\{\Qfree + \frac{1}{2} \veprime
r, 1\right\}$. Then $\Pfree \le \Phatfree  \le \min\{(1 + \veprime) \Pfree, 1\}$. 
\end{lemma}

\begin{proof}
We first show that the first sum in the definition of $\Qfree$ is a good estimate of the amount of free mass in $\Ahatfreq$. We have
\begin{equation}
\label{eq:est_freq}
\biggl|\sum_{A \in \Ahatfreq}\min\{\phat_A, r\} - \sum_{A \in \Ahatfreq} \pfree_A \biggr| \le \sum_{A \in \Ahatfreq} \left | \min\{\phat_A, r\} - \pfree_A\right | \le \sum_{A \in \Ahatfreq} \left |  \phat_A - \pcentral_A \right | \le \frac{1}{4}  \veprime r \ ,
\end{equation}
where the first step uses the triangle inequality; the second step uses the definition of $\pfree$; the third step is by assumption.

Now we show that the second sum in the definition of $\Qfree$ is a good estimate of the amount of free mass in $\Ahatrare$. We have
\begin{equation}
\label{eq:est_rare}
\biggl|\sum_{A \in \Ahatrare} \phat_A - \sum_{A \in \Ahatrare} \pfree_A \biggr| 
= \biggl|\sum_{A \in \Ahatrare} \phat_A - \sum_{A \in \Ahatrare} \pcentral_A \biggr| 
= \biggl| \sum_{A \in \Ahatfreq} \phat_A - \sum_{A \in \Ahatfreq} \pcentral_A \biggr| 
\le \sum_{A \in \Ahatfreq}  \left |  \phat_A - \pcentral_A \right | \le \frac{1}{4} \veprime r \ ,
\end{equation}
where the first step uses the fact that $\Ahatrare \subseteq \Arare$; the second step uses the fact that $\pcentral$ and $\phat$ are probability distributions; the third step uses the triangle inequality; the fourth step is by assumption.

Combining \eqref{eq:est_freq} and \eqref{eq:est_rare} yields $|\Pfree - \Qfree| \le \frac{1}{2}\veprime r$. This, combined with Lemma \ref{lem:free_lb} and the definition of $\Phatfree$, yields the result.
\end{proof}

\begin{lemma}
\label{lem:find_freq}
Let $\phat$ be an empirical estimate of $\pcentral$ using $N = \poly(\frac{1}{r}, \log\left(\frac{1}{\delta}\right))$ samples, and let $\Ahatfreq := \{A \in \A : \phat_A \ge \frac{r}{2}\}$. Then we have $|\Ahatfreq| \le \frac{2}{r}$, and with probability at least $1 - \delta$ we have $\Afreq \subseteq \Ahatfreq$.
\end{lemma}

\begin{proof}
The inequality $|\Ahatfreq| \le \frac{2}{r}$ follows from the definition of $\Ahatfreq$ and the fact that $\phat$ is a probability distribution.

Since $\pcentral$ is a probability distribution and $\pcentral_A \ge r$ for every $A \in
\Afreq$, we have $|\Afreq| \le \frac{1}{r}$. If we choose $N$ appropriately, by using
Chernoff bounds we have 
$\Pr\left[|\phat_A - \pcentral_A | > \frac{r}{2}\right]\le\delta r$ 
for any fixed scenario $A \in \A$. It follows that for any fixed scenario $A \in \Afreq$, we have
$\Pr\bigl[A \not \in \Ahatfreq\bigr] \le \delta r$. By the union bound, we have
$\Pr\bigl[\Afreq \not \subseteq \Ahatfreq\bigr] \le |\Afreq| \delta r \le \delta$. 
\end{proof}

\begin{lemma}
\label{lem:compute_Pfree}
We can compute an estimate $\Phatfree$ of $\Pfree$ such that $\Pfree \le \Phatfree  \le \min\{(1
+ \veprime) \Pfree, 1\}$ with probability at least $1 - 2\delta$ in time $\poly(\inpsize, \frac{1}{\ve r}, \log\left(\frac{1}{\delta}\right))$. 
\end{lemma}

\begin{proof}
First, we use Lemma \ref{lem:find_freq} to obtain a set of scenarios $\Ahatfreq$ of size $|\Ahatfreq| \le \frac{2}{r}$ that is a superset of $\Afreq$ with probability at least $1- \delta$. Next, we compute a empirical estimate $\phat$ of $\pcentral$ using $N$ samples. Using Chernoff bounds, we can choose $N = \poly(\frac{1}{\ve r}, \log\left(\frac{1}{\delta}\right))$ so that $\Pr\left[|\phat_A - \pcentral_A| > \frac{1}{4} \frac{1}{|\Ahatfreq|} \veprime r \right] \le \frac{1}{|\Ahatfreq|}\delta$ for every scenario $A \in \A$. By the union bound, this event does not happen for any of the scenarios $A \in \Ahatfreq$ with probability at least $1 - |\Ahatfreq| \frac{1}{|\Ahatfreq|}\delta = 1 - \delta$. In this case, the probability distribution $\phat$ and the partition $(\Ahatfreq, \Ahatrare := \A \setminus \Ahatfreq)$ of $\A$ satisfy the conditions of Lemma \ref{lem:est_free}, and so we can compute $\Phatfree$ as described in that lemma.

The success probability is at least $(1 - \delta)^2\geq 1-2\dt$.
\end{proof}

\paragraph{A proxy function for \boldmath $\hpcentral{x}$.} 
We assume in the sequel that the estimate $\Phatfree$ computed in Lemma \ref{lem:compute_Pfree}   
satisfies $\Pfree\leq\Phatfree\leq\min\{(1+\ve)\Pfree,1\}$.
Consider the polytope 
$\K :=\bigl\{q \in \Rplus^{\A} : \sum_{A \in \A} q_A \le \Phatfree, \quad q_A \le r \ \forall A \in \A\bigr\}$. 
Our proxy function is then defined as
\[
\hproxy{x} := c^\T x + \E[A \sim \pcentral]{g(x,A)} + 
\underbrace{\max_{q \in \K} \sum_{A \in \A} q_A g(x,A)}_{\text{\small (\maxinprox)}}.
\]
Informally, $\E[A \sim \pcentral]{g(x,A)}$ and $\max_{q \in \K} \sum_{A \in \A} q_A
g(x,A)$ can be seen as upper bounds on the contributions to
$\max_{q:L_\infty(\pcent,q)\leq r}\E[A\sim q]{g(x, A)}$ from the blocked mass and the free
mass of $\pcentral$ respectively. 
We will argue that this proxy function $\hproxy{x}$ is a good pointwise approximation of
$\hpcentral{x}$. First, we need the following preliminary lemma. 

\begin{lemma}
\label{lem:scaled_feasible}
For every $x \in \Pc$, we have $\max_{q : \|\pcentral-q\|_\infty \le r} \E[A \sim q]{g(x,A)} \ge \frac{1}{1+\ve}\max_{q \in \K} \sum_{A \in \A} q_A g(x,A)$.
\end{lemma}

\begin{proof} 
Let $q^*$ be an optimal solution to (\maxinprox). We prove that there exists a distribution $\tq$ with $\|\pcentral-\tq\|_\infty \le r$ such that $\tq \ge \frac{1}{1+\ve}q^*$. This yields the result, since we obtain
\[
\max_{q : \|\pcentral-q\|_\infty \le r} \E[A \sim q]{g(x,A)} \ge \E[A \sim \tq]{g(x, A)} \ge \frac{1}{1+\ve}\sum_{A \in \A} q^*_A g(x, A) .
\]

We give a constructive proof of the existence of $\qtilde$, via an iterative
algorithm. 
Recall that $\pblocked := (\max\{\pcent_A - r, 0\})_{A \in \A}$ denotes the blocked mass of the distribution $\pcentral$. We
start by setting $\qtilde_A := \pblocked + \frac{1}{1+\ve}q^*$. Note
that for all $A \in \A$ we have 
$\qtilde_A \ge \pblocked_A \ge \pcentral_A - r$ and $\qtilde_A \ge \frac{1}{1+\ve}q^*_A$. From now on, we will only increase components of $\qtilde$, so these two properties will be conserved; therefore we maintain the invariant $\tq \ge \frac{1}{1+\ve}q^*$. We only need to work towards ensuring that $\qtilde$ is a
probability distribution and that $\qtilde_A \le \pcentral_A + r$ for every $A \in \A$ (which, along with $\qtilde_A \ge \pcentral_A - r$ for every $A \in \A$, implies $\|\pcentral-\tq\|_\infty\le r$).

Note that for every $A \in \A$ we have $\qtilde_A \le \max\{\pcent_A,r\} \le 1$ (which also implies $\qtilde_A \le \pcentral_A +
r$). Moreover, we have $\sum_{A \in \A} \qtilde_A = \sum_{A \in A} \pblocked_A + \frac{1}{1+\ve} \sum_{A \in A} q^*_A \le \sum_{A \in A} \pblocked_A + \Pfree = 1$. It is possible that $\qtilde$ is
not a probability distribution yet, if this inequality is not tight. If this is the case,
then there must be a scenario $A \in \A$ such that $\qtilde_A < \pcentral_A$. We increase
the component $\qtilde_A$ until either we obtain $\sum_{A \in \A} \qtilde_A = 1$ (and
hence $\qtilde$ is a probability distribution) or $\qtilde_A = \pcentral_A + r$. If
$\qtilde$ is still not a probability distribution we repeat the same step with a different
scenario. As each step (except possibly the final one) decreases the number of scenarios $A$ such that $\qtilde_A <
\pcentral_A$, this process eventually stops. At this moment, $\qtilde$ is a probability
distribution and satisfies $\pcentral_A - r \le \qtilde_A \le \pcentral_A + r$ for every
$A$, and so $\|\pcentral-\tq\|_\infty \le r$. 
\end{proof}

\begin{lemma} \label{pwclose}
For every $x \in\Pc$, we have $\hpcentral{x} \le \hproxy{x} \le 2(1 + \veprime) \hpcentral{x}$.
\end{lemma}

\begin{proof}
We start by proving the first inequality. Let $\qstar := \argmax_{q : \|\pcentral-q\|_\infty \le r} \E[A \sim q]{g(x,A)}$, so that $\hpcentral{x} = c^\T x + \E[A \sim
  \qstar]{g(x,A)}$. We decompose $\qstar$ into two vectors as follows: we write $\qstar =
q^1 + q^2$, where $q^1_A := \min\{\qstar_A, \pcentral_A\}$ and $q^2_A := \qstar_A - q^1_A$
for every scenario $A \in \A$. Next we upper bound the contribution of each of these two
vectors to the objective value $\hpcentral{x}$. Since $q^1 \le \pcentral$, we have
$\sum_{A \in \A} q^1_A g(x, A) \le \E[A \sim \pcentral]{g(x,A)}$. Note that since
$\|\pcentral-\qstar\|_\infty \le r$, and by the way we defined $q^2$, we must have $q^2_A
\le r$ for every scenario $A \in \A$. Further, we have $\sum_{A \in \A} q^2_A \le \Pfree
\le \Phatfree$. It follows that $q^2 \in \K$, and so $\sum_{A \in \A} q^2_A g(x, A) \le
\max_{q \in \K} \sum_{A \in \A} q_A g(x,A)$. Therefore we have 
\[
\hpcentral{x} = c^\T x + \sum_{A \in \A} q^1_A g(x, A) + \sum_{A \in \A} q^2_A g(x, A) \le \hproxy{x} \ ,
\]
proving the first inequality.

Now we proceed to prove the second inequality. We have
\begin{align*}
\hpcentral{x} & = c^\T x + \max_{q : \|\pcentral-q\|_\infty \le r} \E[A \sim q]{g(x,A)} \\
& \ge c^\T x + \frac{1}{2} \E[A \sim \pcentral]{g(x,A)} + \frac{1}{2} \left( \frac{1}{1+\ve} \max_{q \in \K} \sum_{A \in \A} q_A g(x,A) \right) \ge \frac{1}{2(1+\ve)} \hproxy{x} .
\end{align*}
The second step uses Lemma~\ref{lem:scaled_feasible} and the fact that $\max_{q : \|\pcentral-q\|_\infty \le r} \E[A \sim q]{g(x,A)} \ge \E[A \sim \pcentral]{g(x,A)}$ (since $\pcentral$ is feasible for the maximization problem on the left side).
\end{proof}

\paragraph{Solving the proxy problem \boldmath $\min_{x\in\Pc}\hproxy{x}$.}

We assume that for all $x\in X$, and all $A\sse A'$, we have $g(x,A)\leq g(x,A')$, which
holds for all covering problems. Recall that $\Pc\sse\Rplus^m$. 
Recall from property \ref{p4} that for every $A \in \A$, the function $g(\cdot,A)$ is convex, and at every $x\in\Pc$ 
we can efficiently compute its value. We will assume the following stronger version of
\ref{p5}: 

\begin{enumerate}[label=(P5'), topsep=0.25ex, itemsep=0ex, parsep=0ex,
    labelwidth=\widthof{(P5')}, leftmargin=!]
\item \label{p5p}
For every $x\in\Pc$ and $A\in\A$, we can efficiently compute a subgradient $\sgr^{x,A}$ of
$g(\cdot,A)$ at $x$ with $-\ld c \le \sgr^{x,A} \le 0$. 
\end{enumerate} 

Shmoys and Swamy~\cite{ShmoysS06} define a broad class of 2-stage problems for
which~\ref{p5p} holds, which includes all the 2-stage problems considered in the literature.
Recall that by \ref{p3}, 
$\Pc\sse B(\bo,R)=\{x:\|x\|\leq R\}$ and $\Pc$ contains a ball of radius $V\leq 1$ such
that $\ln\bigl(\frac{R}{V}\bigr)=\poly(\inpsize)$.  
Let $\tK$ be the Lipschitz constant of $\hproxy{\cdot}$; we show in Lemma~\ref{LinftyLip}
that $\log\tK=\poly(\inpsize)$. 
Under this setup, we have the following result from~\cite{ShmoysS06}.

\begin{theorem}[see Theorem 4.7, Lemma 4.14 in~\cite{ShmoysS06}] \label{convsolve}
Let $\ve<1/2$, $\dt>0$. 
Define $N=\ceil{2m^2\ln\bigl(\frac{16KR^2}{V\kp}\bigr)}$ and
$n=N\ln\bigl(\frac{8NKR}{\kp}\bigr)$, and 
$\w=\ve/2n=\poly\bigl(\frac{\ve}{\inpsize},\log(\frac{1}{\kp})\bigr)$.
Suppose we have a procedure that given any point $x\in\Pc$ finds an $\w$-subgradient 
of $\hproxy{\cdot}$ at $x$ with probability at least $1-\dt$ in time
$T(\w,\dt)$. 
Then, we can find $\bx\in\Pc$ satisfying
$\hproxy{\bx}\leq\frac{1}{1-\ve}\cdot\min_{x\in\Pc}\hproxy{x}+\kp$ with probability at
least $1-\dt$ in time
$O\bigl(T(\w,\frac{\dt}{N+n})\cdot m^2\log^2(\frac{\tK Rm}{V\kp})\bigr)
=\poly\bigl(\inpsize, T(\w,\frac{\dt}{N+n}),\log(\frac{1}{\kp})\bigr)$.  
\end{theorem}

We show that one can compute an $\w$-subgradient with probability at least $1-\dt$ in time
$T(\w,\dt)=\poly\bigl(\inpsize, \frac{\ld}{r\w},\log(\frac{1}{\dt})\bigr)$. 
Lemma~\ref{proxysubgrad} (ii) shows that to obtain an $\w$-subgradient, it suffices to be
able to (a) find a vector that is componentwise close to $\E[A\sim\pcent]{\sgr^{x,A}}$, and
(b) find an optimal solution to the maximization problem (\maxinprox) in the definition of
$\hproxy{x}$.
Lemma~\ref{estsgr} argues using simple Chernoff bounds that one can obtain a vector that
is componentwise close to $\E[A\sim\pcent]{\sgr^{x,A}}$, and Lemma~\ref{Kopt} shows that
one can compute an optimal solution to (\maxinprox) (with polynomial support). Finally,
Lemma~\ref{LinftyLip} bounds the Lipschitz constant of $\hproxy{\cdot}$. Putting 
everything together yields Theorem~\ref{linftysolve}. 

\begin{lemma} \label{proxysubgrad}
\begin{enumerate}[(i), topsep=0ex, itemsep=0.25ex, parsep=0ex, labelwidth=\widthof{(ii)},
    leftmargin=!] 
\item The function $\hproxy{\cdot}$ is convex, and the vector
$d:=c+\E[A\sim\pcent]{\sgr^{x,A}}+\sum_{A\in\A}q^*_A\sgr^{x,A}$ is a subgradient of
$\hproxy{\cdot}$ at $x$; here $q^*$ is an optimal solution to 
\emph{(\maxinprox)}. 

\item Moreover, if $\sgrest$ is a vector such that 
$-\w c \leq\sgrest - \E[A\sim\pcent]{\sgr^{x,A}}\leq 0$, then 
$\hd:=c+\sgrest+\sum_{A\in\A}q^*_A\sgr^{x,A}$ is an $\w$-subgradient of $\hproxy{\cdot}$ at
$x$.  
\end{enumerate}
\end{lemma}

\begin{proof}
Convexity of $\hproxy{\cdot}$ will follow from the fact that we have a subgradient of
$\hproxy{\cdot}$ at every point $x\in\Pc$. Part (i) is a special case of part (ii) with
$\w=0$, so we focus on part (ii). Consider any $x'\in\Pc$. We have
\begin{equation*}
\begin{split}
\hproxy{x'}-\hproxy{x} & \geq 
c^\T(x'-x)+\E[A\sim\pcentral]{g(x',A)-g(x,A)}+\sum_{A\in\A}q^*_A\bigl(g(x',A)-g(x,A)\bigr)\\
& \geq
c^\T(x'-x)+\E[A\sim\pcent]{\sgr^{x,A}\cdot(x'-x)}+\sum_{A\in\A}q^*_A\sgr^{x,A}\cdot(x'-x) \\
& \geq
\bigl(c+\sum_{A\in\A}q^*_A\sgr^{x,A}\bigr) \cdot (x'-x)+\sgrest \cdot (x'-x)
+\sum_{e:x'_e<x_e}(x'_e-x_e)\w c_e \\
& \geq \hd^\T(x'-x)-\w c^\T x \geq \hd^\T(x'-x)-\w\cdot\hproxy{x}.
\end{split}
\end{equation*}
The first inequality follows since $q^*$ is a feasible solution to (\maxinprox[x']); the second
follows since $\sgr^{x,A}$ is a subgradient of $g(\cdot,A)$ at $x$; the third follows from
the componentwise closeness of $\sgrest$ and $\E[A\sim\pcent]{\sgr^{x,A}}$; the fourth
follows since $x,x'\geq 0$, and the last inequality is because $\hproxy{x}\geq c^\T x$.
\end{proof}

\begin{lemma} \label{estsgr}
Let $x\in\Pc$. For any $\w>0$ and $\dt \in (0, 1)$, we can compute a vector $\sgrest$ such that
$-\w c \leq\sgrest - \E[A\sim\pcent]{\sgr^{x,A}}\leq 0$ with probability
at least $1-\dt$ in time $T(\w,\dt):=\poly(\inpsize, \frac{\ld}{\w}, \log \frac{1}{\delta})$.
\end{lemma}

\begin{proof}
This is a simple application of Chernoff-Hoeffding bounds. For $i=1,\ldots,\Nc$, we sample a
scenario $A$ from $\pcent$, and compute $Z^i=\sgr^{x,A}$, so $Z^i_e/\ld c_e\in[-1,0]$ for every $e = 1, \dots, m$ by \ref{p5p}. 
Taking the average of $\Nc$ independent samples, we obtain using Chernoff bounds (see
Theorem 1.1 in~\cite{DubhashiP}), that 
$$
\Pr\biggl[\biggl|\frac{1}{\Nc}\cdot\sum_{i=1}^\Nc\frac{Z^i_e}{\ld c_e}
-\frac{\E[A\sim\pcent]{\sgr^{x,A}_e}}{\ld c_e}\biggr|>\frac{\w}{2\ld}\biggr]
\leq 2\exp\Bigl(-\tfrac{2\w^2}{4\ld^2}\cdot\Nc\Bigr)
$$
for every $e = 1, \dots, m$. 
So $\Nc=\frac{2\ld^2}{\w^2}\ln\bigl(\frac{2m}{\dt}\bigr)$ ensures that the above probability is
at most $\dt/m$. We return $\sgrest = \frac{1}{\Nc}\sum_{i=1}^\Nc Z^i -\frac{1}{2}\w c$. By the union bound, this satisfies $-\w c \leq\sgrest - \E[A\sim\pcent]{\sgr^{x,A}}\leq 0$ with probability at least $1 - m \frac{\delta}{m} = 1 - \delta$.
\end{proof}

We say $(A_1, \dots, A_\kinf)$ is a {\em good $\kinf$-sequence} for $x$ if $A_1,\ldots,A_\kinf$ are
the $\kinf$ scenarios with maximum second-stage cost $g(x,A)$ in that order; i.e., more
precisely, we have 
$g(x,A_1) \ge g(x, A_2) \ge \dots \ge g(x, A_\kinf)\geq\max_{A\in\A\sm\{A_1,\ldots,A_\kinf\}}g(x,A)$.

\begin{lemma} \label{lem:optimalq} \label{Kopt}
Let $\kinf := \min\bigl\{\bigl\lceil\Phatfree/r\bigr\rceil,|\A|\bigr\}$, and fix $x \in \Pc$. Suppose that $g(x,A)\leq g(x,A')$ for all $A\sse A'$. 
\begin{enumerate}[(a), topsep=0.25ex, itemsep=0ex, parsep=0ex, labelwidth=\widthof{(b)},
    leftmargin=!] 
\item We can compute a good $\kinf$-sequence $(A_1, \dots, A_\kinf)$ in time $\poly(\inpsize, \kinf)$.
\item Define the vector $\qx$ as follows: 
\[
\qx_{A} := \begin{cases}
r & \text{if}\ A\in\{A_1, \dots, A_{\kinf-1}\}; \\
\min\Bigl\{r,\Phatfree - (\kinf-1)r\Bigr\} \quad & \text{if}\ A = A_\kinf; \\
0 & \text{otherwise.} 
\end{cases}
\]
Then $\qx$ is an optimal solution to $\max_{q \in \K} \sum_{A \in \A} q_A g(x,A)$.
\end{enumerate}
\end{lemma}

\begin{proof}
By the monotonicity assumption of $g(x, \cdot)$, the costliest scenario is $U$, so we start by setting $A_1 = U$. We then proceed as follows for $i = 2, \dots, \kinf$. Suppose that we have already computed $A_1, \dots, A_{i-1}$. Computing $A_i$ amounts to solving the problem
\begin{gather}
\label{max_i}
\max_{A \in \A \setminus \{A_1, \dots, A_{i-1}\}} g(x, A) .
\end{gather}
We claim that~\eqref{max_i} admits an optimal solution that is a maximal proper subset of $A_{i'}$ for some $1 \le i' < i$. Indeed, let $A^*$ be an optimal solution of~\eqref{max_i} with maximum cardinality, and suppose for a contradiction that it is not a maximal proper subset of $A_{i'}$ for any $1 \le i' < i$. Note that since $A_1 = U$, we have $A^* \neq U$, so there is an element $e \in U \setminus A^*$. Now, consider the scenario $\bA := A^* \cup \{e\}$. Since by assumption $A^*$ is not a maximal subset of $A_{i'}$ for any $1 \le i' < i - 1$, it follows that $\bA$ is feasible for~\eqref{max_i}. By the monotonicity assumption, since $A^* \subseteq \bA$, we have $g(x, \bA) \ge g(x, A^*)$, and so $\bA$ is also an optimal solution for~\eqref{max_i}. Since $|\bA| > |A^*|$, this contradicts the definition of $A^*$.

We now utilize the observation above to show that given $x$ and $A_1, \dots, A_{i-1}$, we can solve~\eqref{max_i} in $\poly(\inpsize, i)$ time. This can be done by enumerating all maximal proper subsets of $A_1, \dots, A_{i-1}$. Since each set $A_{i'}$ has $|A_{i'}|$ maximal proper subsets, we enumerate $\sum_{i'=1}^{i-1} |A_{i'}| \le (i-1)|U| = \poly(\inpsize, i)$ scenarios, and the claim follows. 
We conclude that we can compute a good $\kinf$-sequence by solving~\eqref{max_i} for $i = 2, \dots, \kinf$, which takes $\sum_{i = 2}^{\kinf} \poly(\inpsize, i) = \poly(\inpsize, \kinf)$ time.


For part (b), 
consider the polytope $\frac{1}{r}\K := \{\frac{1}{r} q : q \in \K\}$.
Note that the problem $\max_{q \in \K} \sum_{A \in \A} q_A g(x,A)$ is equivalent to the
problem $\max_{q \in \frac{1}{r}\K} \sum_{A \in \A} q_A g(x,A)$ (up to scaling of the
solutions), which can be seen as a fractional knapsack problem: we have one item of value
$g(x, A)$ and weight $1$ for every $A \in \A$; the capacity of the knapsack is set to
$\frac{\Phatfree}{r}$. The result then follows by using the fact that one can compute an
optimal solution to a fractional knapsack problem in a greedy fashion, by repeatedly
picking among the available items the one with the highest value/weight ratio. 
\end{proof}

\begin{lemma}
\label{LinftyLip}
The function $\hproxy{\cdot}$ has Lipschitz constant at most $\tK=(2\ld + 1) \|c\|$.
\end{lemma}

\begin{proof}
It suffices to show that $\hproxy{\cdot}$ admits a subgradient of Euclidean norm at most $\tK$ at every point $x \in \Pc$. Fix $x \in \Pc$, and consider the subgradient $d := c+\E[A\sim\pcent]{d^{x,A}}+\sum_{A\in\A}q^*_A\sgr^{x,A}$ given by Lemma~\ref{proxysubgrad}. We have
\begin{align*}
\|d\| \le \|c\|+\sum_{A\in\A} \pcent_A\|\sgr^{x,A}\|+\sum_{A\in\A}q^*_A \|\sgr^{x,A}\| \le (2\ld + 1)\|c\| .
\end{align*}
The first step follows from the triangle inequality, and the final step follows because $\|\sgr^{x,A}\| \le \ld \|c\|$ for every $A \in \A$ by assumption~\ref{p5p} and $\sum_{A \in \A}q^*_A \le \Phatfree \le 1$.
\end{proof}

\begin{proofof}{Theorem~\ref{linftysolve}}
Note that $g(\bo,U)=\max_{A \in \A}g(\bo,A)$ by the monotonicity property of the
second-stage costs. If $g(\bo,U)=0$ (so $\A$ contains only null scenarios) then
$\max_{q:\|\pcent-q\|_\infty\leq r}\E[A\sim q]{g(x,A)}=0$, and so $x=\bo$ is an optimal
solution to the \dr problem. Otherwise, the optimal value of (\maxinprox) is at least
$\lb:=r\cdot g(\bo,U)$ since there is always a distribution $q$ with
$\|q-\pcent\|_\infty\leq r$ that places a weight of at least $r$ on $U$ (e.g., take
$q=\pcent$ if $\pcent_U\geq r$; otherwise, take $q_U=r$,
$q_A=(1-r)\pcent_A/\sum_{A'\subsetneq U}\pcent_{A'}$ for all $A\subsetneq U$).
Note that $\log\bigl(\frac{1}{\lb}\bigr)=\poly(\inpsize)$.  

We compute a $(1+\ve)$-estimate of $\Pfree$ using Lemma~\ref{lem:compute_Pfree}. 
We then run the algorithm Theorem~\ref{convsolve}, utilizing
Lemmas~\ref{proxysubgrad}--\ref{lem:optimalq} to compute $\w$-subgradients, and setting $\kp = \ve \cdot \lb$ and $\tK = (2\ld+1)\|c\|$ (using Lemma~\ref{LinftyLip}). Let $\bx$ be the solution returned. Using
Lemma~\ref{pwclose}, we obtain that 
$$
\hpcentral{\bx}\leq\hproxy{\bx}\leq\frac{1}{1-\ve}\cdot\min_{x\in\Pc}\hproxy{x}+\kp
\leq\biggl(\frac{2(1+\ve)}{1-\ve} + \ve\biggr)\cdot\min_{x\in\Pc}\hpcentral{x}
\leq\bigl(2+O(\ve)\bigr)\cdot\min_{x\in\Pc}\hpcentral{x}
$$
where $\frac{2+2\ve}{1-\ve}\leq 2+4\ve$ since $\ve\leq\frac{1}{3}$.
The success probability is at least $1-3\dt$.
\end{proofof}

\bibliographystyle{plain}
\bibliography{drso_arxiv}

\begin{thebibliography}{10}

\bibitem{AgarwalDSY10}
Shipra Agrawal, Yichuan Ding, Amin Saberi, and Yinyu Ye.
\newblock {Price of Correlations in Stochastic Optimization}.
\newblock {\em Operations Research}, 60(1):150--162, 2012.

\bibitem{BertsimasSZ}
D.~Bertsimas, M.~Sim, and M.~Zhang.
\newblock {A practicable framework for distributionally robust linear
  optimization}.
\newblock {\em optimization-online.org}, 2013.

\bibitem{BirgeL97}
John~R. Birge and Fran{\c c}ois Louveaux.
\newblock {\em {Introduction to Stochastic Programming}}.
\newblock Springer Science {\&} Business Media, June 2011.

\bibitem{CharikarCP05}
Moses Charikar, Chandra Chekuri, and Martin P\'{a}l.
\newblock Sampling bounds for stochastic optimization.
\newblock In {\em Proceedings of the 8th International Workshop on
  Approximation, Randomization and Combinatorial Optimization Problems
  (APPROX)}, pages 257--269, 2005.

\bibitem{DelageY10}
Erick Delage and Yinyu Ye.
\newblock {Distributionally Robust Optimization Under Moment Uncertainty with
  Application to Data-Driven Problems.}
\newblock {\em Operations Research}, 58(3):595--612, 2010.

\bibitem{DhamdhereGRS05}
Kedar Dhamdhere, Vineet Goyal, R.~Ravi, and Mohit Singh.
\newblock How to pay, come what may: Approximation algorithms for demand-robust
  covering problems.
\newblock In {\em Proceedings of the 46th Annual IEEE Symposium on Foundations
  of Computer Science (FOCS)}, pages 367--378, 2005.

\bibitem{DubhashiP}
Devdatt Dubhashi and Alessandro Panconesi.
\newblock {\em Concentration of Measure for the Analysis of Randomized
  Algorithms}.
\newblock Cambridge University Press, New York, NY, USA, 1st edition, 2009.

\bibitem{ErdoganI}
Emre Erdo{\u{g}}an and Garud Iyengar.
\newblock {Ambiguous chance constrained problems and robust optimization}.
\newblock {\em Math. Program.}, 107(1-2):37--61, December 2005.

\bibitem{EsfahaniK17}
Peyman~Mohajerin Esfahani and Daniel Kuhn.
\newblock {Data-driven Distributionally Robust Optimization Using the
  Wasserstein Metric: Performance Guarantees and Tractable Reformulations}.
\newblock {\em arXiv.org}, May 2015.

\bibitem{Esfandiari15}
H.~Esfandiari, N.~Korula, and V.~Mirrokni.
\newblock Online allocation with traffic spikes: Mixing adversarial and
  stochastic models.
\newblock In {\em Proceedings of the 16th ACM Conference on Economics and
  Computation (EC)}, pages 169--186, 2015.

\bibitem{FeigeJMM05}
Uriel Feige, Kamal Jain, Mohammad Mahdian, and Vahab Mirrokni.
\newblock Robust combinatorial optimization with exponential scenarios.
\newblock In {\em Proceedings of the 12th International Conference on Integer
  Programming and Combinatorial Optimization (IPCO)}, pages 439--453, 2007.

\bibitem{FriggstadS15}
Zachary Friggstad and Chaitanya Swamy.
\newblock Approximation algorithms for regret-bounded vehicle routing and
  applications to distance-constrained vehicle routing.
\newblock In {\em Proceedings of the 46th Annual ACM Symposium on Theory of
  Computing (STOC)}, pages 744--753, 2014.

\bibitem{GaoK}
Rui Gao and Anton~J. Kleywegt.
\newblock {Distributionally Robust Stochastic Optimization with Wasserstein
  Distance}.
\newblock {\em arXiv.org}, April 2016.

\bibitem{GrotschelLS88}
Martin Gr{\"o}tschel, L{\'a}szl{\'o} Lov{\'a}sz, and Alexander Schrijver.
\newblock {\em {Geometric Algorithms and Combinatorial Optimization}}.
\newblock Springer-Verlag, 1988.

\bibitem{GuptaNRb}
A.~Gupta, V.~Nagarajan, and R.~Ravi.
\newblock {Robust and MaxMin Optimization under Matroid and Knapsack
  Uncertainty Sets}.
\newblock {\em Transactions on Algorithms}, 12(1), 2012.

\bibitem{GuptaRS}
A.~Gupta, R.~Ravi, and A.~Sinha.
\newblock An edge in time saves nine: Lp rounding approximation algorithms for
  stochastic network design.
\newblock In {\em Proceedings of the 45th Annual IEEE Symposium on Foundations
  of Computer Science (FOCS)}, pages 218--227, 2004.

\bibitem{GuptaNR}
Anupam Gupta, Viswanath Nagarajan, and R.~Ravi.
\newblock Thresholded covering algorithms for robust and max-min optimization.
\newblock In {\em Proceedings of the 37th International Colloquium Conference
  on Automata, Languages and Programming (ICALP)}, pages 262--274, 2010.

\bibitem{GuptaPRS04}
Anupam Gupta, Martin P\'{a}l, R.~Ravi, and Amitabh Sinha.
\newblock Boosted sampling: Approximation algorithms for stochastic
  optimization.
\newblock In {\em Proceedings of the 36th Annual ACM Symposium on Theory of
  Computing (STOC)}, pages 417--426, 2004.

\bibitem{SGupta}
S.~Gupta.
\newblock {Building Networks in the Face of Uncertainty}.
\newblock Master's thesis, University of Waterloo, 2011.

\bibitem{HanasusantoK18}
Grani~A. Hanasusanto and Daniel Kuhn.
\newblock {Conic Programming Reformulations of Two-Stage Distributionally
  Robust Linear Programs over Wasserstein Balls.}
\newblock {\em Operations Research}, 66(3):849--869, 2018.

\bibitem{HuangMR15}
Zhiyi Huang, Yishay Mansour, and Tim Roughgarden.
\newblock Making the most of your samples.
\newblock In {\em Proceedings of the 16th ACM Conference on Economics and
  Computation (EC)}, pages 45--60, 2015.

\bibitem{JainMS03}
Kamal Jain, Mohammad Mahdian, Evangelos Markakis, Amin Saberi, and Vijay~V.
  Vazirani.
\newblock {Greedy Facility Location Algorithms Analyzed Using Dual Fitting with
  Factor-Revealing LP.}
\newblock {\em J. ACM}, 50(6):795--824, 2003.

\bibitem{Khandekar}
Rohit Khandekar, Guy Kortsarz, Vahab Mirrokni, and Mohammad~R. Salavatipour.
\newblock {Two-stage Robust Network Design with Exponential Scenarios.}
\newblock {\em Algorithmica}, 65(2):391--408, 2013.

\bibitem{HomemKS}
Anton~J. Kleywegt, Alexander Shapiro, and Tito~Homem de~Mello.
\newblock {The Sample Average Approximation Method for Stochastic Discrete
  Optimization.}
\newblock {\em SIAM Journal on Optimization}, 12(2):479--502, 2002.

\bibitem{Li}
Shi Li.
\newblock A 1.488 approximation algorithm for the uncapacitated facility
  location problem.
\newblock {\em Information and Computation}, 222:45--58, January 2013.

\bibitem{LinharesS19}
Andr\'{e} Linhares and Chaitanya Swamy.
\newblock Approximation algorithms for distributionally robust stochastic
  optimization with black-box distributions.
\newblock In {\em Proceedings of the 51st Annual ACM Symposium on Theory of
  Computing (STOC)}, pages 768--779, 2019.

\bibitem{Mirrokni12}
Vahab~S. Mirrokni, Shayan~Oveis Gharan, and Morteza Zadimoghaddam.
\newblock Simultaneous approximations for adversarial and stochastic online
  budgeted allocation.
\newblock In {\em Proceedings of the 23rd Annual ACM-SIAM Symposium on Discrete
  Algorithms (SODA)}, pages 1690--1701, 2012.

\bibitem{PalT}
Martin P\'{a}l and \'{E}va Tardos.
\newblock Group strategyproof mechanisms via primal-dual algorithms.
\newblock In {\em Proceedings of the 44th Annual IEEE Symposium on Foundations
  of Computer Science (FOCS)}, pages 584--, 2003.

\bibitem{VanparysEK}
Bart P. G.~Van Parys, Peyman~Mohajerin Esfahani, and Daniel Kuhn.
\newblock {From Data to Decisions: Distributionally Robust Optimization is
  Optimal}.
\newblock {\em arXiv.org}, April 2017.

\bibitem{Popescu07}
Ioana Popescu.
\newblock {Robust Mean-Covariance Solutions for Stochastic Optimization.}
\newblock {\em Operations Research}, 55(1):98--112, 2007.

\bibitem{Prekopa95}
Andr{\'a}s Pr{\'e}kopa.
\newblock {\em {Stochastic Programming}}.
\newblock Kluwer Academic Publishers, 1995.

\bibitem{Ravi:2004gp}
R.~Ravi and Amitabh Sinha.
\newblock Hedging uncertainty: Approximation algorithms for stochastic
  optimization problems.
\newblock {\em Mathematical Programming}, 108(1):97--114, August 2006.

\bibitem{RuszczynskiS03}
Andrzej Ruszczy{\'{n}}ski and Alexander Shapiro.
\newblock {\em {Stochastic Programming}}, volume~10 of {\em Handbook in
  Operations Research and Management Science}.
\newblock Elsevier, 2003.

\bibitem{Scarf58}
Herbert~E. Scarf.
\newblock {A min-max solution of an inventory problem}.
\newblock {\em Studies in The Mathematical Theory of Inventory and Production},
  pages 201--209, 1958.

\bibitem{ShmoysS06}
David~B. Shmoys and Chaitanya Swamy.
\newblock {An Approximation Scheme for Stochastic Linear Programming and Its
  Application to Stochastic Integer Programs.}
\newblock {\em J. ACM}, 53(6):978--1012, 2006.

\bibitem{ShmoysTA97}
David~B. Shmoys, \'{E}va Tardos, and Karen Aardal.
\newblock Approximation algorithms for facility location problems (extended
  abstract).
\newblock In {\em Proceedings of the 29th Annual ACM Symposium on Theory of
  Computing (STOC)}, pages 265--274, 1997.

\bibitem{SwamyS06}
C.~Swamy and D.~B. Shmoys.
\newblock {Approximation algorithms for 2-stage stochastic optimization
  problems}.
\newblock {\em ACM SIGACT News}, 37(1):33--46, March 2006.

\bibitem{SwamyS12}
C.~Swamy and D.~B. Shmoys.
\newblock {Sampling-Based Approximation Algorithms for Multistage Stochastic
  Optimization.}
\newblock {\em SIAM J. Comput.}, 41(4):975--1004, 2012.

\bibitem{Swamy11}
Chaitanya Swamy.
\newblock Risk-averse stochastic optimization: Probabilistically-constrained
  models and algorithms for black-box distributions.
\newblock In {\em Proceedings of the 22nd Annual ACM-SIAM Symposium on Discrete
  Algorithms (SODA)}, pages 1627--1646, 2011.

\bibitem{WuDX15}
Chenchen Wu, Donglei Du, and Dachuan Xu.
\newblock {An Approximation Algorithm for the Two-Stage Distributionally Robust
  Facility Location Problem}.
\newblock In {\em Advances in Global Optimization}, volume~95 of {\em Springer
  Proceedings in Mathematics \& Statistics}, pages 99--107. Springer
  International Publishing, October 2014.

\bibitem{ZhaoG15}
C.~Zhao and Y.~Guan.
\newblock {Data-driven risk-averse stochastic optimization with Wasserstein
  metric}.
\newblock {\em Oper. Res. Lett.}, 46(2):262--267, March 2018.

\end{thebibliography}

\appendix \label{appstart}

\section{Proof of Theorem~\ref{mainsaathm}} \label{mainsaaproof}

\paragraph{Overview.}
Let $\ppoly$ denote a
generic empirical estimate of $\pcent$ (which could be any of $\htp^1,\ldots,\htp^k$). 
We discretize $[0,\tau]$ suitably to obtain a set $Y$ so that for any $x\in X$, and
$y\in[0,\tau]$, there is some $y'\in Y$ such that $\bhpcentral[p]{x,y'}$ is close to
$\bhpcentral[p]{x,y'}$ for any central distribution $p$ (Claim~\ref{yclose}).
It follows that approximate solutions to $\min_{x\in X,y\in Y}\bhpcentral[p]{x,y}$
translate to approximate solutions to $\min_{x\in X,y\in[0,\tau]}\bhpcentral[p]{x,y}$. 

The arguments in~\cite{CharikarCP05} can be used to show that an approximate 
solution to $\min_{x\in X, y\in Y}\bhpcentral[\ppoly]{x,y}$ can be used to obtain an
approximate solution to $\min_{x\in X,y\in Y}\bhpcentral{x,y}$ (given a
suitable value oracle for $\bhpcentral[\ppoly]{x,y}$). 
Recall that $\bhpcentral[p]{x,y}=c^\T x+ry+\E[A\sim p]{\bg(x,y,A)}$.
The proof in~\cite{CharikarCP05} proceeds by decomposing $\E[A\sim p]{\bg(x,y,A)}$ 
into two terms, $\El[A\sim p]{.}$ and $\Eh[A\sim p]{.}$, which are the contributions from
``low'' cost and ``high'' cost scenarios respectively. For the low scenarios, 
Chernoff bounds imply that $\El[A\sim\ppoly]{.}$ and $\El[A\sim\pcent]{.}$ 
are close to each other, for all $(x,y)\in X\times Y$, and all SAA problems; this is
stated in \eqref{pfineq4}. 

But the high-scenario contribution could be quite different in the SAA and original
problems, although in both problems, this contribution is essentially independent of
$(x,y)$ since the choice of ``high'' ensures that high scenarios occur with small
probability; 
this is shown by inequalities \eqref{pfineq2}, \eqref{pfineq3}.

Since $\Eh[A\sim p]{.}$ is {\em linear in $p$}, the expectation
of $\Eh[A\sim\ppoly]{.}$, over the choice of $\ppoly$, is precisely $\Eh[A\sim\pcent]{.}$. Thus, 
among our multiple SAA problems (involving empirical estimates $\htp^i$ of $\pcent$),
we can guarantee by Markov's inequality that (with high probability) for at least 
{\em one} of them, $\Eh[A\sim\htp^i]{.}$ will be close to $\Eh[A\sim\pcent]{.}$. It
follows that an $\al$-approximate solution to this SAA problem is also an
$\al\bigl(1+O(\ve)\bigr)$-approximate solution to the original problem. But we do not 
a priori know this index $i$, and evaluating or estimating $\E[A\sim\pcent]{\bg(x,y,A)}$
(and hence, $\bhpcentral{x,y}$) is challenging because (other than the difficulty of
evaluating $\bg(x,y,A)$ for a specific scenario $A$) $\pcent$ can
have exponential support; in fact, this is often {\it \#P}-hard even for standard
2-stage problems. In~\cite{CharikarCP05}, it is shown that if one can estimate the
objective value $\bhpcentral[\ppoly]{x,y}$ for the SAA problem (which seems easier since
$\ppoly$ has polynomial support), then choosing the (solution corresponding
to the) SAA problem with best SAA objective value works.

In our case, we actually want to evaluate $\hpcentral{x,y}$, or roughly equivalently (by
Lemma~\ref{zproxy}), the objective $\bhpcentral{x,y}+\zpcentlng{0}$ for the solution
returned by the SAA problem. While we can once
again decompose $\E[A\sim p]{g(x,y,A)}$ into $\El[A\sim p]{.}$ and $\Eh[A\sim p]{.}$, as
with $\Eh[A\sim p]{.}$, the term $\zpcentlng[p]{0}$ could have very different
contributions in the SAA and original problems, and we need to reason about this
separately. Moreover, a complicating factor is that this term is {\em not} linear in
$p$. We show in Claim~\ref{zlongconc} that this term is concave in $p$, and
this allows us to still use Markov's inequality as above. In the proof below, we consider the
combined term $\Eh[A\sim p]{.}+\zpcentlng[p]{0}$, and apply Markov's inequality to show that
among our multiple SAA problems, there is some index $t$ for which this term is close to
$\Eh[A\sim\pcent]{.}+\zpcentlng{0}$; see inequality \eqref{pfineq5}.

Finally, we show that, although we do not know $t$, and we do not know how to evaluate
$\hpcentral[\ppoly]{x,y}$ or $\bhpcentral[\ppoly]{x,y}$, the index $j$ corresponding to the
best $f^i$ estimate works as well as $t$; this is captured by \eqref{pfineq6}.

\paragraph{Details.}
Instead of directly working with $\hpcentral[p]{x}$ and $\hpcentral[p]{x,y}$, we will work
with the quantities $\bhpcentral[p]{x}+\zpcentlng[p]{0}$ and
$\bhpcentral[p]{x,y}+\zpcentlng[p]{0}$.
It will be cumbersome to carry around the $\zpcentlng[p]{0}$ term, so we define
$\thpcentral[p]{x}:=\bhpcentral[p]{x}+\zpcentlng[p]{0}$, and
$\thpcentral[p]{x,y}:=\bhpcentral[p]{x,y}+\zpcentlng[p]{0}$. 
To further simplify notation, we further abbreviate notation. The convention we follow is
that 
whenever there is an index $i$ in the superscript of a quantity, it refers to that
quantity for the central distribution $\htp^i$ of the $i$-th SAA problem. 
So we use 
\begin{enumerate}[label=--, topsep=0.5ex, itemsep=0.5ex, parsep=0ex] 
\item $\hh^i(x)$ and $\hh^i(x,y)$ to denote
$\hpcentral[\htpcent^i]{x}$ and $\hpcentral[\htpcent^i]{x,y}$ respectively; 
\item $\bh^i(x)$ and $\bh^i(x,y)$ to denote
$\bhpcentral[\htpcent^i]{x}$ and $\bhpcentral[\htpcent^i]{x,y}$ respectively; 
\item $\tdh^i(x)$ and $\tdh^i(x,y)$ to denote 
$\thpcentral[\htpcent^i]{x}$ and $\thpcentral[\htpcent^i]{x,y}$ respectively;
\item $\hzshort[i](x)$ and $\hzlong[i]$ to denote $\zpcentshort[\htpcent^i]{x}$ and
  $\zpcentlng[\htpcent^i]{0}$ respectively.
\end{enumerate}
\noindent
We focus on showing that 
\begin{alignat}{1}
\bhpcentral{\hx}+\zpcentlng{0} & \leq 
2\beta\rho\bigl(1+O(\ve)\bigr)\cdot\min_{x\in X}\Bigl(\bhpcentral{x}+\zpcentlng{0}\Bigr)+2\beta\rho\kp. 
\notag \\
\text{that is,} \quad 
\thpcentral{\hx} & \leq 
2\beta\rho\bigl(1+O(\ve)\bigr)\cdot\min_{x\in X}\thpcentral{x}+2\beta\rho\kp. 
\label{goalineq}
\end{alignat}
Combining this with Lemma~\ref{zproxy} completes the proof.

Let $\kp':=\frac{\kp}{2+8\ve}$. 
Define 
$Y:=\{0,\tau\}\cup\{\text{integer multiples of $\frac{\kp'}{\ld r}$ in }[0,\tau]\}$.%
\footnote{The discretization considered in~\cite{CharikarCP05} is incorrect: it assumes
implicitly that the search region of the SAA problem is (or may be) restricted to points
whose first-stage cost is within some factor of the optimum of the original problem, but
this need not hold. It also assumes that the grid points lie in the feasible region, which
again need not hold.} 
Note that $|Y|=O\bigl(\frac{\tau\ld r}{\kp'}\bigr)$. 

\begin{claim} \label{bndinfl}
The discretized 2-stage problem $\min_{x\in X,y\in Y}\bhpcentral[p]{x,y}$ satisfies 
properties \ref{p1}, \ref{p2} with inflation parameter $\Ld=\ld$, i.e., we 
have
$$
\bg(x,y,A)\leq \bg(0,0,A)\leq\bg(x,y,A)+\ld(c^\T x+ry) \qquad 
\forall A,\ \forall x\in X, y\geq 0.
$$
\end{claim}


\begin{claim} \label{yclose}
For any $x\in X$, $y\in [0,\tau]$, and any distribution $p$, there is some $y'\in Y$ such that
$\bhpcentral[p]{x,y}-\kp'\leq\bhpcentral[p]{x,y'}\leq\bhpcentral[p]{x,y}+\kp'$.
\end{claim}

\begin{proof} 
There is some $y'\in Y$ with $|y-y'|\leq\frac{\kp'}{\ld r}$. 
If $y'\geq y$, then 
$\bhpcentral[p]{x,y'}\leq\bhpcentral[p]{x,y}+r\cdot\frac{\kp'}{\ld r}
\leq\bhpcentral[p]{x,y}+\kp'$.
Also, $\bg(x,y,A)\leq\bg(x,y',A)+(y'-y)\cdot\ld r$ for all $A$, so 
$\bhpcentral[p]{x,y}\leq\bhpcentral[p]{x,y'}+\kp'$.
If $y'<y$, then we can interchange the arguments; the claim follows.
\end{proof} 

We now adapt and generalize the arguments in~\cite{CharikarCP05}.
Let $\bx\in X$ be an optimal solution to $\min_{x\in X}\bhpcentral{x}$, which is also an
optimal solution to $\min_{x\in X}\thpcentral{x}$. 
Let $\bopt:=\bhpcentral{\bx}$. 
Let $y^*\in[0,\tau]$ be such that $\bopt=\bhpcentral{\bx,y^*}$, and let $\by\in Y$ given
by Claim~\ref{yclose} be 
such that $\bopt-\kp'\leq\bhpcentral{\bx,\by}\leq\bopt+\kp'$.

Let $H=\frac{2\ld}{\ve}\cdot\bopt$. Call a scenario $A$ ``high'', if
$\bg(0,0,A)>H$, and ``low'' otherwise. Let $\pcent^h=\sum_{A: A\text{ is high}}\pcent_A$ 
We use $\El{.}$ (respectively $\Eh{.}$) to denote
the expectation $\E[A\sim\pcent]{.}$ 
where non-low (respectively non-high)
scenarios contribute 0 (so $\E[A\sim\pcent]{.}=\El{.}+\Eh{.}$). 
Let $\htpcent^{i,h}$, $\hEl{.}$, and $\hEh{.}$ denote these quantities for the
$i$-th SAA problem.
Since
$\bhpcentral{\bx}\geq\Eh{\bg(\bx,y^*,A)}\geq\pcent^h\bigl(H-\ld(c^\T\bx+ry^*)\bigr))$
(the second inequality is due to Claim~\ref{bndinfl}),
we have 
$\pcent^h\leq\frac{\ve}{\ld}$.%
\footnote{If $\bopt=0$, then $c^\T\bx+ry^*=0$, and $\bg(\bx,y^*,A)=0$ for all $A$ with
$\pcent_A>0$. Therefore, $\bg(0,0,A)=0=H$ for all $A$ with $\pcent_A>0$, and all scenarios
in the support of $\pcent$ are low scenarios.}
The sample size $N$ is chosen so that Chernoff bounds
ensure that with probability at least $1-\dt$, for every $i$, we have 
$\htpcent^{i,h}\leq\frac{2\ve}{\ld}$.
Hence,
\begin{alignat}{2}
\Eh{\bg(0,0,A)}-\Eh{\bg(x,y,A)} 
& \leq \ve(c^\T x+ry) \qquad && \forall x\in X, y\geq 0 \label{pfineq2} \\
\hEh{\bg(0,0,A)}-\hEh{\bg(x,y,A)} 
& \leq 2\ve(c^\T x+ry) && \forall i=1,\ldots k,\ \forall x\in X, y\geq 0. \label{pfineq3}  
\end{alignat}
Since $\bg(x,y,A)\leq\bg(0,0,A)\leq H$ for all low scenarios $A$ and all 
$(x,y)\in X\times\Rplus$, the choice of $N$ shows that, again using Chernoff bounds, with 
probability $1-\dt$, we have 
\begin{equation}
\Bigl|\hEl{\bg(x,y,A)}-\El{\bg(x,y,A)}\Bigr|\leq\ve\bopt \qquad
\forall i=1,\ldots,k,\ \forall (x,y)\in X\times Y. \label{pfineq4}
\end{equation}

Next, we argue that there is some index $t$ such that 
$\hEh[t]{\bg(0,0,A)}+\hzlong[t]$ is close to $\Eh{\bg(0,0,A)}+\zpcentlng{0}$.
For every $i$, the expected value of 
$\hEh{\bg(0,0,A)}$ 
is precisely $\Eh{\bg(0,0,A)}$, so we can use Markov's inequality. But it is more tricky
to reason about the expected value of $\hzlong[i]$ since $\zpcentlng[p]{0}$ is not linear
in $p$. 

\begin{claim} \label{zlongconc}
$\zpcentlng[p]{0}$ is a concave function of $p$.
\end{claim}

\begin{proof}
Consider any two distributions $p$ and $q$, and $\bp=\tht\cdot p+(1-\tht)\cdot q$, where
$\tht\in[0,1]$. Let $\gm^p$ and $\gm^q$ be the optimal solutions to the optimization
problems defining $\zpcentlng[p]{0}$ and $\zpcentlng[q]{0}$. Then,
$\tht\cdot\gm^p+(1-\tht)\cdot\gm^q$ is a feasible solution to the optimization problem
defining $\zpcentlng[\bp]{0}$, and its objective value is
$\tht\cdot\zpcentlng[p]{0}+(1-\tht)\cdot\zpcentlng[q]{0}$. 
\end{proof}

Using the above claim and Jensen's inequality, we obtain that the expected value of
$\hzlong[i]$ is at most $\zpcentlng{0}$. Therefore, by Markov's inequality, we have that
the event
$\hEh{\bg(\bo,0,A)}+\hzlong[i]>(1+\ve)\bigl(\Eh{\bg(\bo,0,A)}+\zpcentlng{0}\bigr)$
happens with probability at most $\frac{1}{1+\ve}\leq 1-\ve/2$. 
The probability that this happens for {\em all} $i=1,\ldots,k$ is at most
$(1-\ve/2)^k\leq\dt$. So we may assume that there is some index 
$t\in\{1,\ldots,k\}$ such that 
\begin{alignat}{1}
\hEh[t]{\bg(\bo,0,A)}&+\hzlong[t]-\Eh{\bg(\bo,0,A)}-\zpcentlng{0}
\leq\ve\Bigl(\Eh{\bg(\bo,0,A)}+\zpcentlng{0}\Bigr) \notag \\
&\leq\ve\Bigl(\Eh{\bg(\bx,y^*,A)}+\zpcentlng{0}\Bigr)+\pcent^h\ld\bigl(c^\T\bx+ry^*)
\leq \ve\bigl(\bopt+\zpcentlng{0}\bigr).
\label{pfineq5}
\end{alignat}

Now we show that the index $j$ obtained from the $f^i$ estimates can be used in place of
the index $t$. To do this, we first use the properties of the $f^i$'s and the index $j$ to
relate the quality of $\hx^j$ for the $j$-th SAA problem to the quality of $(\bx,\by)$
under any of the other SAA problems.
Let $y^j\geq 0$ be such that $\bh^j(\hx^j)=\bh^j(\hx^j,y^j)$, and let $\hy^j$ be the point
in $Y$ given by Claim~\ref{yclose}. 
We have that for every $i=1,\ldots,k$,
\begin{equation}
\tdh^j(\hx^j,\hy^j)-\kp' \leq\tdh^j(\hx^j)\leq 2\hh^j(\hx^j) 
\leq 2\beta f^j\leq 2\beta f^i 
\leq 2\beta\rho\cdot\hh^i(\bx) \leq 2\beta\rho\cdot\tdh^i(\bx,\by).
\label{pfineq1}
\end{equation}
The first inequality follows from Claim~\ref{yclose}; the second follows from
Lemma~\ref{zproxy}; the next three inequalities follow from the properties
of the $f^i$ estimates, and the choice of index $j$;
the last inequality again uses Lemma~\ref{zproxy}, and that $\tdh^i(\bx)\leq\tdh^i(\bx,y)$
for any $y\geq 0$. 

Let $\al=2\beta\rho$.
Let 
$\Dt^j = \Eh{\bg(\bo,0,A)}+\zpcentlng{0}-\hEh[j]{\bg(\bo,0,A)}-\hzlong[j]$, and
$\Dt^t = \Eh{\bg(\bo,0,A)}+\zpcentlng{0}-\hEh[t]{\bg(\bo,0,A)}-\hzlong[t]$.
Applying \eqref{pfineq1} to $j$ and $t$, we have 
$\tdh^j(\hx,\hy^j)-\kp'\leq\al\cdot\tdh^j(\bx,\by)$ and
$\tdh^j(\hx,\hy^j)-\kp'\leq\al\cdot\tdh^t(\bx,\by)$. 
Multiplying the first inequality by $\frac{1}{\al}$ and the second by $1-\frac{1}{\al}$
and adding, we get
\begin{equation}
\tdh^j(\hx,\hy^j)-\kp'\leq\tdh^j(\bx,\by)+(\al-1)\tdh^t(\bx,\by). \label{pfineq6}
\end{equation}

\medskip
We now combine these various inequalities to obtain the desired result.
By repeatedly using
\eqref{pfineq2}--\eqref{pfineq4}, we get   
\begin{alignat}{1}
\thpcentral{\hx,\hy^j} & = 
c^\T\hx+r\hy^j+\El{\bg(\hx,\hy^j,A)}+\Eh{\bg(\hx,\hy^j,A)}+\zpcentlng{0} \notag \\ 
& \leq c^\T\hx+r\hy^j+\El{\bg(\hx,\hy^j,A)}+\hEh[j]{\bg(\bo,0,A)}+\hzlong[j]+\Dt^j \notag \\ 
& \leq c^\T\hx+r\hy^j+\Bigl(\hEl[j]{\bg(\hx,\hy^j,A)}+\ve\bopt\Bigr)
+\Bigl(\hEh[j]{\bg(\hx,\hy^j,A)}+2\ve(c^\T\hx+r\hy^j)\Bigr)+\Dt^j+\hzlong[j] \notag \\
& = \Bigl(\tdh^j(\hx,\hy^j)+\Dt^j\Bigr)+\ve\bopt+2\ve(c^\T\hx+r\hy^j) \notag \\
& \leq \Bigl(\tdh^j(\bx,\by)+\Dt^j\Bigr)+(\al-1)\tdh^t(\bx,\by)+\ve\bopt+2\ve(c^\T\hx+r\hy^j)+\kp' 
\label{pfineq7}
\end{alignat}
where the last inequality above follows by applying \eqref{pfineq6}.
We bound $\tdh^j(\bx,\by)+\Dt^j$ as follows.
\begin{alignat}{1}
\tdh^j(\bx,\by)+\Dt^j & \leq
c^\T\bx+r\by+\Bigl(\El{\bg(\bx,\by,A)}+\ve\bopt\Bigr)+\Bigl(\Eh{\bg(\bx,\by,A)}+
\zpcentlng{0}+\ve(c^\T\bx+r\by)\Bigr) \notag \\
& \leq(1+\ve)\thpcentral{\bx,\by}+\ve\bopt. 
\label{pfineq8}
\end{alignat}
Similarly, we have 
\begin{equation*}
\tdh^t(\bx,\by)\leq
c^\T\bx+r\by+\Bigl(\El{\bg(\bx,\by,A)}+\ve\bopt\Bigr)+
\Bigl(\Eh{\bg(\bx,\by,A)}+\ve(c^\T\bx+r\by)+\zpcentlng{0}-\Dt^t\Bigr).
\end{equation*}
Substituting $-\Dt^t\leq \ve\bigl(\bopt+\zpcentlng{0}\bigr)$ from \eqref{pfineq5},
we can simplify this to
$$
\tdh^t(\bx,\by)\leq(1+\ve)\thpcentral{\bx,\by}+2\ve\bopt. 
$$
Finally, substituting this bound and \eqref{pfineq8}, in \eqref{pfineq7}, we obtain 
\begin{equation*}
\thpcentral{\hx,\hy^j}\leq 
\Bigl((1+\ve)\thpcentral{\bx,\by}+\ve\bopt\Bigr) 
+(\al-1)\Bigl((1+\ve)\thpcentral{\bx,\by}+2\ve\bopt\Bigr) 
+\ve\bopt+2\ve(c^\T\hx+r\hy^j)+\kp'
\end{equation*}
This implies that 
\begin{equation*}
\begin{split}
\thpcentral{\hx}\leq\thpcentral{\hx,\hy^j}&\leq
\frac{\al(1+\ve)}{1-2\ve}\cdot\thpcentral{\bx,\by}+O(\al\ve)\cdot\bopt+\frac{\kp'}{1-2\ve} \\
& \leq\al(1+4\ve)\cdot\thpcentral{\bx,\by}+O(\al\ve)\cdot\bopt+(1+4\ve)\kp' \\ 
& \leq\al(1+4\ve)\cdot\thpcentral{\bx,y^*}+(\al+1)(1+4\ve)\kp'+O(\al\ve)\cdot\bopt
\end{split}
\end{equation*}
where $\frac{1+\ve}{1-2\ve}\leq 1+4\ve$ since $\ve\leq\frac{1}{3}$.
This proves \eqref{goalineq}. Combining this with Lemma~\ref{zproxy} yields the inequality
in Theorem~\ref{mainsaathm}.
The success probability is the probability that inequalities
\eqref{pfineq3}--\eqref{pfineq5} hold, which is at least $1-3\dt$. 
\hfill\qed

\end{document}